\documentclass[11pt]{article}
\usepackage{amsmath, mathtools,amssymb,bbm}
\makeatletter
\renewcommand*\env@matrix[1][*\c@MaxMatrixCols c]{%
  \hskip -\arraycolsep
  \let\@ifnextchar\new@ifnextchar
  \array{#1}}
\DeclareMathOperator*{\argmax}{arg\,max}

\makeatother
\usepackage{mathdots}
\usepackage{mdframed}
\usepackage[margin=1in]{geometry}
\usepackage{thm-restate}
\usepackage{thmtools}
\usepackage{changepage}
\usepackage[numbers,comma,sort&compress]{natbib}
\usepackage{subcaption}
\usepackage{authblk}
\usepackage[english]{babel}
\usepackage[utf8]{inputenc}
\usepackage{tcolorbox}
\usepackage{amsthm}
\theoremstyle{definition}
\usepackage{caption}
\usepackage{cases}
\usepackage{algorithmicx}
\usepackage{algorithm}
\usepackage[noend]{algpseudocode}
\usepackage{xcolor,soul}
\usepackage{qcircuit}
\usepackage[colorinlistoftodos, color=green!40, prependcaption]{todonotes}
\usepackage{blindtext} 
\usepackage[pdftex, pdftitle={Article}, pdfauthor={Author}, linktocpage]{hyperref} 
\hypersetup{
    colorlinks=true,
    citecolor=magenta,
    linkcolor=blue,
    filecolor=magenta,
    urlcolor=red,
}

\newcommand{\shortminus}{\text{-}}
\newcommand{\shortplus}{\text{+}}
\newcommand{\compositeone}{\hyperref[eq:composite1]{Composite 1}}
\newcommand{\compositetwo}{\hyperref[eq:composite2]{Composite 2}}
\newcommand{\compositethree}{\hyperref[eq:composite3]{Composite 3}}
\newcommand{\compositefour}{\hyperref[eq:composite4]{Composite 4}}

\newtheorem*{remark}{Remark}
\newtheorem{lemma}{Lemma}
\newtheorem{definition}{Definition}
\newtheorem{proposition}{Proposition}
\newtheorem{theorem}{Theorem}
\newtheorem{corollary}{Corollary}
\newtheorem*{conjecture}{Conjecture}
\mdfdefinestyle{MyFrame}{%
    linecolor=black,
    outerlinewidth=2pt,
    innertopmargin=-5pt,
    innerbottommargin=4pt,
    innerrightmargin=5pt,
    innerleftmargin=-50pt,
        leftmargin = 4pt,
        rightmargin = 4pt
        }

\allowdisplaybreaks

\usetikzlibrary{patterns}
\usetikzlibrary{decorations.pathmorphing}
\usetikzlibrary{decorations.pathreplacing}
\usepackage{tikz-feynman}
\usetikzlibrary{cd}

\usetikzlibrary{matrix}
\tikzset{
    mymatrix/.style={
        matrix of math nodes,
        inner ysep= 0pt,
        nodes ={inner xsep=0pt,
            inner ysep=1pt},
        ampersand replacement=\&, 
        }
    }

\newcommand{\thicklen}{2.5pt}
\newcommand{\mym}[1]{%
    \begin{tikzpicture}[line width=.7pt, baseline={([yshift=-.55ex]current bounding box.center)}]
    \matrix[mymatrix](a){#1\\};
    \draw ([xshift=\thicklen]a.north west)  -- ++(-\thicklen,0) -- (a.south west) -- ++(\thicklen,0);
    \draw ([xshift=-\thicklen]a.north east)  -- ++(\thicklen,0) -- (a.south east) -- ++(-\thicklen,0);
    \end{tikzpicture}}

\newcommand{\entropy}[1]{S\left(#1\right)}
\newcommand{\Tr}{\text{Tr}}
\newcommand{\nil}{\text{nil}}
\newcommand{\maxmerge}{\Join}
\newcommand{\rightmerge}{\lhd}
\newcommand{\consistent}{\stackrel{c}{=}}

\newcommand{\horizontalbaby}[2]
{
\mym{
\centeredTikZ{
    \node[] (first) at (0.25, -1) {\footnotesize $#1$};
    \node[] (last) at (1.75, 1.25) {\footnotesize $#2$};
    \emptysquare{0}{0};
    \draw[dotted, line width=0.375mm] (0.75,0.25) -- (1.25, 0.25);
    \emptysquare{1.5}{0};
    \draw[->] (first) -- (0.25, 0);
    \draw[->] (last) -- (1.75, 0.5);
    }
}
}

\newcommand{\horizontalunit}[1]
{
\mym{
    \centeredTikZ{
    \emptysquare{0}{0};
    \emptysquare{0.5}{0};
    }_{#1}
}
}

\newcommand{\twobytwounit}[1]
{
\mym{
    \centeredTikZ{
    \emptysquare{0}{0};
    \emptysquare{0.5}{0};
    \emptysquare{-0.25}{0.5};
    \emptysquare{0.25}{0.5};
    }_{#1}
}
}

\newcommand{\adulthorizontalsnake}[4]
{
\mym{
\centeredTikZ{
\node[] (first) at (0.25, -1) {\footnotesize $(#1,#2)$};
\node[] (last) at (1.5, 1.75) {\footnotesize $(#3,#4)$};
    \emptysquare{0}{0};
    \emptysquare{1.5}{0};
    \emptysquare{-0.25}{0.5};
    \emptysquare{1.25}{0.5};
    \draw[dotted, line width=0.375mm] (0.625,0.5) -- (1.125, 0.5);
    \draw[->] (first) -- (0.25, 0);
    \draw[->] (last) -- (1.5, 1);
    }
}
}

\newcommand{\thicksnake}[4]
{
\mym{
\centeredTikZ{
\node[left] (first) at (-0.5, 0.25) {\footnotesize $(#1, #2)$};
\node[right] (last) at (1.75, 1.75) {\footnotesize $(#3, #4)$};
    \emptysquare{0}{0};
    \emptysquare{-0.75}{1.5};
    \draw[dotted, line width=0.375mm] (0, 0.75) -- (-0.25, 1.25);
    \draw[dotted, line width=0.375mm] (0.75,0.25) -- (1.25, 0.25);
    \emptysquare{1.5}{0};
    \emptysquare{0.75}{1.5};
    \draw[dotted, line width=0.375mm] (0,1.75) -- (0.5, 1.75);
    \draw[dotted, line width=0.375mm] (1.5, 0.75) -- (1.25, 1.25);
    \draw[->] (first) -- (0, 0.25);
    \draw[->] (last) -- (1.25, 1.75);

    }
}
}

\newcommand{\miniblockfull}{
\begin{tikzpicture}[scale=0.2, xscale=-1, baseline={([yshift=-.55ex]current bounding box.center)}]\foreach \x in {0, 1}
    {
    \foreach \y in {0, 1}
    {
    \emptysquare{\x*0.5+\y*0.25}{\y*0.5};
    }
    }\end{tikzpicture}
}

\newcommand{\mymarginal}{\centeredTikZmini{\squaresempty{3}{3}{0}{0};}}

\newcommand{\dottedsquare}[2]{\draw[pattern color=red, pattern=dots] (#1, #2) rectangle (#1 +0.5, #2 + 0.5)}
\newcommand{\emptysquare}[2]{\draw[fill=white] (#1, #2) rectangle (#1 +0.5, #2 + 0.5)}

\newcommand{\filledsquare}[2]{\draw[fill=black!25!white, draw=black] (#1, #2) rectangle (#1 +0.5, #2 + 0.5)}

\newcommand{\squaresdotted}[4]{
\foreach \x in {1,...,#1}
{
\foreach \y in {1,...,#2}
{
\dottedsquare{#3*0.5 - #4*0.25 +0.5*\x - 0.25*\y}{#4*0.5+0.5*\y};
}
}
}

\newcommand{\squaresempty}[4]{
\foreach \x in {1,...,#1}
{
\foreach \y in {1,...,#2}
{
\emptysquare{#3*0.5 - #4*0.25 +0.5*\x - 0.25*\y}{#4*0.5+0.5*\y};
}
}
}

\newcommand{\squaresfilled}[4]{
\foreach \x in {1,...,#1}
{
\foreach \y in {1,...,#2}
{
\filledsquare{#3*0.5 - #4*0.25 +0.5*\x - 0.25*\y}{#4*0.5+0.5*\y};
}
}
}

\newcommand{\centeredTikZ}[1]{
\begin{tikzpicture}[scale=0.55, baseline={([yshift=-.55ex]current bounding box.center)}]
#1
\end{tikzpicture}
}

\newcommand{\centeredTikZmini}[1]{
\begin{tikzpicture}[scale=0.25, baseline={([yshift=-.55ex]current bounding box.center)}]
#1
\end{tikzpicture}
}

\begin{document}
\title{Entropy scaling law and the quantum marginal problem}

\author{\normalsize Isaac H. Kim\thanks{3 Centre for Engineered Quantum Systems, School of Physics, University of Sydney, Sydney, NSW 2006, Australia}}

\date{\today} 
\maketitle
\begin{abstract}
Quantum many-body states that frequently appear in physics often obey an entropy scaling law, meaning that an entanglement entropy of a subsystem can be expressed as a sum of terms that scale linearly with its volume and area, plus a correction term that is independent of its size. We conjecture that these states have an efficient dual description in terms of a set of marginal density matrices on bounded regions, obeying the same entropy scaling law locally. We prove a restricted version of this conjecture for translationally invariant systems in two spatial dimensions. Specifically, we prove that a translationally invariant marginal obeying three non-linear constraints --- all of which follow from the entropy scaling law straightforwardly --- must be consistent with some global state on an infinite lattice.  Moreover, we derive a closed-form expression for the maximum entropy density compatible with those marginals, deriving a variational upper bound on the thermodynamic free energy. Our construction's main assumptions are satisfied exactly by solvable models of topological order and approximately by finite-temperature Gibbs states of certain quantum spin Hamiltonians.
\end{abstract}                             

 \vspace*{\fill}
 \textit{I can illustrate the second approach with the same image of a nut to be opened. The first analogy that came to my mind is of immersing the nut in some softening liquid, and why not simply water? From time to time you rub so the liquid penetrates better, and otherwise you let time pass. The shell becomes more flexible through weeks and months—when the time is ripe, hand pressure is enough, the shell opens like a perfectly ripened avocado! A different image came to me a few weeks ago. The unknown thing to be known appeared to me as some stretch of earth or hard marl, resisting penetration ... the sea advances insensibly in silence, nothing seems to happen, nothing moves, the water is so far off you hardly hear it ... yet it finally surrounds the resistant substance.} 
 [Grothendieck 1985–1987, pp. 552-3]

\newpage
\tableofcontents

\newpage

\section{Introduction}
\label{sec:intro}
The discovery that there are materials exhibiting superconductivity at a temperature significantly higher than what the BCS theory of superconductivity~\cite{Cooper1957} predicts has been one of the major driving forces behind the study of strongly correlated electron systems over the past few decades~\cite{Bednorz1986}. The struggle to definitively identify the underlying microscopic physical mechanism has been a major theme of research in strongly correlated material over this period.

While writing down a toy model of such systems is often straightforward, solving them is far from trivial. Kitaev showed that the problem of estimating the ground state energy of such models is QMA-complete, which is likely to be hard even for a quantum computer~\cite{Kitaev2002}. To make progress, one must resort to a set of physical assumptions that can alleviate this otherwise exorbitant computational cost.

A ``model citizen'' of such physical assumption is the spectral gap. Hastings showed that one-dimensional quantum many-body systems with a constant spectral gap obey area law, in which case the ground state wavefunction can be efficiently described by a tensor network called matrix product state~\cite{Hastings2007}. Moreover, there is a classical algorithm that, given a one-dimensional gapped Hamiltonian, finds the ground state wavefunction in polynomial time~\cite{Landau2013}. These developments firmly establish that one-dimensional gapped systems form an easy subclass of Hamiltonians that are amenable to classical computational methods.

However, attempts to extend these insights to higher dimensions have been stymied by many obstacles. Despite a recent work on establishing the subvolume law of entanglement entropy in two spatial dimensions~\cite{Anshu2020}, a rigorous proof of the area law in higher dimensions remains open. Moreover, even if we can prove area law in the future, that alone cannot be the end of the story. Ge and Eisert showed that a class of states in two-dimensions that satisfy the area law for all Renyi entropies \emph{cannot} be characterized by tensor networks in higher dimensions using a polynomial number of parameters~\cite{Ge2014}. These authors further conclude that, even amongst the states that obey area law, the states that have an efficient classical description constitutes an even smaller corner of that already restrictive set. 

If so, what kind of physical states can we describe efficiently? In this paper, we will take a top-down approach to tackle this problem. By postulating a plausible property of gapped systems, we will be able to \emph{derive} an efficient tensor network that can accurately approximate their ground states. (In fact, the property we demand can be satisfied, albeit approximately, even by the finite-temperature Gibbs state of some locally interacting many-body Hamiltonian.)

The derived tensor network, which we briefly mention in Appendix~\ref{appendix:tensor_network}, have two useful properties. First, the network can be contracted efficiently, in time that scales linearly with the system size. These tensor networks are in stark contrast with (more) general tensor networks, which are computationally hard to contract~\cite{Schuch2007}.\footnote{Our work has a tenuous resemblance with Ref.~\cite{Zaletel2020}, but how these two constructions are related to each other is not entirely clear at the moment.} Second, the global entropy of the tensor network can be computed efficiently. This is a highly nontrivial feature that not many tensor networks can bolster; see however, Ref.~\cite{Kim2017} for a notable exception.

While the existence of such a tensor network is undoubtedly surprising, why it has those properties is not apparent from its definition. A better perspective is to interpret the tensor network as a maximum-entropy state consistent with a set of marginal density matrices, \emph{marginals} for short. From this point of view, the fundamental data that characterizes the global state is the set of marginals on balls of bounded radius that overlap with each other; see Fig.~\ref{fig:rdms}.
\begin{figure}
  \centering
  \includegraphics[width=0.3\columnwidth]{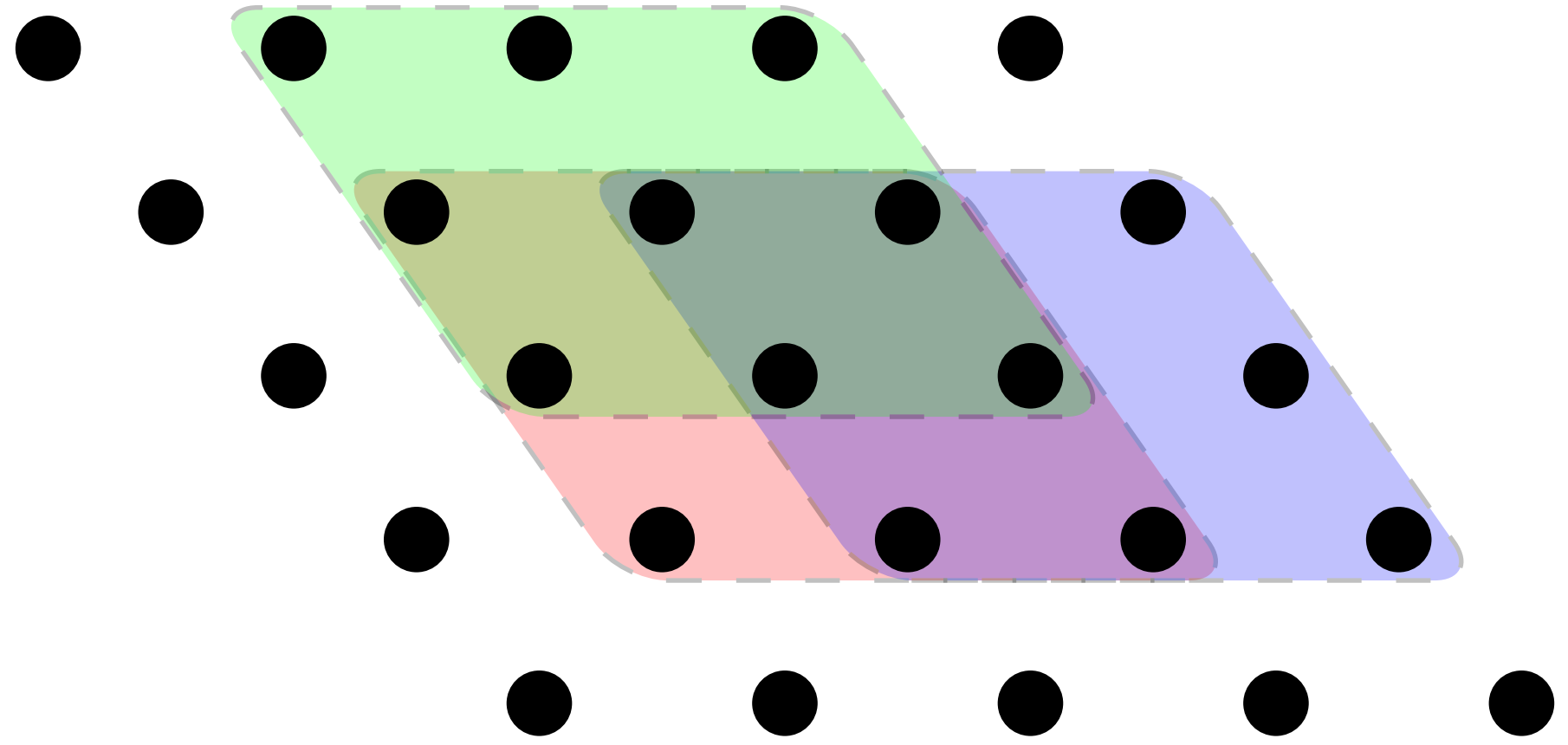}
    \caption{Different supports of the density matrices, enclosed in three different dashed regions, colored in red, blue, and green in a counterclockwise direction. (The overlaps of the supports are colored with their respective color mixtures.) Each subsystem contains $3\times 3=9$ local degrees of freedom, denoted by the black dots. The supports of the density matrices may overlap with each other, as shown in this example.}
    \label{fig:rdms}
\end{figure}

Of course, there is a well-known difficulty in such characterization. Given a set of marginals, one may not be able to efficiently decide if those marginals are consistent with some global state. A completely general version of this problem is known to be QMA-complete~\cite{Liu2006,Broadbent2019}. Therefore, it will be too ambitious to find a complete solution to this problem.

This is where our main result comes in. We provide a physically motivated sufficient condition on the set of marginals such that, if the condition is satisfied, one can guarantee the existence of a global state consistent with those marginals. Moreover, under the same condition, we can derive an \emph{exact} expression for the maximum global entropy consistent with the given marginals. Therefore, not only can we compute an upper bound to the ground state energy, but we can also compute an upper bound to the thermodynamic free energy.


Our conditions are satisfied if the underlying state, up to a finite-depth quantum circuit, obeys the following scaling law for entanglement entropy:
\begin{equation}
    S(A) = \alpha_0 |A|  +\alpha_1 |\partial A| - \gamma, \label{eq:ent_scaling}
\end{equation}
for any disk-shaped region $A$, where the first two terms are proportional to the volume and the boundary of $A$ respectively, quantified in terms of the number of degrees of freedom in the interior/boundary of $A$, and the last term is a constant that is independent of the size of $A$. (See Fig.~\ref{fig:entropy_formula}.) Here $\alpha_0$ and $\alpha_1$ are non-universal constants and $\gamma$ is a universal term that is independent of these microscopic details. Modulo the correction term that vanishes in the $|A| \to \infty$ limit, Eq.~\eqref{eq:ent_scaling} is a commonly used form of the entanglement entropy that is expected to hold in physical systems with a finite correlation length, both at zero~\cite{Kitaev2006,Levin2006} and finite temperature\cite{Castelnovo2007}; see Ref.~\cite{Eisert2010} for a review. In so far as Eq.~\eqref{eq:ent_scaling} is valid, our result is exact. Since Eq.~\eqref{eq:ent_scaling} has been observed to hold in a large class of systems up to a correction that decays exponentially in $|\partial A|$, we expect the approximation error to decay rapidly as we increase the size of $A$. Provided that our speculation is correct, by minimizing the free energy within the space of such marginals, we will be able to obtain a variational upper bound on the free energy which well-approximates the true thermodynamic free energy. 

\begin{figure}[h]
  \centering
  \includegraphics[width=0.3\columnwidth]{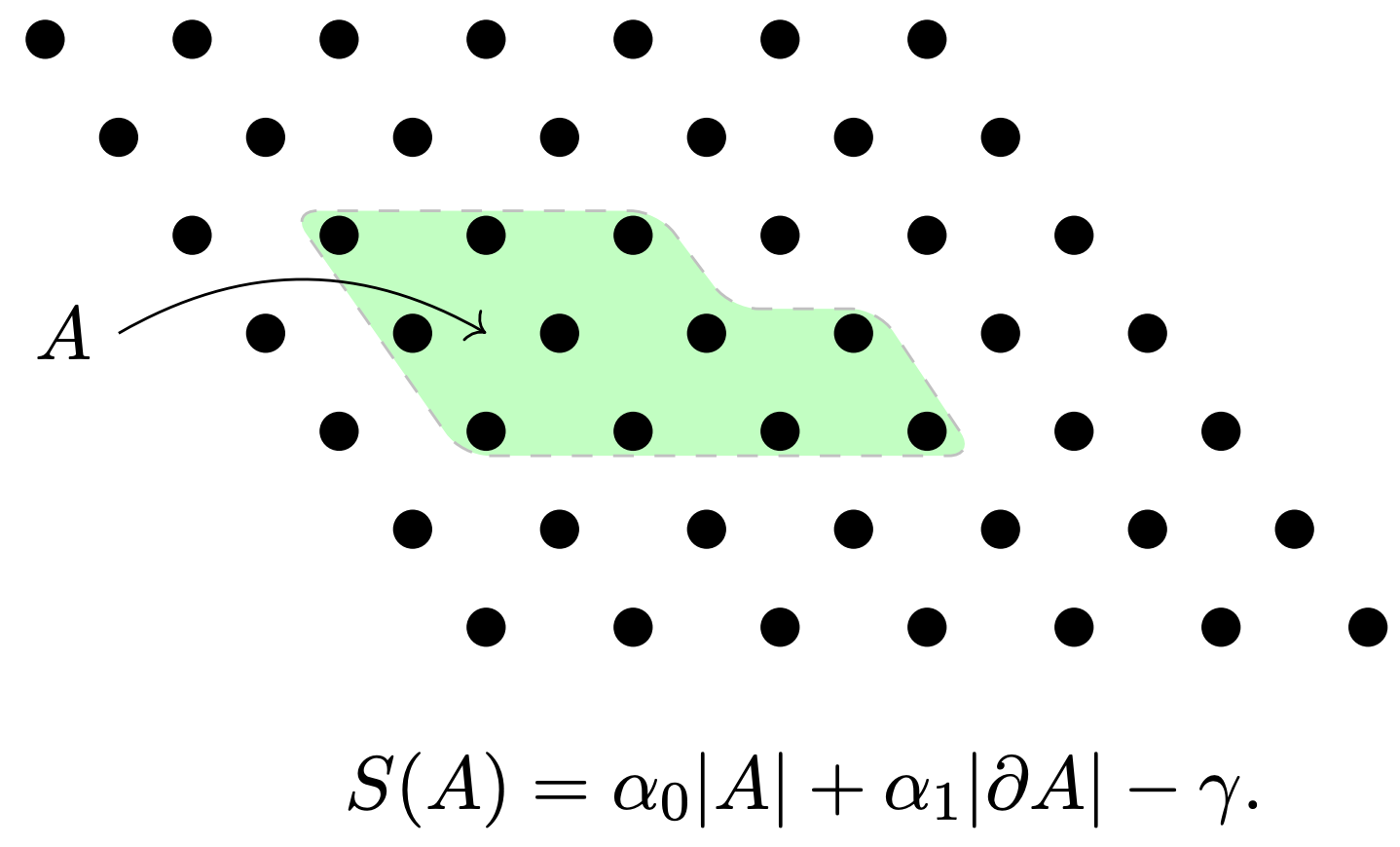}
    \caption{Entanglement entropy of a simply connected subsystem.}
    \label{fig:entropy_formula}
\end{figure}

With Eq.~\eqref{eq:ent_scaling}, we can describe our main result somewhat more precisely, however, still glossing over the mathematical details. Roughly speaking, we prove the following correspondence:
\begin{equation}
\boxed{
\begin{tikzcd}
    \shortstack{\text{Translationally invariant marginal}\\ \text{obeying Eq.~\eqref{eq:ent_scaling} locally.}}  \arrow[d, Leftrightarrow] \\
    \shortstack{\text{Translationally invariant state}\\  \text{obeying Eq.~\eqref{eq:ent_scaling} globally.}}
\end{tikzcd}
}
\label{eq:main_rough}
\end{equation}
Specifically, there is a one-to-one map between a translationally invariant marginal on a \emph{finite} cluster obeying Eq.~\eqref{eq:ent_scaling} within that marginal and a global state consistent with those marginals obeying Eq.~\eqref{eq:ent_scaling} at a larger scale. This means that a translationally invariant marginal on a bounded region obeying Eq.~\eqref{eq:ent_scaling} \emph{defines} a translationally invariant state of an infinite system. Moreover, the maximum entropy (density) of the the latter can be computed exactly using Eq.~\eqref{eq:ent_scaling}.

To reach this conclusion, we make a simple but nontrivial observation: that the states that satisfy Eq.~\eqref{eq:ent_scaling} obey the following identity 
\begin{equation}
    S(AB) + S(BC) - S(B) - S(ABC) = 0.\label{eq:ssa_equality}
\end{equation}
for a judiciously chosen set of subsystems; see Fig.~\ref{fig:ssa_saturation} for examples. The states that satisfy Eq.~\eqref{eq:ssa_equality} have a very special structure known as the \emph{quantum Markov chain structure}~\cite{Petz1988,Petz2003}. The derivation of our result is an ``elementary'' consequence of this fact in the sense that every argument is  based only on the properties of the quantum Markov chain.
\begin{figure}[h]
  \centering
  \includegraphics[width=0.3\columnwidth]{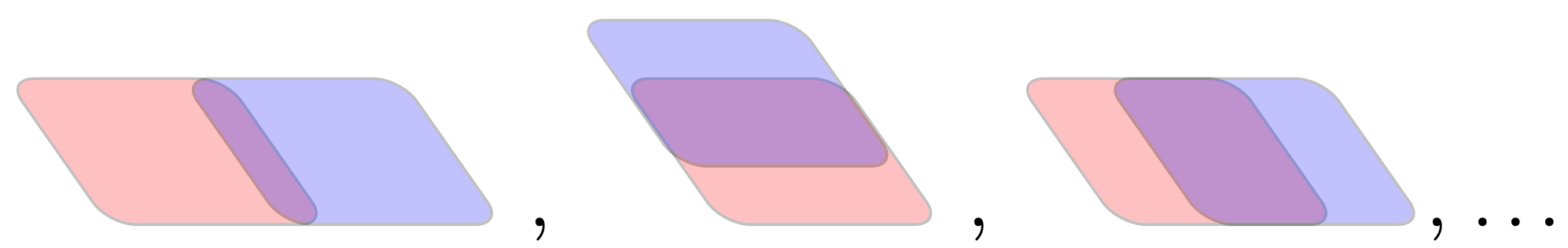}
    \caption{Examples of subsystems for which the relation $S(AB) + S(BC) - S(B) - S(ABC)=0$ holds, assuming Eq.~\eqref{eq:ent_scaling}. The color convention is red for $A$, purple for $B$, and blue for $C$.}
    \label{fig:ssa_saturation}
\end{figure}

However, the proof of Eq.~\eqref{eq:main_rough} is rather long and technical. To alleviate this difficulty, we developed a theory of \emph{merging algebra}. This is an algebra equipped with two types of binary operations over a set of marginals, with nontrivial relations inherited from Eq.~\eqref{eq:ent_scaling}. Interestingly, these operations are generally non-associative, and one of them is not even commutative. Instead, the commutation/association relations hold only in certain specific contexts. Despite this complication, these rules are still structured enough to let us prove our central claims. 

While the merging algebra is very different from what we ordinarily encounter in physics, there is clearly some structure hidden beneath all this new mathematics. That a generic-looking equation like Eq.~\eqref{eq:ent_scaling} leads to such a novel mathematical structure is surprising. But what is even more surprising is that this very structure can lead to a physically well-motivated condition to ensure the existence of a global state consistent with some marginals. This is known as the quantum marginal problem, which has been notoriously difficult to make progress on; see Ref.~\cite{Schilling2015} and the references therein.

One may wonder: is it really essential to invoke a new mathematical structure? Is there no simple way to derive the same result more concisely? Unfortunately, we do not have a satisfactory answer to these questions. However, let us point out that our construction becomes significantly simpler in the classical setting~\cite{Kim2020a}, giving hope that perhaps a more straightforward proof will be found in the future. For now, the main reason why the arguments in Ref.~\cite{Kim2020a} do not directly generalize to the quantum setting is that in quantum mechanics, reduced density matrices generally do not commute with each other. For that reason, many of the operations we consider are intrinsically non-commutative. In fact, they are generally not even associative. Further progress in this direction seems to require a deeper understanding of the structure of multipartite quantum states that satisfy Eq.~\eqref{eq:ssa_equality} on multiple subsystems.

An important message from our work is that one should take Eq.~\eqref{eq:ent_scaling} seriously. As we have already discussed, Eq.~\eqref{eq:ent_scaling} leads to a formulation of efficiently contractable tensor networks, physically-motivated solutions to the quantum marginal problem, and even a new kind of algebra whose exact nature is still a mystery. Moreover, it was shown in Ref.~\cite{SKK2019,Shi2020,Shi2020a} that Eq.~\eqref{eq:ent_scaling} leads to a highly constrained set of consistency equations that govern the basic ``laws'' of the low-energy excitations in such systems. These results strongly indicate that there are rich physics and mathematics hidden beneath Eq.~\eqref{eq:ent_scaling} that warrants a further exploration. In particular, we would like to conjecture the following generalization of Eq.~\eqref{eq:main_rough}:
\begin{equation}
\boxed{
\begin{tikzcd}
    \shortstack{\text{Locally consistent marginals obeying} \\ \text{the identities in Fig.~\ref{fig:ssa_saturation} internally.}}  \arrow[d, Leftrightarrow, "?"] \\
    \shortstack{\text{Global quantum state obeying} \\ \text{the identities in Fig.~\ref{fig:ssa_saturation} everywhere.}}
\end{tikzcd}
}
\label{eq:conjecture_rough}
\end{equation}
That is, one may be able to remove the translational invariance condition entirely. Proving or disproving Eq.~\eqref{eq:conjecture_rough} is an important open problem that is left for future work.

The rest of this paper is structured as follows. In Section~\ref{sec:summary}, we provide a more in-depth summary of our main results. In Section~\ref{sec:marginals}, we introduce the fundamental objects behind our work, namely the marginals over bounded regions. We will summarize the constraints imposed on those marginals and briefly sketch their implications. In Section~\ref{sec:prelim}, we provide a short review on the properties of the von Neumann entropy. In Section~\ref{sec:merging_algebra}, we introduce and set out the basic rules of the merging algebra. Moreover, we introduce an object called \emph{snake} and study its properties. In Section~\ref{sec:snakes_in_action}, we will summarize the key statements that directly lead to the main results of this paper. This section will only provide an overview of these statements, deferring the proofs to Appendix~\ref{appendix:2to1},~\ref{appendix:1to2}, and~\ref{appendix:twist}. In Section~\ref{sec:main_results}, we derive the main results from the statements in Section~\ref{sec:snakes_in_action}. We end with a discussion in Section~\ref{sec:discussion}.

\section{Solving quantum many-body physics locally}
\label{sec:summary}
In this section, we will provide a more explicit explanation of our main result. The class of systems we are considering is depicted on the right hand side of Fig.~\ref{fig:lattice}, which is a (triangular) lattice of quantum spins which are locally interacting with each other. Without loss of generality, with a sufficient number of coarse-graining steps, we can always assume that the interaction terms of the Hamiltonian can be localized to a $2\times 2$ \emph{cluster} of the coarse-grained lattice.
\begin{figure}[h]
    \centering
\includegraphics[width=0.33\columnwidth]{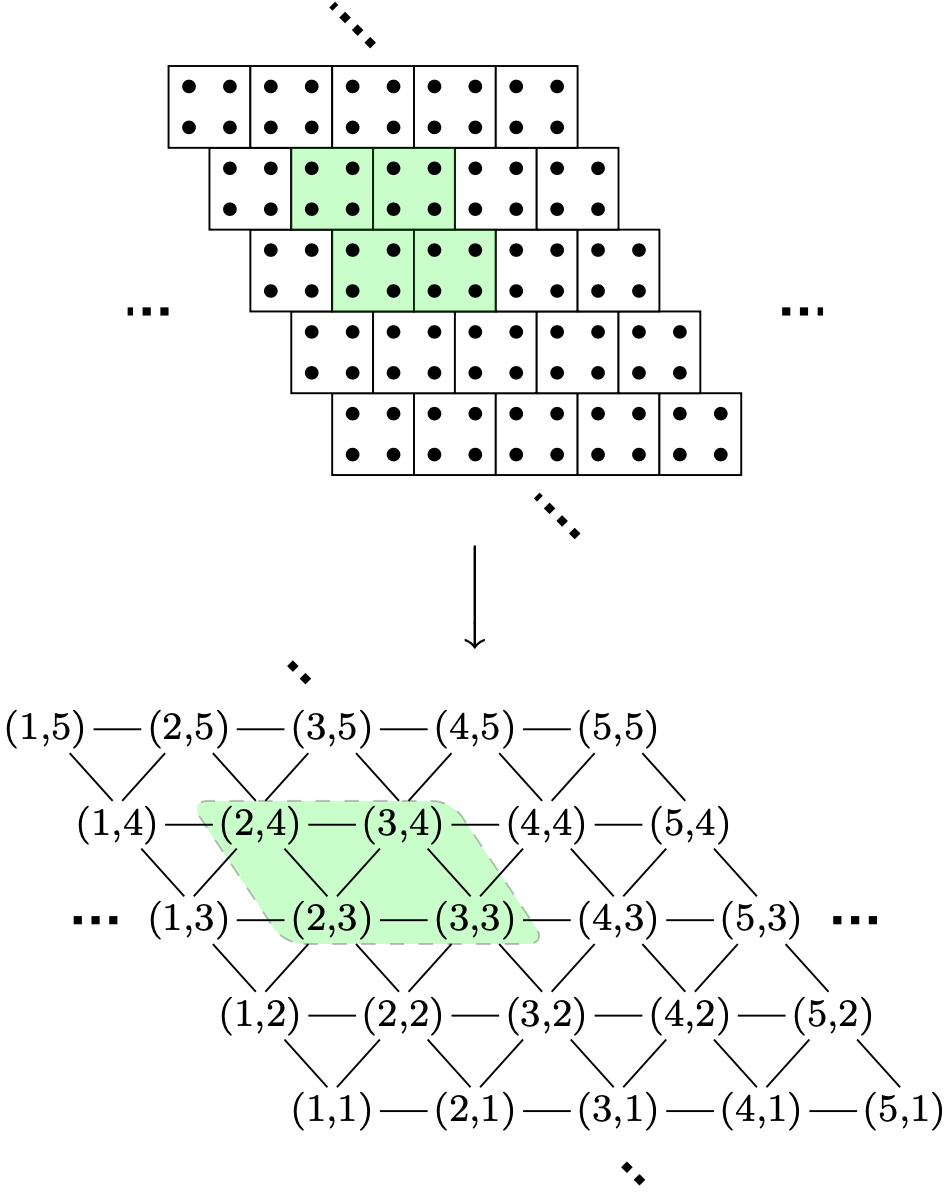}
    \caption{Upon coarse-graining the elementary degrees of freedom (dots on the top figure), we obtain a quantum spin system with larger local Hilbert space dimension, on a triangular lattice. Each of the coarse-grained degrees of freedom are referred to as \emph{clusters}. After coarse-graining enough times, the local terms of the Hamiltonian can be localized to a $2\times 2$ cluster(green). The pair of integers represent the $x$- and $y$-coordinate of each lattice site on the coarse-grained lattice.}
    \label{fig:lattice}
\end{figure}

Given a Hamiltonian that acts locally on this lattice, one may wish to study its properties at both zero and non-zero temperature, by minimizing the energy and the free energy respectively. Unfortunately, this is often a difficult problem because a direct minimization in the entire Hilbert space incurs an exponential cost.

One may hope to reduce the exorbiant cost by considering a smaller subset of physically relevant states and variationally minimizing the energy or the free energy within the family of such states. In this paper, we will take a top-down approach to identify such a family. Instead of coming up with an ansatz and attempting to justify their physical relevance, we will begin with a physical principle and explore its consequences seriously. Our starting point is Eq.~\eqref{eq:ent_scaling}, which for the reader's convenience is restated below:
\begin{equation}
    S(A) = \alpha_0 |A| + \alpha_1 |\partial A| - \gamma. \tag{\ref{eq:ent_scaling}}
\end{equation}

For a class of states that satisfy Eq.~\eqref{eq:ent_scaling}, there is a universal fact that is independent of $\alpha_0, \alpha_1,$ and $\gamma$. This is the fact that the the first two terms of Eq.~\eqref{eq:ent_scaling} obey the inclusion-exclusion principle. Therefore, with a judicious choice of subsystems, certain linear combinations of entanglement entropies may vanish. Of particular relevance to us are the following three identities. Consider a $2\times 2$ and a $3\times 3$ cluster, denoted as \miniblockfull and \mymarginal respectively. The marginals over these clusters should satisfy
\begin{equation}
\begin{aligned}
\entropy{\centeredTikZ{\squaresfilled{2}{2}{0}{0};\squaresempty{2}{1}{0}{0};\squaresempty{1}{1}{0}{1};}} &= 
\entropy{\centeredTikZ{\squaresfilled{2}{2}{0}{0};\squaresempty{2}{1}{0}{0};\squaresfilled{1}{1}{0}{1};}}
+
\entropy{\centeredTikZ{\squaresfilled{2}{2}{0}{0};\squaresempty{1}{2}{0}{0};\squaresfilled{1}{1}{1}{0};}}
-
\entropy{\centeredTikZ{\squaresfilled{2}{2}{0}{0};\squaresfilled{2}{1}{0}{0};\squaresfilled{1}{1}{0}{1};\squaresempty{1}{1}{0}{0};}},\\[8pt]
\entropy{\centeredTikZ{\squaresfilled{2}{2}{0}{0};\squaresempty{2}{1}{0}{1};\squaresempty{1}{1}{1}{0};}} &= 
\entropy{\centeredTikZ{\squaresfilled{2}{2}{0}{0};\squaresempty{2}{1}{0}{1};\squaresfilled{1}{1}{1}{0};}}
+
\entropy{\centeredTikZ{\squaresfilled{2}{2}{0}{0};\squaresempty{1}{2}{1}{0};\squaresfilled{1}{1}{0}{1};}}
-
\entropy{\centeredTikZ{\squaresfilled{2}{2}{0}{0};\squaresfilled{2}{1}{0}{1};\squaresfilled{1}{1}{1}{0};\squaresempty{1}{1}{1}{1};}},
\end{aligned}
\label{eq:ent_assumption1}
\end{equation}
and
\begin{equation}
\begin{aligned}
    \entropy{\centeredTikZ{\squaresempty{3}{3}{0}{0};}}
    =
    &\entropy{\centeredTikZ{\squaresfilled{3}{3}{0}{0};\squaresempty{2}{2}{0}{0};}}
    +
    \entropy{\centeredTikZ{\squaresfilled{3}{3}{0}{0};\squaresempty{2}{2}{1}{0};}} \\
    &+
    \entropy{\centeredTikZ{\squaresfilled{3}{3}{0}{0};\squaresempty{2}{2}{0}{1};}}
    +
    \entropy{\centeredTikZ{\squaresfilled{3}{3}{0}{0};\squaresempty{2}{2}{1}{1};}} \\
    &-
    \entropy{\centeredTikZ{\squaresfilled{3}{3}{0}{0};\squaresempty{2}{1}{0}{1};}}
    -
    \entropy{\centeredTikZ{\squaresfilled{3}{3}{0}{0};\squaresempty{2}{1}{1}{1};}}\\
    &-
    \entropy{\centeredTikZ{\squaresfilled{3}{3}{0}{0};\squaresempty{1}{2}{1}{0};}}
    -
    \entropy{\centeredTikZ{\squaresfilled{3}{3}{0}{0};\squaresempty{1}{2}{1}{1};}}
     \\
     &+ 
    \entropy{\centeredTikZ{\squaresfilled{3}{3}{0}{0};\squaresempty{1}{1}{1}{1};}},
\end{aligned}
\label{eq:ent_assumption2}
\end{equation}
where $S(\cdot)$ represent the entanglement entropy of the empty clusters in the paranthesis over some (fixed) global state that satisfies Eq.~\eqref{eq:ent_scaling}. These are precisely the clusters appearing on the left-hand side of Fig.~\ref{fig:lattice}, representing coarse-grained degrees of freedom. The gray clusters mean that we have applied a partial trace on those clusters. For instance, the first term on the right hand side of Eq.~\eqref{eq:ent_assumption2} is the entropy of the reduced density matrix of a $2\times 2$ cluster, located at the bottom left corner of the $3\times 3$ cluster.

Let us briefly comment on what the ``volume'' ($|A|$) and the ``area'' ($|\partial A|$) terms mean. There are two heuristic ways to calculate them, both leading to Eq.~\eqref{eq:ent_assumption1} and~\eqref{eq:ent_assumption2}. The first approach is to count the number of clusters. Specifically, for the volume term, simply count the number of clusters; for the area term, simply count the number of clusters that are placed at the boundary. It is a straightforward exercise to verify that both Eq.~\eqref{eq:ent_assumption1} and~\eqref{eq:ent_assumption2} follow from this simple rule.

Alternatively, one may view each of the cluster as collection of smaller clusters. For instance, imagine replacing each cluster by a $2\times 2$ cluster.
\begin{equation*}
\begin{aligned}
\centeredTikZ{\squaresempty{2}{2}{0}{0};}   &\longrightarrow
\centeredTikZ{\squaresempty{4}{4}{0}{0};}, \\
\centeredTikZ{\squaresempty{3}{3}{0}{0};}   &\longrightarrow
\centeredTikZ{\squaresempty{6}{6}{0}{0};}.
\end{aligned}
\end{equation*}
Now Eq.~\eqref{eq:ent_assumption1} and~\eqref{eq:ent_assumption2} reads:
\begin{equation*}
\begin{aligned}
\entropy{\centeredTikZ{\squaresfilled{4}{4}{0}{0};\squaresempty{4}{2}{0}{0};\squaresempty{2}{2}{0}{2};}}&=
\entropy{\centeredTikZ{\squaresfilled{4}{4}{0}{0};\squaresempty{4}{2}{0}{0};}}
+
\entropy{\centeredTikZ{\squaresfilled{4}{4}{0}{0};\squaresempty{2}{4}{0}{0};}}\\
&-
\entropy{\centeredTikZ{\squaresfilled{4}{4}{0}{0};\squaresempty{2}{2}{0}{0};}},\\
\entropy{\centeredTikZ{\squaresfilled{4}{4}{0}{0};\squaresempty{4}{2}{0}{2};\squaresempty{2}{2}{2}{0};}}&=
\entropy{\centeredTikZ{\squaresfilled{4}{4}{0}{0};\squaresempty{4}{2}{0}{2};}}
+
\entropy{\centeredTikZ{\squaresfilled{4}{4}{0}{0};\squaresempty{2}{4}{2}{0};}}\\
&-
\entropy{\centeredTikZ{\squaresfilled{4}{4}{0}{0};\squaresempty{2}{2}{2}{2};}},
\end{aligned}
\end{equation*}
and
\begin{equation*}
\begin{aligned}
\entropy{\centeredTikZ{\squaresempty{6}{6}{0}{0};}} &= 
\entropy{\centeredTikZ{\squaresfilled{6}{6}{0}{0};\squaresempty{4}{4}{0}{0};}}
+
\entropy{\centeredTikZ{\squaresfilled{6}{6}{0}{0};\squaresempty{4}{4}{2}{0};}}
-
\entropy{\centeredTikZ{\squaresfilled{6}{6}{0}{0};\squaresempty{4}{4}{0}{2};}}
-
\entropy{\centeredTikZ{\squaresfilled{6}{6}{0}{0};\squaresempty{4}{4}{2}{2};}} \\
&-
\entropy{\centeredTikZ{\squaresfilled{6}{6}{0}{0};\squaresempty{4}{2}{0}{2};}}
-
\entropy{\centeredTikZ{\squaresfilled{6}{6}{0}{0};\squaresempty{4}{2}{2}{2};}}
-
\entropy{\centeredTikZ{\squaresfilled{6}{6}{0}{0};\squaresempty{2}{4}{2}{0};}}
-
\entropy{\centeredTikZ{\squaresfilled{6}{6}{0}{0};\squaresempty{2}{4}{2}{2};}}\\
&+
\entropy{\centeredTikZ{\squaresfilled{6}{6}{0}{0};\squaresempty{2}{2}{2}{2};}},
\end{aligned}
\end{equation*}
which again all follow from the same rule based on the counting of clusters. The same conclusion continues to hold even if we replace each cluster in Eq.~\eqref{eq:ent_assumption1} and~\eqref{eq:ent_assumption2} to $n\times m$ cluster for $n, m \in \mathbb{Z}^+$.

Admittedly, this discussion is somewhat heuristic. However, an important point is that there is a large class of quantum many-body systems that are expected to satisfy Eq.~\eqref{eq:ent_assumption1} and~\eqref{eq:ent_assumption2}. We have supplemented our heuristic discussion with concrete many-body quantum states in Appendix~\ref{appendix:examples} for which these conditions are satisfied either exactly (at zero temperature) or approximately (at finite temperature).

So far, we found that the reduced density matrices of a state that satisfies Eq.~\eqref{eq:ent_scaling} obeys Eqs.~\eqref{eq:ent_assumption1} and~\eqref{eq:ent_assumption2} over the $2\times 2$ and $3\times 3$ clusters. Moreover, these reduced density matrices must be \emph{locally consistent} with each other, meaning that their measurement statistics must be identical on their overlapping supports. In this paper, we put forward the following conjecture, which would establish a map in the \emph{opposite} direction:
\begin{conjecture}
Consider a set of locally consistent marginals that satisfy Eqs.~\eqref{eq:ent_assumption1} and~\eqref{eq:ent_assumption2}. The maximum-entropy state consistent with those marginals exists and furthermore obeys $S(ABC)= S(AB) + S(BC) - S(B)$ for every $A, B,$ and $C$ such that (i) $A, B, C, AB, BC,$ and $ABC$ are all disk-like regions and (ii) $A$ and $C$ are not adjacent to each other.
\end{conjecture}
\noindent
The part of the conjecture that involves the entropy is formalizing an expectation that Eq.~\eqref{eq:ent_scaling} holds everywhere. In other words, this conjecture states that the local reduced density matrices satisfying Eqs.~\eqref{eq:ent_assumption1} and~\eqref{eq:ent_assumption2} form a ``dual'' (or more precisely, an equivalent) description of a density matrix that obeys Eq.~\eqref{eq:ent_scaling} everywhere. If true, an immediate application of this duality would be an efficient computation of the expectation values of local observables and the entanglement entropy. The former can be computed from the marginals directly and the latter can be computed from the fact that the leading terms of the entropy obeys the inclusion-exclusion principle.

We provide two nontrivial evidences in support of this conjecture, by focusing on translationally-invariant systems in two spatial dimensions. First, given a set of locally consistent marginals satisfying Eqs.~\eqref{eq:ent_assumption1} and~\eqref{eq:ent_assumption2}, we show that for any $N, M\geq 2$, there always exists a global state on a $N\times M$ cluster which is consistent with the marginal \miniblockfull and its translations. Second, we derive a closed-form expression for the maximum entropy of the $N\times M$ cluster that is consistent with Eqs.~\eqref{eq:ent_assumption1} and~\eqref{eq:ent_assumption2}. This expression is consistent with Eq.~\eqref{eq:ent_scaling}.

More concretely, we assume that the marginals \miniblockfull and \mymarginal obey the following set of constraints:
\begin{equation}
\begin{aligned}
\left\{    \centeredTikZ{\squaresempty{2}{2}{0}{0};}, 
    \centeredTikZ{\squaresempty{3}{3}{0}{0};} \right\} \consistent
    \mathcal{T} \left(\left\{\centeredTikZ{\squaresempty{2}{2}{0}{0};},
    \centeredTikZ{\squaresempty{3}{3}{0}{0};}
    \right\}\right)
\end{aligned}
\label{eq:intro_ti}
\end{equation}
and Eqs.~\eqref{eq:ent_assumption1} and~\eqref{eq:ent_assumption2}. Here $\mathcal{T}= \langle t_x, t_y \rangle$ is the group of translations generated by the $x$- and the $y$-translations ($t_x$ and $t_y$ respectively) and $\mathcal{T}(S)$ of a set $S$ is an orbit of $\mathcal{T}$. The notation $\{ \rho_1, \rho_2, \ldots \} \consistent \{ \sigma_1, \sigma_2, \ldots \}$ means $\rho_i \consistent \sigma_j$ for all $\rho_i$ and $\sigma_j$, which is another way of saying that $\rho_i$ and $\sigma_j$ have identical reduced density matrices on their overlapping supports for all $i$ and $j$. For instance, if $\rho_1$ acts on $\mathcal{H}_A \otimes \mathcal{H}_B$ and $\sigma_1$ acts on $\mathcal{H}_B \otimes \mathcal{H}_C$, $\rho_1 \consistent \sigma_1$ means $\Tr_A (\rho_1) = \Tr_C(\sigma_1)$.

From these assumptions, we deduce that (for any $N$ and $M$) there is a state $\rho_{N\times M}$ on a $N\times M$ cluster such that
\begin{equation}
\boxed{
    \rho_{N\times M} \consistent \mathcal{T}\left(\left\{\centeredTikZ{\squaresempty{2}{2}{0}{0};} \right\} \right).}\label{eq:intro_consistency}
\end{equation}
That is, there is a state that is consistent with the marginals  on the $2\times 2$ cluster and its translations. Moreover, we obtain an exact expression for the maximum entropy that obeys Eq.~\eqref{eq:intro_consistency}. This expression can be readily computed from the marginal \miniblockfull.

These expressions are particularly useful for directly probing the thermodynamic limit. Given a translationally-invariant Hamiltonian, the expectation value of the local term with respect to the marginal \miniblockfull ought to be consistent with the energy density of some translationally invariant quantum state, upper bounding the ground state energy density. Moreover, in the infinite-volume limit, the expression for the maximum entropy density reads:
\begin{equation}
\boxed{
\begin{aligned}
    &\lim_{N,M\to \infty} \left( \max_{\tilde{\rho}_{N\times M} \consistent \mathcal{T}\left(\left\{\centeredTikZmini{\squaresempty{2}{2}{0}{0};} \right\} \right)}\left(\frac{\entropy{\tilde{\rho}_{N\times M}}}{NM}\right) \right)
    \\&= 
    \entropy{\centeredTikZ{\squaresempty{2}{2}{0}{0};}}
    -\entropy{\centeredTikZ{\squaresempty{2}{1}{0}{0};\squaresempty{1}{1}{0}{1};}}.
\end{aligned}
}
\end{equation}
Combining these two results, we can obtain an upper bound on the thermodynamic free energy density that can be readily computed from \miniblockfull.

To reach this conclusion, we extensively use a technique called \emph{merging}, first introduced by Kato \emph{et al.}~\cite{Kato2016}. The authors of Ref.~\cite{Kato2016} proved the following statement. Consider two quantum states $\rho_{ABC}$ and $\sigma_{BCD}$ such that $I(A:C|B)_{\rho} = I(B:D|C)_{\sigma} =0$ and $\rho \consistent \sigma$, where $I(A:C|B)_{\rho}$ is the conditional mutual information of the state $\rho$, defined as $I(A:C|B)_{\rho} = S(\rho_{AB}) + S(\rho_{BC}) - S(\rho_B) - S(\rho_{ABC}).$ Then there exists a quantum state on $ABCD$ that is consistent with both $\rho_{ABC}$ and $\sigma_{BCD}$. Moreover, the maximum entropy state consistent with those two marginals again satisfy a conditional independence relation, namely $I(A:CD|B)_{\tau} = I(AB:D|C)_{\tau} =0$.

An important fact about the merging technique is that one can deduce the existence of a consistent global state by merely verifying \emph{local} conditions. Note that all the required conditions can be verified by either (i) computing the entanglement entropies of the given density matrices or (ii) comparing their reduced density matrices. Once those conditions are verified, we can deduce the existence of \emph{some} state $\tau$ that is consistent with both $\rho$ and $\sigma$. Moreover, the maximum-entropy state consistent with those marginals again have zero conditional mutual information. 

That the outcome of the merging process is a state with zero conditional mutual information implies that the merging process can be bootstrapped. In the context of our work, this bootstrapping process can be schematically expressed as follows:
\begin{equation}
\begin{aligned}
     \twobytwounit{(x,y)}  &\longrightarrow 
    \adulthorizontalsnake{1}{y}{N}{y\shortplus 1} \\
    &\longrightarrow
    \thicksnake{1}{1}{N}{M},
\end{aligned}
\end{equation}
where the tuples represent locations of the clusters. We defer the exact meaning of this diagram to the later part of the paper, specifically, to Section~\ref{sec:main_results}. Roughly speaking, by using the merging technique, we merge the marginals on the $2\times 2$ clusters to obtain a set of density matrices on a $N\times 2$ cluster for any $N\geq 2$. Then, we merge these density matrices to obtain a density matrix over a $N\times M$ cluster for any $N,M\geq 2$. By showing that this argument applies to any $N$ and $M$, we can conclude that there is a translationally invariant state consistent with \miniblockfull if the requisite conditions in Eqs.~\eqref{eq:ent_assumption1},~\eqref{eq:ent_assumption2}, and~\eqref{eq:intro_ti} are satisfied. Moreover, using the conditional independence of the maximum-entropy merged state, we can decompose the entropy of the merged state into a linear combination of entropies that can be readily computed from the marginals. 

An important subtlety that we have not discussed up to this point is \emph{how} to use the merging technique. To apply the technique, one must establish certain internal conditional independence relations. At first, one may think that such a relation follows straightforwardly from the repeated application of the merging technique. While this is, in some sense, true, establishing the requisite conditional independence relations is in fact not a straightforward task.

To get a hint of the underlying complexity of this problem, it is helpful to consider a simplified example. Suppose that, somehow, one managed to construct a density matrix over the following clusters:
\begin{equation*}
    \centeredTikZ{\squaresempty{5}{5}{0}{0}; \squaresempty{4}{1}{0}{5};}.
\end{equation*}
Next, the goal is to extend this density matrix to the density matrix supported on the following  (enlarged) set of clusters:
\begin{equation*}
    \centeredTikZ{\squaresempty{5}{5}{0}{0}; \squaresempty{5}{1}{0}{5};}.
\end{equation*}
At first, one may hope to do so by merging density matrices over the following two clusters:
\begin{equation*}
\centeredTikZ{\squaresempty{5}{6}{0}{0};\squaresdotted{1}{1}{4}{5};},
\centeredTikZ{\squaresdotted{5}{6}{0}{0};\squaresempty{3}{3}{2}{3};},
\end{equation*}
and choose the subsystems in the following way:
\begin{equation*}
\begin{aligned}
    A&=\centeredTikZ{\squaresempty{5}{6}{0}{0};\squaresdotted{3}{3}{2}{3};},
    B&=\centeredTikZ{\squaresdotted{5}{6}{0}{0};\squaresempty{3}{3}{2}{3};\squaresdotted{2}{2}{3}{4};},\\
    C&=\centeredTikZ{\squaresdotted{5}{6}{0}{0};\squaresempty{2}{2}{3}{4};\squaresdotted{1}{1}{4}{5};},
    D&=\centeredTikZ{\squaresdotted{5}{6}{0}{0};\squaresempty{1}{1}{4}{5};},
\end{aligned}
\end{equation*}
hoping to verify the aforementioned conditions, \emph{i.e.,} $I(A:C|B)=0$ and $I(B:D|C)=0$. (The red dotted squares are simply the bookkeeping devices to keep track of the relative locations of the white squares. A more detailed explanation of this convention shall be given in Section~\ref{sec:marginals}.) While the second condition is relatively easy to verify (see Appendix~\ref{appendix:descendants}), the first condition is far from trivial.

The key issue is that the condition $I(A:C|B)=0$ involves subsystems that cannot be supported strictly inside of the $3\times 3$ clusters. As such, one cannot directly use the conditions imposed on a single patch of $3\times 3$ clusters to derive the requisite conditional independence condition. To do so, one must somehow \emph{combine} the conditions on overlapping patches of $3\times 3$ clusters.

Combining these local conditions into a nonlocal conditional independence condition is a difficult task and the tools developed in Section~\ref{sec:merging_algebra} will prove to be indispensable for proving such statements; without such tools, the proof is expected to be significantly longer and much less modular. With that said, a further simplification of the proof will be undoubtedly useful for better understanding our construction. This is left for future work.

\section{Fundamental marginals}
\label{sec:marginals}
In this section, we introduce the key objects behind our work, the \emph{fundamental marginals}. These are the basic building blocks of our theory, denoted as 
\begin{equation}
\mathcal{M}=
    \left\{ \centeredTikZ{\squaresempty{2}{2}{0}{0};},  \centeredTikZ{\squaresempty{3}{3}{0}{0};}\right\}.
\end{equation}
The two elements of $\mathcal{M}$ respectively represent density matrices over the $2\times 2$ and $3\times 3$ clusters.\footnote{While both marginals are fundamental marginals, the former will play a more important role. We will prove that there is a global state consistent with the former. However, we do not know if the same conclusion applies to the latter. For our purpose, this is not an issue because all the physically important quantities such as the energy and the (maximum) global entropy consistent with these marginals can be computed directly from the former.}

As it stands, there is an ambiguity because we have not specified the location of these clusters. We will simplify this issue by imposing translational invariance on these marginals. These are the linear constraints, specified as $\mathcal{C}_L$:
\begin{equation}
\mathcal{C}_L:
\left\{    \centeredTikZ{\squaresempty{2}{2}{0}{0};}, 
    \centeredTikZ{\squaresempty{3}{3}{0}{0};} \right\} \consistent
    \mathcal{T} \left(\left\{\centeredTikZ{\squaresempty{2}{2}{0}{0};},
    \centeredTikZ{\squaresempty{3}{3}{0}{0};}
    \right\}\right)
\label{eq:translational_invariance}
\end{equation} 
In particular, the reduced density matrices of these marginals are locally consistent with each other. 

While $\mathcal{C}_L$ removes the need to specify the absolute location of these clusters in many cases, sometimes these locations will be crucial to our analysis. Two kinds of situations will occur. In one case, the location of these clusters within some bounded region will matter. However, the exact location of the region itself will be immaterial to our analysis. In that case, we shall specify this bounded region by dotted squares. One can view this as a ``canvas'' on which the clusters representing the support of the marginals are drawn. For instance,
\begin{equation*}   \centeredTikZ{\squaresdotted{4}{3}{0}{0};\squaresempty{2}{2}{0}{0};} \, = \, \centeredTikZ{\squaresempty{2}{2}{0}{0}} \, \text{ at the bottom left corner of } \, \centeredTikZ{\squaresdotted{4}{3}{0}{0}}.
\end{equation*}
In some cases, we will deal with unbounded regions. In that case, we can no longer avoid specifying the absolute coordinate of the clusters. In those cases, we will use the following convention:
\begin{equation*}
\twobytwounit{(x,y)},
\mym{
\centeredTikZ{\squaresempty{3}{3}{0}{0}}_{(x, y)}
}, \ldots
\end{equation*}
which represent the marginals whose bottom-right corner of their support is located at $(x,y)$. We will say these clusters are \emph{anchored} at $(x,y)$.

The fundamental marginals obey extra constraints inherited from Eqs.~\eqref{eq:ent_assumption1} and ~\eqref{eq:ent_assumption2}. After using the translational invariance condition, we obtain the following \emph{primary Markovian constraints}, primaries for short.

\begin{numcases}{\mathcal{C}_{M,P}:}
    \entropy{\centeredTikZ{\squaresempty{2}{1}{0}{0};
    \squaresempty{1}{1}{0}{1};
    }} = \entropy{\centeredTikZ{\squaresempty{2}{1}{0}{0};}}
    +
    \entropy{\centeredTikZ{\squaresempty{1}{2}{0}{0};}}
    -
    \entropy{\centeredTikZ{\squaresempty{1}{1}{0}{0};}} , \label{constraint:primary1}\\
    \entropy{\centeredTikZ{\squaresempty{2}{1}{0}{0};
    \squaresempty{1}{1}{1}{-1};
    }} = \entropy{\centeredTikZ{\squaresempty{2}{1}{0}{0};}}
    +
    \entropy{\centeredTikZ{\squaresempty{1}{2}{0}{0};}}
    -
    \entropy{\centeredTikZ{\squaresempty{1}{1}{0}{0};}}, \label{constraint:primary2}\\
    \entropy{\centeredTikZ{\squaresempty{3}{3}{0}{0};}} =
    4\entropy{\centeredTikZ{\squaresempty{2}{2}{0}{0};}}
    -
    2\entropy{\centeredTikZ{\squaresempty{2}{1}{0}{0};}}
    -
    2\entropy{\centeredTikZ{\squaresempty{1}{2}{0}{0};}}
    +
    \entropy{ \centeredTikZ{\squaresempty{1}{1}{0}{0};}}. \label{constraint:primary3}
\end{numcases}

The primaries are important because they imply a number of \emph{descendant constraints}, descendants for short. The descendants are the ``emergent'' constraints that arise from the primaries, which follow straightforwardly from the strong subadditivity of entropy(SSA)~\cite{Lieb1973}. While the descendants are more intimately related to the core of our proof, the number of primaries is significantly fewer compared to that of the descendants. Therefore, the set of primaries is a more economical way of organizing the key assumptions behind our work.

The set of descendant constraints we use is summarized in Table~\ref{table:constraints}; see Appendix~\ref{appendix:descendants} for the derivation.\footnote{Let us make a side remark. Let $\mathcal{C}_{M,D}$ be the constraints summarized in Table~\ref{table:constraints}. One can actually show that the set of density matrices obeying \emph{both} $\mathcal{C}_{M,D}$ and $\mathcal{C}_L$ is equal to the set of density matrices obeying $\mathcal{C}_{M,P}$ and $\mathcal{C}_L$. We can formally state this fact as
\begin{equation}
    \mathcal{C}_{M,P}, \mathcal{C}_L \cong \mathcal{C}_{M,D}, \mathcal{C}_L.
\end{equation}
Therefore, one can choose either of these constraints without losing any generality.} In deriving our main results, we shall frequently consult this table and refer to these constraints by their category and the lower-case Roman letters. For example,
\hyperref[constraints:cell1]{Type I-a) cluster constraint} would refer to the first equation in Eq.~\eqref{constraints:cell1}.

\begin{table*}
  \framebox[\textwidth]{\parbox{0.95\textwidth}{
      \vskip 5pt
      \raggedright
      \textbf{Snake constraints:}
      \begin{equation}
      \begin{matrix}
      \text{\textbf{a)} } \entropy{\centeredTikZ{\squaresempty{2}{1}{0}{0};\squaresempty{1}{1}{0}{1};}} = \entropy{\centeredTikZ{\squaresempty{2}{1}{0}{0};}} + \entropy{\centeredTikZ{\squaresempty{1}{2}{0}{0};}}- \entropy{\centeredTikZ{\squaresempty{1}{1}{0}{0};}}, \quad
      &
      \quad 
      \text{\textbf{b)} } \entropy{\centeredTikZ{\squaresempty{3}{1}{0}{0};}} = 2\entropy{\centeredTikZ{\squaresempty{2}{1}{0}{0};}} - \entropy{\centeredTikZ{\squaresempty{1}{1}{0}{0};}}, \\
      \text{\textbf{c)} }
      \entropy{\centeredTikZ{\squaresempty{2}{1}{0}{0};\squaresempty{1}{1}{1}{-1};}} = \entropy{\centeredTikZ{\squaresempty{2}{1}{0}{0};}} + \entropy{\centeredTikZ{\squaresempty{1}{2}{0}{0};}}- \entropy{\centeredTikZ{\squaresempty{1}{1}{0}{0};}}, \quad
      &
      \quad 
      \text{\textbf{d)} } \entropy{\centeredTikZ{\squaresempty{1}{3}{0}{0};}} = 2\entropy{\centeredTikZ{\squaresempty{1}{2}{0}{0};}} - \entropy{\centeredTikZ{\squaresempty{1}{1}{0}{0};}}.
\end{matrix}
\label{constraints:snake}
\end{equation}
      \vskip 5pt
      \textbf{Type I cluster constraints:}

\begin{equation}
\begin{matrix}
\text{\textbf{a)} } &
\entropy{\centeredTikZ{\squaresempty{3}{2}{0}{0};}}=2 \entropy{\centeredTikZ{\squaresempty{2}{2}{0}{0};}} - 
\entropy{\centeredTikZ{\squaresempty{1}{2}{0}{0};}}, 
\\[10pt] 
\text{\textbf{b)} }  &
\entropy{\centeredTikZ{\squaresempty{3}{2}{0}{0};\squaresempty{2}{1}{0}{2};}}
=
\entropy{\centeredTikZ{\squaresempty{3}{2}{0}{0};}} +
\entropy{\centeredTikZ{\squaresempty{2}{2}{0}{0};}}-
\entropy{\centeredTikZ{\squaresempty{2}{1}{0}{0};}},
\\[10pt] 
\text{\textbf{c)} }  &
\entropy{\centeredTikZ{\squaresempty{3}{2}{0}{0};\squaresempty{2}{1}{0}{2};}}
=
\entropy{\centeredTikZ{\squaresempty{2}{3}{0}{0};}} +
\entropy{\centeredTikZ{\squaresempty{2}{2}{0}{0};}}-
\entropy{\centeredTikZ{\squaresempty{1}{2}{0}{0};}},
\\[10pt]
\text{\textbf{d)} }  &
\entropy{\centeredTikZ{\squaresempty{3}{3}{0}{0};}}
=
\entropy{\centeredTikZ{\squaresempty{3}{2}{0}{0};\squaresempty{2}{1}{0}{2};}} +
\entropy{\centeredTikZ{\squaresempty{2}{2}{0}{0};}}-
\entropy{\centeredTikZ{\squaresempty{2}{1}{0}{0};\squaresempty{1}{1}{0}{1};}}.
\end{matrix}
\label{constraints:cell1}
\end{equation}

\vskip 5pt
      \textbf{Type II cluster constraints:}

\begin{equation}
\begin{matrix}
\text{\textbf{a)} } &
\entropy{\centeredTikZ{\squaresempty{2}{3}{0}{0};}} = 
2 \entropy{\centeredTikZ{\squaresempty{2}{2}{0}{0};}} -
\entropy{\centeredTikZ{\squaresempty{2}{1}{0}{0};}}, \\[10pt]
\text{\textbf{b)} } &
\entropy{\centeredTikZ{\squaresempty{3}{2}{0}{1};\squaresempty{2}{1}{1}{0};}}
=
\entropy{\centeredTikZ{\squaresempty{3}{2}{0}{0};}} +
\entropy{\centeredTikZ{\squaresempty{2}{2}{0}{0};}}-
\entropy{\centeredTikZ{\squaresempty{2}{1}{0}{0};}},
\\[10pt]
\text{\textbf{c)} } &
\entropy{\centeredTikZ{\squaresempty{3}{2}{0}{1};\squaresempty{2}{1}{1}{0};}}
=
\entropy{\centeredTikZ{\squaresempty{2}{3}{0}{0};}} +
\entropy{\centeredTikZ{\squaresempty{2}{2}{0}{0};}}-
\entropy{\centeredTikZ{\squaresempty{1}{2}{0}{0};}},
\\[10pt]
\text{\textbf{d)} } &
\entropy{\centeredTikZ{\squaresempty{3}{3}{0}{0};}}
=\entropy{\centeredTikZ{\squaresempty{3}{2}{0}{1};\squaresempty{2}{1}{1}{0};}} +
\entropy{\centeredTikZ{\squaresempty{2}{2}{0}{0};}}-
\entropy{\centeredTikZ{\squaresempty{2}{1}{0}{0};\squaresempty{1}{1}{1}{-1};}}.
\end{matrix}
\label{constraints:cell2}
\end{equation}
    }
  }\caption{Descendant constraints. These identities follow from $\mathcal{C}_L$ and $\mathcal{C}_{M,P}$. \label{table:constraints}}
\end{table*}

\section{Interlude: Entropy}
\label{sec:prelim}
In this section, we provide a streamliend overview of the properties of the von Neumann entropy. Of particular interest to us is the strong subadditivity (SSA) of entropy~\cite{Lieb1973}. Without loss of generality, consider a tripartite state $\rho$, say over $\mathcal{H}_A \otimes \mathcal{H}_B \otimes \mathcal{H}_C$. SSA is the following statement:
\begin{equation}
    S(\rho_{AB}) + S(\rho_{BC}) - S(\rho_B) - S(\rho_{ABC})\geq 0.
\end{equation}

We will frequently use an object called \emph{conditional mutual information}, defined as \begin{equation}
    I(A:C|B)_{\rho} := S(\rho_{AB}) + S(\rho_{BC}) - S(\rho_B) - S(\rho_{ABC}).
\end{equation}
We will refer to $B$ as the \emph{conditioning subsystem} and $A$ and $C$ as the \emph{conditioned subsystems}.

Conditional mutual information is useful because it obeys a number of nontrivial \emph{monotonicity} relations. Specifically, conditional mutual information is non-increasing under a removal of a conditioned subsystem.
\begin{equation}
    I(A:C|B) \leq I(A:CD|B) \label{eq:ssa_monotinicity_intro}
\end{equation}
Also, conditional mutual information is nonincreasing under moving (a part of) conditioned subsystem to the conditioning subsystem.
\begin{equation}
    I(A:C|BD) \leq I(A:CD|B) \label{eq:ssa_monotonicity_intro2}
\end{equation}
Both of these inequalities follow straightforwardly from the SSA.

Some states satisfy SSA with an \emph{equality}. Such states are said to be \emph{conditionally independent} and are important for us for two reasons. First, all the identities summarized in Table~\ref{table:constraints} are precisely statements about certain states being conditionally independent. Second, conditionally independent states have a very special structure. The following theorem was proved by Petz~\cite{Petz1988,Petz2003}.
\begin{theorem}
\cite{Petz2003} $I(A:C|B)=0$ if and only if
\begin{equation}
\begin{aligned}
    \rho_{ABC} &= \Phi_{B\to BC}(\rho_{AB}) \\
    &= \Phi_{B\to AB} (\rho_{BC}),
\end{aligned}
\end{equation}
where $\Phi_{B,\to BC}$ and $\Phi_{B\to AB}$ are \emph{Petz maps}, defined as 
\begin{equation}
\begin{aligned}
    \Phi_{B\to BC}(\cdot) &:= \rho_{BC}^{\frac{1}{2}}\rho_B^{-\frac{1}{2}}(\cdot)\rho_B^{-\frac{1}{2}}\rho_{BC}^{\frac{1}{2}},\\
    \Phi_{B\to AB}(\cdot )&:= \rho_{AB}^{\frac{1}{2}}\rho_B^{-\frac{1}{2}}(\cdot) \rho_B^{-\frac{1}{2}}\rho_{AB}^{\frac{1}{2}},
\end{aligned}
\end{equation}
\label{thm:Petz}
\end{theorem}
\noindent
where we have suppressed the identity operator for convenience. For instance, $\rho_B$ should be really viewed as $I_A\otimes \rho_B \otimes I_C$. This theorem will be useful because we can build up the ``global'' state $\rho_{ABC}$ from its marginals $\rho_{AB}$ and $\rho_{BC}$. The main difficulty of using Theorem~\ref{thm:Petz} comes from the fact that the reduced density matrices generally do not commute with each other. This motivates the development of the \emph{merging algebra}, explained in our next Section.

\section{Merging algebra}
\label{sec:merging_algebra}
The constraints in Table~\ref{table:constraints} imply that the reduced density matrices of the fundamental marginals can be \emph{merged} in a particular way. In this section, we discuss a variety of facts that will be useful in proving such a statement. 

The main content of this section is an algebra called the \emph{merging algebra}. Most of the objects of this algebra belong to the following set.
\begin{equation}
    \mathcal{S}_{\Lambda} :=\bigcup_{A\in \mathcal{P}(\Lambda)} \mathcal{S}(\mathcal{H}_{A}),
\end{equation}
where $\mathcal{P}(S)$ of a set $S$ is the power set of $S$, $\Lambda$ is the set of lattice sites, $\mathcal{H}_A$ is the Hilbert space associated with $A\subset \Lambda$, and $\mathcal{S}(\mathcal{H})$ of a Hilbert space $\mathcal{H}$ is the state space of $\mathcal{H}$. From now on, unless specified otherwise, we will assume that density matrices belong to $\mathcal{S}_{\Lambda}$.

The merging algebra is equipped with two binary operations, referred to as \emph{right-merge} and \emph{max-merge} operations. Consider $\sigma, \lambda \in \mathcal{S}_{\Lambda}$. These binary operations are defined as follows.
\begin{equation}
\begin{aligned}
    &\text{Right-merge}:\, \sigma \rightmerge \lambda := \Phi_{\lambda}(\sigma)\\
     &\text{Max-merge}:\, \sigma \maxmerge \lambda := \argmax_{\rho \consistent \sigma,\rho \consistent \lambda} \left(S(\rho) \right)
\end{aligned}
    \label{eq:merging_ops_definitions}
\end{equation}
Here $\Phi_{\lambda}$ is the Petz map, which in this context is defined as follows. (See Theorem~\ref{thm:Petz}.) Let $B := \text{Supp}(\lambda) \cap \text{Supp}(\sigma)$. Then the action of $\Phi_{\lambda}$ is defined as 
\begin{equation}
    \Phi_{\lambda}(\sigma) := \lambda^{\frac{1}{2}}\lambda_B^{-\frac{1}{2}}\sigma \lambda_B^{-\frac{1}{2}} \lambda^{\frac{1}{2}}.
\end{equation}
Note that this is a slight abuse of notation because the map $\Phi_{\lambda}$ actually depends on the argument; the arguments with different supports may give rise to a different quantum channel.

By definition, the max-merge operation is commutative. However, the right-merge is not. Moreover, neither of them are generally associative. Nevertheless, the tools developed in this section, together with Table~\ref{table:constraints}, will provide enough structure to let us derive the main results of this paper.

While the max-merge operation will always yield a valid density matrix in this paper, the same cannot be said in a more general context. For completeness, we introduce a special object called $\nil$. Given two density matrices, there may not be any density matrix (on the union of their supports) that are consistent with those density matrices. In that case, we will simply say that
\begin{equation}
    \sigma \maxmerge \lambda = \nil.
\end{equation}
Similarly, we will use the following rule:
\begin{equation}
    \begin{aligned}
    \nil \maxmerge \lambda &= \sigma \maxmerge \nil &= \nil, \\
    \nil \rightmerge \lambda &= \sigma \rightmerge \nil &= \nil.
    \end{aligned}
\end{equation}

Provided that we began with density matrices, the set of objects generated by the max-merge and the right-merge operations are either a density matrix or the $\nil$ object. The merging algebra is generated by these binary operations, applied to the reduced density matrices of \mymarginal; see Table~\ref{table:merging_algebra}. 

\begin{table*}
  \framebox[\textwidth]{\parbox{0.95\textwidth}{
      \vskip 5pt
      \raggedright
      \begin{center}
      \underline{\textbf{Merging algebra}}
      \end{center}
      
      \textbf{Generating set:}
      \[
      \left\{\centeredTikZ{\squaresempty{3}{3}{0}{0};} \text{ and its reduced density matrices, over every $3\times 3$ clusters} \right\}.
      \]
      (\mymarginal must obey $\mathcal{C}_L$ and $\mathcal{C}_{M,P}$.)
      \vskip 5pt
      
      \textbf{Binary operations:}
      \begin{equation*}
      \begin{matrix}
    \sigma \rightmerge \lambda := \Phi_{\lambda}(\sigma), \qquad &
    \qquad \sigma \maxmerge \lambda := \displaystyle\argmax_{\rho \consistent \sigma,\rho \consistent \lambda} \left(S(\rho) \right)
     \end{matrix}
\end{equation*}
where $\Phi_{\lambda}(\sigma) :=  \lambda^{\frac{1}{2}}\lambda_B^{-\frac{1}{2}} \sigma \lambda_B^{-\frac{1}{2}} \lambda^{\frac{1}{2}}$ and $B := \text{Supp}(\sigma) \cap \text{Supp}(\lambda)$.
\begin{equation*}
    \begin{aligned}
    \nil \maxmerge \lambda &= \sigma \maxmerge \nil &= \nil, \\
    \nil \rightmerge \lambda &= \sigma \rightmerge \nil &= \nil.
    \end{aligned}
\end{equation*}
    }
  }\caption{Merging algebra.   \label{table:merging_algebra}}
\end{table*}

\subsection{Basics of the merging algebra}
In this section, we will derive fundamental facts about the merging algebra. This section is a bit abstract but the meaning behind each results should be well-encapsulated by their respective names.

We begin with the \emph{fundamental lemma}, which establishes a web of equivalence relations between various statements about conditionally independent states. This is the lemma that we use the most in this paper.
\begin{lemma}
(Fundamental lemma) The following statements are equivalent:
\begin{enumerate}
    \item $\lambda \maxmerge \sigma$ exists. Moreover, $S(\lambda \maxmerge \sigma)=S(\lambda) + S(\sigma) - S(\tau)$, where $\tau$ is the reduced density matrix of $\lambda$ (and $\sigma$) on $\text{Supp}(\lambda) \cap \text{Supp}(\sigma)$.
    \item $\lambda \maxmerge \sigma = \sigma \rightmerge \lambda$.
    \item $\lambda \maxmerge \sigma = \lambda \rightmerge \sigma$.
    \item $\lambda \rightmerge \sigma \consistent \lambda$ and $\lambda \consistent \sigma$.
    \item $\sigma \rightmerge \lambda \consistent \sigma$ and $\lambda \consistent \sigma$.
    \item There exists a quantum channel $\Phi$ acting on $\text{Supp}(\lambda) \cap \text{Supp}(\sigma)$  such that $\Phi(\lambda) \consistent \{ \lambda, \sigma\}$.
    \item There exists a quantum channel $\Phi$ acting on  $\text{Supp}(\sigma) \cap \text{Supp}(\lambda)$  such that $\Phi(\sigma) \consistent \{\sigma, \lambda\}$.
\end{enumerate}
Moreover, $\Phi(\lambda)$ and $\Phi(\sigma)$ appearing in condition 6 and 7 are both the maximum entropy state consistent with $\lambda$ and $\sigma$.
\label{lemma:fundamental}
\end{lemma}
\begin{proof}
First, $1\to 2$ and $1\to 3$ follows from Petz's theorem (Theorem~\ref{thm:Petz}). Moreover, $2\to 5$ and $3\to 4$ follows from the definition of the max-merge. Lastly, $4\to 6$ and $5\to 7$ is true by definition.

Now we prove $6\to 1$. Let $B$ be the intersection of the support of $\lambda$ and $\sigma$. Also, let $A= \text{Supp}(\lambda) \setminus B$ and $C=\text{Supp}(\sigma) \setminus B$.  Because conditional entropy is concave~\cite{Lieb1973},
\begin{equation}
    S(\lambda) - S(\lambda_B) \leq S(\Phi(\lambda)) - S(\Phi(\lambda_B)).
\end{equation}
Note that $\lambda = (\Phi(\lambda))_{AB}$ because of our very assumption. Moreover, $\lambda_B = \Phi(\lambda)_{B}$ by the same reason. Lastly, $\Phi(\lambda_B) = \Tr_A (\Phi(\lambda)) = (\Phi(\lambda))_{BC}$. Therefore, we have
\begin{equation}
S(\tau_{AB}) + S(\tau_{BC}) - S(\tau_B) - S(\tau_{ABC})\leq 0,
\end{equation}
where $\tau_{ABC} = \Phi(\lambda)$. By SSA, we have
\begin{equation}
S(\tau_{AB}) + S(\tau_{BC}) - S(\tau_B) - S(\tau_{ABC})\geq 0.
\end{equation}
Therefore, it must be that
\begin{equation}
S(\tau_{AB}) + S(\tau_{BC}) - S(\tau_B) - S(\tau_{ABC})=0.
\end{equation}

Note that $\Phi(\lambda)$ is consistent with both $\lambda$ and $\sigma$. Moreover, its entropy is maximum over all possible states that are consistent with both $\lambda$ and $\sigma$. Therefore, a density matrix $\lambda \maxmerge \sigma$  satisfying the first condition exists. Such a state is unique~\cite{Kim2014}, so it must be $\Phi(\lambda)$. The proof of $7 \to 1$ is similar.
\end{proof}

Next, we have the commutation lemma. Colloquially speaking, the commutation lemma provides a sufficient condition under which two right-merges commute with each other. This is an intuitive result. If the two density matrices associated with the right-merges do not have overlapping supports, the Petz map assoicated with them act on disjoint supports. As such, their action should commute. 
\begin{lemma}
(Commutation lemma)
If $\text{Supp}(\lambda) \cap \text{Supp}(\tau) = \varnothing$, then for any density matrix $\sigma$,
\begin{equation}
(\sigma \rightmerge \lambda) \rightmerge \tau = (\sigma \rightmerge \tau) \rightmerge \lambda.    
\end{equation}
\label{lemma:commutation}
\end{lemma}
\begin{proof}
Specifically, 
\begin{equation}
\begin{aligned}
    (\sigma \rightmerge \lambda) \rightmerge \tau &= \Phi_{\tau}(\Phi_{\lambda}(\sigma)) \\
    &= \Phi_{\tau}(\lambda^{\frac{1}{2}}\lambda_{B}^{-\frac{1}{2}}\sigma \lambda_{B}^{-\frac{1}{2}}\lambda^{\frac{1}{2}}) \\
    &= \tau^{\frac{1}{2}}\tau_{B'}^{-\frac{1}{2}}(\lambda^{\frac{1}{2}}\lambda_{B}^{-\frac{1}{2}}\sigma \lambda_{B}^{-\frac{1}{2}}\lambda^{\frac{1}{2}})\tau_{B'}^{-\frac{1}{2}}\tau^{\frac{1}{2}},
\end{aligned}
\end{equation}
where $B= \text{Supp}(\lambda) \cap \text{Supp}(\sigma)$ and $B' = \text{Supp}(\tau) \cap (\text{Supp}(\lambda) \cup \text{Supp}(\sigma))$. Because $\text{Supp}(\lambda) \cap \text{Supp}(\tau) = \varnothing,$ $B'=\text{Supp}(\tau) \cap \text{Supp}(\sigma)$. In particular, $B\cap B'= \varnothing$. Therefore, 
\begin{equation}
    [\tau^{\frac{1}{2}}\tau_{B'}^{-\frac{1}{2}}, \lambda^{\frac{1}{2}} \lambda_B^{-\frac{1}{2}}]=0.
\end{equation}
Therefore, we conclude that
\begin{equation}
(\sigma \rightmerge \lambda) \rightmerge \tau =  \lambda^{\frac{1}{2}}\lambda_{B}^{-\frac{1}{2}}(\tau^{\frac{1}{2}}\tau_{B'}^{-\frac{1}{2}}
\sigma \tau_{B'}^{-\frac{1}{2}}\tau^{\frac{1}{2}})\lambda_{B}^{-\frac{1}{2}}\lambda^{\frac{1}{2}}.
\label{eq:mergelemma_triple_temp}
\end{equation}
It is straightforward to show that the right hand side of Eq.~\eqref{eq:mergelemma_triple_temp} is equal to $(\sigma \rightmerge \tau) \rightmerge \tau$, establishing Eq.~\eqref{eq:merginglemma_triple_key1}.
\end{proof}

Next, we have the merging lemma. This lemma was originally proved by Kato \emph{et al.}~\cite{Kato2016}. Given two conditionally independent states, it provides a sufficient condition under which one can merge those two states into a larger state. Importantly, the merged state possesses nontrivial conditional independence relations that are inherited from the states prior to the merging.
\begin{lemma}
(Merging lemma)
Suppose 
\begin{equation}
\begin{aligned}
    \lambda \maxmerge \sigma &= \lambda \rightmerge \sigma \\
    \sigma \maxmerge \tau &= \sigma \rightmerge \tau
\end{aligned}
\end{equation}
 and $\text{Supp}(\lambda) \cap \text{Supp}(\tau) = \varnothing$. Then
\begin{equation}
    \begin{tikzcd}
    (\lambda \maxmerge \sigma) \maxmerge \tau \arrow[r,equal]\arrow[d, equal] & (\lambda \maxmerge \sigma) \rightmerge \tau\arrow[d, equal] \\
    (\sigma \maxmerge \tau)\maxmerge \lambda \arrow[r, equal] & (\sigma \maxmerge \tau) \rightmerge \lambda
    \end{tikzcd}
    \label{eq:merging_square}
\end{equation} 
 \label{lemma:merging_lemma}
\end{lemma}
\begin{proof}
Note that
\begin{equation}
\begin{aligned}
    (\lambda \maxmerge \sigma) \rightmerge \tau  &=  (\sigma \rightmerge \lambda) \rightmerge \tau \\
    &= (\sigma \rightmerge \tau) \rightmerge \lambda, \label{eq:merginglemma_triple_key1}
\end{aligned}
\end{equation}
where the first line is our assumption and the second line follows from the commutation lemma (Lemma~\ref{lemma:commutation}). 

The following identity is the key.
\begin{equation}
    \Tr_{\text{Supp}(\tau) \setminus \text{Supp}(\sigma)} ((\sigma \rightmerge \tau)\rightmerge \lambda) = \sigma \rightmerge \lambda.
\end{equation}
To see why, note that
\begin{equation}
    \begin{aligned}
    &\Tr_{\text{Supp}(\tau) \setminus \text{Supp}(\sigma)} ((\sigma \rightmerge \tau)\rightmerge \lambda)\\
    &=  \Tr_{\text{Supp}(\tau) \setminus \text{Supp}(\sigma)}(\sigma \rightmerge \tau) \rightmerge \lambda \\
    &= \sigma \rightmerge \lambda.
    \end{aligned}
\end{equation}
In the first line, we used the fact that the support of $\tau$ does not intersect with the support of $\lambda$. Therefore, the Petz map associated with $\lambda$ ``commutes" with the partial trace.

Note that $\sigma \rightmerge \lambda = \lambda \maxmerge \sigma$ follows from Lemma~\ref{lemma:fundamental} and our assumption. Therefore,
\begin{equation}
    (\lambda \maxmerge \sigma) \rightmerge \tau \consistent \lambda \maxmerge \sigma. 
\end{equation}
Moreover, $\lambda \maxmerge \sigma \consistent \tau$. To prove this fact, it suffices to show that $\sigma \consistent \tau$ because the support of $\lambda$ does not intersect with that of $\tau$. Because $\sigma \maxmerge \tau = \sigma \rightmerge \tau$, and $\sigma \rightmerge \tau$ is a density matrix, $\sigma \maxmerge \tau$ cannot be $\nil$. Therefore, $\sigma \consistent \tau$, which implies $\lambda \maxmerge \sigma \consistent \tau$.

Therefore, the fourth condition in Lemma~\ref{lemma:fundamental} is satisfied. Therefore, 
\begin{equation}
    (\lambda \maxmerge \sigma) \maxmerge \tau = (\lambda \maxmerge \sigma) \rightmerge \tau,
\end{equation}
establishing the horizontal identity on the top of Eq.~\ref{eq:merging_square}.

Note that our conditions, by the virtue of Lemma~\ref{lemma:fundamental}, can be rewritten as
\begin{equation}
\begin{aligned}
    \sigma \maxmerge \lambda &= \sigma \rightmerge \lambda, \\
    \tau \maxmerge \sigma &= \tau \rightmerge \sigma.
\end{aligned}
\end{equation}
Thus, we can apply the same logic to conclude that
\begin{equation}
    (\tau \maxmerge \sigma)\maxmerge \lambda = (\tau \maxmerge \sigma) \rightmerge \lambda.
\end{equation}
By the commutativity of $\maxmerge$, establishing the horizontal identity on the bottom of Eq.~\ref{eq:merging_square}.

Now, it remains to establish the vertical identities. Note that
\begin{equation}
\begin{aligned}
    (\lambda \maxmerge \sigma) \rightmerge \tau &= (\sigma \rightmerge \lambda) \rightmerge \tau \\
    &= (\sigma \rightmerge \tau)\rightmerge \lambda \\
    &= (\sigma \maxmerge \tau) \rightmerge \lambda,
\end{aligned}
\end{equation}
establishing the vertical identity on the right side of Eq.~\eqref{eq:merging_square}. This completes the proof of Eq.~\eqref{eq:merging_square}.
\end{proof}

Let us make an important remark on the merging lemma: that the lemma can be bootstrapped. Specifically, the lemma begins with two conditionally independent states and ends up with another conditionally independent state over an enlarged system. The fact that this new state is conditionally independent (with an appropriate choice of subsystems) implies that it can be merged with yet another state provided that the conditions stated in the merging lemma holds. Therefore, by chaining this argument, it is possible to start with a set of marginals on bounded regions and merge them to form a density matrix on an unbounded region. Importantly, the fact that this can be done can be verified locally, over the originally given density matrices. We shall see a nontrivial example in Section~\ref{sec:snake_intro}.

\subsection{Snake}
\label{sec:snake_intro}
There is a very important object that will repeatedly appear in the remainder of this paper. This is the \emph{snake}.\footnote{One may think of the snake as a quantum Markov chain. However, that point of view does not mesh well with the analysis in Section~\ref{sec:snakes_in_action}.} In this section, we will study their properties. Let us begin with the definition.
\begin{definition}
Consider a sequence of density matrices $\left(\rho_{i}\right)_{i=1}^N:= \left(\rho_1,\ldots, \rho_N\right)$ such that
\begin{enumerate}
    \item $\rho_i \maxmerge \rho_{i+1} = \rho_i \rightmerge \rho_{i+1}$ and 
    \item $\text{Supp}(\rho_{i}) \cap \text{Supp}(\rho_j) =\varnothing$ unless $|i-j| \leq 1$.
\end{enumerate}
Then, we define a snake of $\left(\rho_{i}\right)_{i=1}^N:= \left(\rho_1,\ldots, \rho_N\right)$ as
\begin{equation}
    \mathbb{S}\left(\left(\rho_i \right)_{i=1}^N \right) := (\cdots(\rho_1 \maxmerge \rho_2) \maxmerge \cdots) \maxmerge \rho_{N-1} ) \maxmerge \rho_N.
\end{equation}
\label{definition:snake}
\end{definition}

\begin{remark}
A snake of a sequence $\left(\rho_{i}\right)_{i=1}^N:= \left(\rho_1,\ldots, \rho_N\right)$, if it exists, is a quantum state~\cite{Kim2017}. Moreover, this state is consistent with all $\rho_i$ from $i=1$ to $N$.
\end{remark}

\begin{remark}
If a sequence of density matrices $\left(\rho_i \right)_{i=1}^N$ can form a snake, so can its subsequence $\left(\rho_i \right)_{i=n}^m$ for any $1\leq n < m \leq N.$
\end{remark}

Snake is a fluid object. It can take many different forms that are equivalent to each other. We will prove a number of results in this direction. First, let us introduce the mutation lemma. This lemma says two things. First, a snake forms a Markov chain; see Eq.~\eqref{eq:mutation_result1}. Second, it can be written as a sequence of right-merges, either from the left to the right, or from the right to the left; see Eq.~\eqref{eq:mutation_result2}. The main point of this lemma is that snake is a mutable object that can have many different forms. Depending on the context, we can ``mutate'' a snake into a form that becomes more amenable to our analysis.
\begin{lemma}
(Mutation lemma)
Consider a snake $\mathbb{S}\left(\left(\rho_i \right)_{i=1}^N \right)$. Then,
\begin{equation}
    \mathbb{S}\left(\left( \rho_i\right)_{i=1}^n\right) \maxmerge \rho_{n+1} =  \mathbb{S}\left(\left( \rho_i\right)_{i=1}^n\right) \rightmerge \rho_{n+1} \label{eq:mutation_result1}
\end{equation}
for all $1\leq n < N$. Moreover,
\begin{equation}
\begin{aligned}
    \mathbb{S}\left(\left(\rho_i \right)_{i=1}^N \right) &= (\cdots((\rho_1 \rightmerge \rho_2) \rightmerge \cdots) \rightmerge \rho_{N-1} ) \rightmerge \rho_N) \\
    &= ((\cdots(\rho_N \rightmerge \rho_{N-1})\rightmerge \cdots) \rightmerge \rho_{2} ) \rightmerge \rho_1).
\end{aligned}\label{eq:mutation_result2}
\end{equation}
\label{lemma:snake_mutation}
\end{lemma}
\begin{proof}
The proof is based on induction. For $n=1$, the statement is true by our assumptions. Suppose the claim is true for $n\geq 1$.

Let us first prove
\begin{equation}
    \mathbb{S}\left(\left( \rho_i\right)_{i=1}^n\right) \maxmerge \rho_{n+1} =  \mathbb{S}\left(\left( \rho_i\right)_{i=1}^n\right) \rightmerge \rho_{n+1}.
\end{equation}
By our assumption,

\begin{equation}
\begin{aligned}
    \mathbb{S}(\left( \rho_i\right)_{i=1}^n) \rightmerge \rho_{n+1} &= ((\cdots((\rho_n \rightmerge \rho_{n-1}) \rightmerge \rho_{n-2} )\rightmerge \cdots) \rightmerge \rho_{2} ) \rightmerge \rho_1) \rightmerge \rho_{n+1} \\
    &= ((\cdots(((\rho_n \rightmerge \rho_{n+1} ) \rightmerge \rho_{n-1}) \rightmerge \rho_{n-2} )\rightmerge \cdots) \rightmerge \rho_{2} ) \rightmerge \rho_1) \\
    &= ((\cdots(((\rho_{n+1} \rightmerge \rho_{n} ) \rightmerge \rho_{n-1}) \rightmerge \rho_{n-2} )\rightmerge \cdots) \rightmerge \rho_{2} ) \rightmerge \rho_1).
\end{aligned}
\label{eq:348}
\end{equation}

In the second line, we used the commutation lemma (Lemma~\ref{lemma:commutation}). In the third line, we used our assumption.

By tracing out $\text{Supp}(\rho_{n+1}) \setminus \text{Supp}(\mathbb{S}(\left( \rho_i \right)_{i=1}^n))$, we obtain
\begin{equation}
\mathbb{S}(\left( \rho_i\right)_{i=1}^n) \rightmerge \rho_{n+1}  \consistent \mathbb{S}(\left( \rho_i\right)_{i=1}^n).
\end{equation}
Furthermore, 
\begin{equation}
    \mathbb{S}(\left( \rho_i\right)_{i=1}^n) \rightmerge \rho_{n+1} \consistent \rho_{n+1}.
\end{equation}
Therefore, by the equivalence of the fourth condition and the third condition in Lemma~\ref{lemma:fundamental}, we conclude
\begin{equation}
    \mathbb{S}(\left( \rho_i\right)_{i=1}^n) \maxmerge \rho_{n+1} =  \mathbb{S}(\left( \rho_i\right)_{i=1}^n) \rightmerge \rho_{n+1}.
\end{equation}

Thus, we have proved Eq.~\eqref{eq:mutation_result1}. Repeatedly applying this result, the first line of Eq.~\eqref{eq:mutation_result2} follows. The second line follows from the expression in the last line of Eq.~\eqref{eq:348}.
\end{proof}

Snake has a \emph{local} entropy decomposition. The proof follows straightforwardly from Lemma~\ref{lemma:fundamental} and the mutation lemma  (Lemma~\ref{lemma:snake_mutation}).
\begin{corollary}
\begin{equation}
    S\left(\mathbb{S}(\left(\rho_i \right)_{i=1}^N)\right) = \left( \sum_{i=1}^N S(\rho_i)\right) - \left(\sum_{i=1}^{N-1} S(\tau_i) \right),
\end{equation}
where $\tau_i$ is the reduced density matrix of $\rho_i$ (and $\rho_{i+1}$) on $\text{Supp}(\rho_i) \cap \text{Supp}(\rho_{i+1})$.
\label{corollary:snake_entropy}
\end{corollary}

Another very useful property is that snakes can be ``split" into two snakes. This lemma says that a snake can be viewed as a merged state of two ``shorter'' snakes.
\begin{lemma}
(Splitting lemma) Consider a snake $\mathbb{S}\left(\left(\rho_i \right)_{i=1}^N \right)$.

\begin{equation}
\begin{tikzcd}
    \mathbb{S}\left(\left(\rho_i \right)_{i=1}^N \right)
     \arrow[r, equal]\arrow[d, equal]& \mathbb{S}\left(\left(\rho_i \right)_{i=1}^n \right) \maxmerge
        \mathbb{S}\left(\left(\rho_i \right)_{i=n+1}^N \right) \arrow[d, equal]\\
  \mathbb{S}\left(\left(\rho_i \right)_{i=n+1}^N \right) \rightmerge \mathbb{S}\left(\left(\rho_i \right)_{i=1}^n \right)     \arrow[r, equal] & \mathbb{S}\left(\left(\rho_i \right)_{i=1}^n \right) \rightmerge
        \mathbb{S}\left(\left(\rho_i \right)_{i=n+1}^N \right)
        \label{eq:1140}
\end{tikzcd}
\end{equation}

for all $n$ such that $1\leq n < N$.
\label{lemma:snake_splitting}
\end{lemma}
\begin{proof}
By the mutation lemma (Lemma~\ref{lemma:snake_mutation}), 
\begin{equation}
\begin{aligned}
    \Tr_{n\shortplus 1, \ldots, N} \left(\mathbb{S}\left(\left(\rho_i \right)_{i=1}^N \right) \right)&= \mathbb{S}\left(\left(\rho_i \right)_{i=1}^n \right),\\
    \Tr_{1, \ldots, n} \left(\mathbb{S}\left(\left(\rho_i \right)_{i=1}^N \right) \right)&= \mathbb{S}\left(\left(\rho_i \right)_{i=n+1}^N \right).
\end{aligned}
\end{equation}

From Corollary~\ref{corollary:snake_entropy}, 
\begin{equation}
\begin{aligned}
    \entropy{\mathbb{S}\left(\left(\rho_i \right)_{i=1}^N \right)}&= \entropy{\mathbb{S}\left(\left(\rho_i \right)_{i=1}^n \right)} \\ &+\entropy{\mathbb{S}\left(\left(\rho_i \right)_{i=n+1}^N \right)}
    - \entropy{\tau_n},
\end{aligned}
\end{equation}
where $\tau_n$ is the reduced  density matrix of $\mathbb{S}\left(\left(\rho_i \right)_{i=1}^n \right)$ and $\mathbb{S}\left(\left(\rho_i \right)_{i=n+1}^N \right)$) on the $\text{Supp}\left(\mathbb{S}\left(\left(\rho_i \right)_{i=1}^n \right)\right) \cap \text{Supp}\left(\mathbb{S}\left(\left(\rho_i \right)_{i=n+1}^N \right)\right)$. Thus, the first condition of Lemma~\ref{lemma:fundamental} is satisfied by $\mathbb{S}\left(\left(\rho_i \right)_{i=1}^N\right)$. By the uniqueness of such a state~\cite{Kim2014}, we establish the horizontal identity on the top. The remaining identities follow from the second and the third condition of Lemma~\ref{lemma:fundamental}.
\end{proof}

\section{Snakes in action}
\label{sec:snakes_in_action}
Up to this point, we have focused on the properties of the fundamental marginals. From here on, we will shift our focus to the \emph{composite marginals} that are made out of those marginals from a sequence of merging operations.

Specifically, we will be studying the properties of the density matrices that are created by merging the following density matrices:
\begin{equation}
    \horizontalunit{(x,y)} \qquad \text{and} \qquad  \twobytwounit{(x,y)},
\end{equation}
which represent the marginals inherited from \mymarginal over the $2\times 1$ and $2\times 2$ clusters, anchored at an arbitrary point $(x,y)$. 

From these objects, we can define the \emph{level}-$1$ and \emph{level}-$2$ snakes, introduced below.
\begin{definition}
Level-$1$ snake at row $y$ is defined as
\begin{equation}
\begin{aligned}
    \horizontalbaby{(1,y)}{(N,y)} &:= \left(\left(\horizontalunit{(2,y)} \maxmerge \horizontalunit{(3,y)} \right) \cdots \right) \\
    &\maxmerge \horizontalunit{(N,y)}.
\end{aligned}
\end{equation}
\label{def:snake1}
\end{definition}
\begin{definition}
Level-$2$ snake at row $y$ is defined as
\begin{equation}
\begin{aligned}
    \adulthorizontalsnake{1}{y}{N}{y\shortplus 1} &:= \left(\left(\twobytwounit{(2,y)} \maxmerge \twobytwounit{(3,y)} \right) \cdots \right)\\ 
    &\maxmerge \twobytwounit{(N,y)}.
\end{aligned}
\end{equation}
\label{def:snake2}
\end{definition}
\noindent
Note that both the level-1 and the level-2 snakes are indeed snakes in the sense of Definition~\ref{definition:snake}, by the \hyperref[constraints:snake]{Snake-b) constraint} and the \hyperref[constraints:cell1]{Type I-a) cluster constraint} respectively, by applying these constraints to the first condition of Lemma~\ref{lemma:fundamental}.

These snakes have properties that will play a vital role in Section~\ref{sec:consistency} and~\ref{sec:entropy}. To wit, it will be useful to introduce a notion of \emph{extension maps}. Let $\rho_{\leq y}$ be a density matrix supported on a set of sites with $y$-coordinate less or equal to $y$ and let $\rho_{\geq y}$ be a density matrix supported on a set of sites $\{(x', y'): y'\leq y \}$. The \emph{upward extension} acting on the $y$-th row, denoted as $\mathcal{E}_{y,\uparrow}$, acts as follows:
\begin{equation*}
    \mathcal{E}_{y,\uparrow}(\rho_{\leq y}) = \left(\left(\rho_{\leq y} \rightmerge \twobytwounit{(2,y)}\right)\cdots \right) \rightmerge  \twobytwounit{(N,y)}.
\end{equation*}
The \emph{downward extension} acting on the $y$-th row, denoted as $\mathcal{E}_{y,\downarrow}$, acts as follows:
\begin{equation*}
    \mathcal{E}_{y,\downarrow}(\rho_{\leq y}) = \left(\left(\rho_{\geq y} \rightmerge \twobytwounit{(N,y\shortminus 1)}\right)\cdots \right) \rightmerge  \twobytwounit{(2,y\shortminus 1)}.
\end{equation*}
Note that the extension maps ``extend" the density matrix, either in an upward or a downward direction. 

Roughly speaking, these extension maps can be thought as the ``inverses" of the partial trace of a single row. That is, for a class of quantum states we consider, we will see that
\begin{equation}
\begin{aligned}
\Tr_{y+1}\left(\mathcal{E}_{y,\uparrow}(\rho_{\leq y})\right) &= \rho_{\leq y}, \\
\Tr_{y-1}\left(\mathcal{E}_{y,\downarrow}(\rho_{\geq y})\right) &= \rho_{\geq y},
\end{aligned}    
\label{eq:extension_meaning}
\end{equation}
where $\Tr_{\cdots}$ means taking a partial trace, over the clusters with $y$-coordinate specified in the subscript.
Note that we only said that Eq.~\eqref{eq:extension_meaning} holds on a class of states \emph{we consider}. This is because the composition of the partial trace and the extensions are generally not equal to an identity map.\footnote{However, if the states are classical, \emph{i.e.,} diagonal in a fixed product basis, the composition does become the identity operation.}

Let us summarize the important identities. The first equation is Eq.~\eqref{eq:snakes_in_action1}, the solid lines in particular:

\begin{equation}
\begin{tikzcd}[column sep=tiny]
& \text{Level-$2$ snake at } y 
\arrow[dl,leftharpoondown,shift right=0.25ex, "\mathcal{E}_{y,\uparrow}"'] 
\arrow[dl,rightharpoonup,shift left=0.25ex, "\Tr_{y\shortplus 1}"]
\arrow[dr,leftharpoonup,shift left=0.25ex, "\Tr_y"'] 
\arrow[dr,rightharpoondown,shift right=0.25ex, "\mathcal{E}_{y\shortplus 1, \downarrow}"]
& \\
\text{Level-$1$ snake at } y
\arrow[rr,leftharpoondown,dashed,shift right=0.25ex, "t_y^{-1}"']
\arrow[rr,rightharpoonup,dashed,shift left=0.25ex, "t_y"]
& & \text{Level-$1$ snake at } y\shortplus 1
\end{tikzcd},
\label{eq:snakes_in_action1}
\end{equation}

where $t_y\in \mathcal{T}$ is a translation in the $y$-direction by $1$; see Section~\ref{sec:summary} for the setup.

We will also prove the following \emph{twist identity}:
\begin{equation}
    \mathcal{E}_{y,\uparrow}\left(\adulthorizontalsnake{1}{y\shortminus 1}{N}{y} \right) = \mathcal{E}_{y,\downarrow} \left(\adulthorizontalsnake{1}{y}{N}{y\shortplus 1} \right).
    \label{eq:snakes_in_action2}
\end{equation}

More formally, the map from the level-$2$ snakes to the level-$1$ snakes is encapsulated by the following proposition.
\begin{restatable}[]{proposition}{twotoone}
\begin{equation}
\begin{aligned}
&\Tr_y \left(\adulthorizontalsnake{1}{y}{N}{y\shortplus 1}\right) = \horizontalbaby{(1,y\shortplus 1)}{(N, y\shortplus 1)}, \\
&\Tr_{y+1} \left(\adulthorizontalsnake{1}{y}{N}{y\shortplus 1}\right) = \horizontalbaby{(1,y)}{(N, y)}.
\end{aligned}
\end{equation}
\label{prop:2to1}
\end{restatable}
\noindent
The map from the level-$1$ snakes to the level-$2$ snakes can be summarized as follows.
\begin{restatable}[]{proposition}{onetotwo}
\begin{equation}
    \begin{aligned}
    \adulthorizontalsnake{1}{y}{N}{y\shortplus 1} &= \mathcal{E}_{y, \uparrow} \left(\horizontalbaby{(1,y)}{(N,y)} \right)  \\
    &= \mathcal{E}_{y\shortplus 1, \downarrow} \left(\horizontalbaby{(1,y\shortplus 1)}{(N,y\shortplus 1)} \right).
    \end{aligned}
\end{equation}
\label{prop:onetotwo}
\end{restatable}
Lastly, we have the following proposition.
\begin{restatable}[]{proposition}{twist}
\begin{equation}
    \mathcal{E}_{y,\uparrow}\left(\adulthorizontalsnake{1}{y\shortminus 1}{N}{y} \right)
    =
    \mathcal{E}_{y, \downarrow}\left(\adulthorizontalsnake{1}{y}{N}{y\shortplus 1} \right).\label{eq:twist_main_result}
\end{equation}
\label{prop:twist}
\end{restatable}
\noindent
The proofs of these statements are quite technical. We will explain them in great detail in Appendix~\ref{appendix:2to1},~\ref{appendix:1to2}, and~\ref{appendix:twist}.

These propositions are important because, as a whole, they imply that one can form a snake from a sequence of level-$2$ snakes, as we explain in Section~\ref{sec:consistency}.

\section{Main results}
\label{sec:main_results}
From Proposition~\ref{prop:2to1},~\ref{prop:onetotwo}, and~\ref{prop:twist}, our main claims follow immediately, as we explain below.

\subsection{Consistency}
\label{sec:consistency}
Suppose \miniblockfull and \mymarginal obey $\mathcal{C}_L$ and $\mathcal{C}_{M,P}$. Then, we can show that \miniblockfull is consistent with a translationally invariant state on an infinite lattice.\footnote{We expect the same to be true for the marginal on the $3\times 3$ cluster as well. However, we do not attempt to prove this in this paper.} Our approach is to build up a global state gradually, from the $2\times 2$ clusters to a level-$2$ snake, and then a snake made out of the level-$2$ snakes. Schematically, we have
\begin{equation}
\begin{aligned}
     \twobytwounit{(x,y)}  &\longrightarrow 
    \adulthorizontalsnake{1}{y}{N}{y\shortplus 1} \\
    &\longrightarrow
    \thicksnake{1}{1}{N}{M}.
\end{aligned}
\end{equation}

An important observation is that two level-$2$ snakes that overlap with each other can be merged together.
\begin{theorem}
\label{thm:forgotten}
\begin{equation}
\begin{aligned}
    \adulthorizontalsnake{1}{y\shortminus 1}{N}{y} \maxmerge 
    \adulthorizontalsnake{1}{y}{N}{y\shortplus 1}\\
    =
    \adulthorizontalsnake{1}{y\shortminus 1}{N}{y} \rightmerge 
    \adulthorizontalsnake{1}{y}{N}{y\shortplus 1}.
\end{aligned}
\end{equation}
\end{theorem}
\begin{proof}
\begin{equation}
\begin{aligned}
&\Tr_{(y+1)\text{'th row}}\left(\mathcal{E}_{y,\uparrow}\left(
\adulthorizontalsnake{1}{y\shortminus 1}{N}{y}
\right)\right)\\
&= \Tr_{(y+1)\text{'th row}}\left(\mathcal{E}_{y,\downarrow}\left(
\adulthorizontalsnake{1}{y}{N}{y\shortplus 1}\right)\right)\\
&= \mathcal{E}_{y,\downarrow}\left(
\horizontalbaby{(1,y)}{(N,y)}\right) \\
&= \adulthorizontalsnake{1}{y\shortminus 1}{N}{y},
\end{aligned}
\end{equation}
using Proposition~\ref{prop:twist}, Proposition~\ref{prop:2to1} (second identity), and Proposition~\ref{prop:onetotwo}.

By the first identity in Proposition~\ref{prop:2to1}, 
\begin{equation}
\begin{aligned}
  \Tr_{(y-1)\text{'th row}}\left(\mathcal{E}_{y,\uparrow}\left(
\adulthorizontalsnake{1}{y\shortminus 1}{N}{y}
\right)\right) \\
= \adulthorizontalsnake{1}{y}{N}{y\shortplus 1}.
\end{aligned}
\end{equation}

Thus, the sixth condition of Lemma~\ref{lemma:fundamental} is satisfied, with the following choice of $\lambda, \sigma,$ and $\Phi$:
\begin{equation}
\begin{aligned}
    \lambda &= \adulthorizontalsnake{1}{y\shortminus 1}{N}{y}, \\
    \sigma &= \adulthorizontalsnake{1}{y}{N}{y\shortplus 1},  \text{ and }\\
    \Phi &= \mathcal{E}_{y,\uparrow}.
\end{aligned}
\label{eq:lambsigphi}
\end{equation}
The third condition of Lemma~\ref{lemma:fundamental} implies the main claim.
\end{proof}

It follows that we can define the following snake:
\begin{equation}
    \mathbb{S}\left( \left(\adulthorizontalsnake{1}{y}{N}{y\shortplus 1} \right)_{y=1}^{M-1} \right). \label{eq:global_state}
\end{equation}
Importantly, this object is not $\nil$. (See Definition~\ref{definition:snake}.) Therefore, this snake must be consistent with 
\begin{equation}
\adulthorizontalsnake{1}{y}{N}{y\shortplus 1}
\end{equation}
for all $0<y\leq N-1$. It then follows
\begin{equation}
    \mathbb{S}\left( \left(\adulthorizontalsnake{1}{y}{N}{y\shortplus 1} \right)_{y=1}^{M-1} \right) \consistent \twobytwounit{(x,y)}\label{eq:main1_final}
\end{equation}
for all $x,y \in \mathbb{Z}^+$; the consistencies of the $2\times 2$ clusters on the boundary invoke $\mathcal{C}_L$.

\begin{restatable}[]{theorem}{theoremone}
For every \miniblockfull and \mymarginal that satisfy $\mathcal{C}_L$ and $\mathcal{C}_{M,P}$, for every integer $N, M\geq 2$, there is a density matrix over $N\times M$ cluster(see Fig.~\ref{fig:lattice}) that is consistent with \miniblockfull over every $2\times 2$ cluster.
\label{thm:main1}
\end{restatable}

\begin{remark}
The proof of Theorem~\ref{thm:main1} is constructive, as one can see in Eq.~\eqref{eq:main1_final}.
\end{remark}

\subsection{Entropy}
\label{sec:entropy}
In this Section, we derive an \emph{exact} formula for the maximum entropy consistent with $\centeredTikZmini{\squaresempty{2}{2}{0}{0};}$, subject to the constraints $\mathcal{C}_L$ and $\mathcal{C}_{M,P}$. We will proceed in two steps. First, we will compute the entropy of the snakes made out of the level-$2$ snakes. Using Corollary~\ref{corollary:snake_entropy}, this calculation becomes straightforward. Second, we will prove that this entropy is in fact the \emph{maximum} entropy consistent with $\centeredTikZmini{\squaresempty{2}{2}{0}{0};}$. This way, we derive an expression for the maximum entropy.
\begin{theorem}
\begin{equation}
\begin{aligned}
    &\entropy{\mathbb{S}\left( \left(\adulthorizontalsnake{1}{y}{N}{y\shortplus 1} \right)_{y=1}^{M-1} \right)} \\
    &=(M-1)(N-1)\entropy{\centeredTikZ{\squaresempty{2}{2}{0}{0};}} - (M-1)(N-2)\entropy{\centeredTikZ{\squaresempty{1}{2}{0}{0};}}\\
    &-(M-2)(N-1)\entropy{\centeredTikZ{\squaresempty{2}{1}{0}{0};}}
    +(M-2)(N-2)\entropy{\centeredTikZ{\squaresempty{1}{1}{0}{0};}}
\end{aligned}
\end{equation}
\label{thm:entropy_maxmerge}
\end{theorem}
\begin{proof}
From Corollary~\ref{corollary:snake_entropy},
\begin{equation}
\begin{aligned}
    &\entropy{\mathbb{S}\left( \left(\adulthorizontalsnake{1}{y}{N}{y\shortplus 1} \right)_{y=1}^{M-1} \right)}\\
    &= (M-1)\entropy{\adulthorizontalsnake{1}{1}{N}{2}} \\
    &- (M-2) \entropy{\horizontalbaby{(1,y)}{(N,y)}}.
\end{aligned}
\end{equation}
Therefore, again from Corollary~\ref{corollary:snake_entropy},
\begin{equation}
\begin{aligned}
    &\entropy{\mathbb{S}\left( \left(\adulthorizontalsnake{1}{y}{N}{y\shortplus 1} \right)_{y=1}^{M-1} \right)}\\
    &=(M-1)(N-1)\entropy{\centeredTikZ{\squaresempty{2}{2}{0}{0};}} - (M-1)(N-2)\entropy{\centeredTikZ{\squaresempty{1}{2}{0}{0};}}\\
    &-(M-2)(N-1)\entropy{\centeredTikZ{\squaresempty{2}{1}{0}{0};}}
    +(M-2)(N-2)\entropy{\centeredTikZ{\squaresempty{1}{1}{0}{0};}}.
\end{aligned}
\end{equation}
\end{proof}

Now, let us derive an upper bound on the entropy that matches Theorem~\ref{thm:entropy_maxmerge}. First, consider a density matrix on the first two rows, denoted as $\rho[2]$. By repeatedly using SSA, we obtain
\begin{equation}
    \entropy{\rho[2]} \leq (N-1)\entropy{\centeredTikZ{\squaresempty{2}{2}{0}{0};}} - (N-2) \entropy{\centeredTikZ{\squaresempty{1}{2}{0}{0};}}.
\end{equation}
More generally, let $\rho[k]$ be a density matrix over the first $k$ rows. We can derive the following recursive inequality:
\begin{equation}
\begin{aligned}
    \entropy{\rho[k+1]} &\leq \entropy{\rho[k]} + (N-1)\entropy{\centeredTikZ{\squaresempty{2}{2}{0}{0};}} \\
    &- (N-2)\entropy{\centeredTikZ{\squaresempty{2}{1}{0}{0};\squaresempty{1}{1}{0}{1};}} - \entropy{\centeredTikZ{\squaresempty{2}{1}{0}{0};}}.
\end{aligned}
\end{equation}
Therefore, we can obtain the following bound. Using \hyperref[constraints:snake]{Snake-a) constraint}, for any density matrix over the first $M$ rows,
\begin{equation}
\begin{aligned}
    \entropy{\rho[M]}&\leq (M-1)(N-1)\entropy{\centeredTikZ{\squaresempty{2}{2}{0}{0};}} \\ &- (M-1)(N-2)\entropy{\centeredTikZ{\squaresempty{1}{2}{0}{0};}}\\
    &-(M-2)(N-1)\entropy{\centeredTikZ{\squaresempty{2}{1}{0}{0};}}\\
    &+(M-2)(N-2)\entropy{\centeredTikZ{\squaresempty{1}{1}{0}{0};}}.
\end{aligned}
\end{equation}

Thus, we have proved Theorem~\ref{thm:main2}, stated below.
\begin{restatable}[]{theorem}{theoremtwo}
Consider a family of density matrices acting on a $N\times M$ cluster (see Fig.~\ref{fig:lattice}) which are consistent with \miniblockfull obeying $\mathcal{C}_L$ and $\mathcal{C}_{M,P}$. The maximum entropy within this family is 
\begin{equation}
\begin{aligned}
   &(N-1)(M-1)S\left(\centeredTikZ{\squaresempty{2}{2}{0}{0};} \right) \\ &+ 
    (N-2)(M-2)S\left( 
    \centeredTikZ{
    \foreach \x in {0}
    {
    \foreach \y in {0}
    {
    \emptysquare{\x*0.5+\y*0.25}{\y*0.5};
    }
    }
    }
    \right) \\
    &-
    (N-2)(M-1)
    S\left( 
    \centeredTikZ{
    \foreach \x in {0}
    {
    \foreach \y in {0,1}
    {
    \emptysquare{\x*0.5-\y*0.25}{\y*0.5};
    }
    }
    }\right)
    \\&- (N-1)(M-2)
    S\left( 
    \centeredTikZ{
    \foreach \x in {1,0}
    {
    \foreach \y in {0}
    {
    \emptysquare{\x*0.5+\y*0.25}{\y*0.5};
    }
    }
    }
    \right).
\end{aligned}
\end{equation}
\label{thm:main2}
\end{restatable}

We also obtain the following expression for the maximum entropy density in the thermodynamic limit.
\begin{corollary}
For any \miniblockfull obeying $\mathcal{C}_L$ and $\mathcal{C}_{M,P}$,
\begin{equation}
\begin{aligned}
    \lim_{N,M\to \infty} \left( \max_{\tilde{\rho}_{N\times M} \consistent \mathcal{T}\left(\left\{\centeredTikZmini{\squaresempty{2}{2}{0}{0};} \right\} \right)}\left(\frac{\entropy{\tilde{\rho}_{N\times M}}}{NM}\right) \right)
    \\= 
    \entropy{\centeredTikZ{\squaresempty{2}{2}{0}{0};}}
    -\entropy{\centeredTikZ{\squaresempty{2}{1}{0}{0};\squaresempty{1}{1}{0}{1};}}.
\end{aligned}
\end{equation}
\end{corollary}

\subsection{Beyond mean-field}
From Theorem~\ref{thm:main1} and~\ref{thm:main2}, so long as $\mathcal{C}_L$ and $\mathcal{C}_{M,P}$ are satisfied \emph{exactly}, one can compute an upper bound to the global free energy. By minimizing this upper bound over a family of marginals that satisfy these constraints, we can obtain a variational upper bound to the global free energy. In this section, we emphasize that this upper bound must be better than mean-field theory in some sense.

At least for translationally-invariant systems, our family of marginals include the mean-field solutions. This is because mean-field states have a tensor product structure and any states that has a tensor product structure automatically satisfies $\mathcal{C}_{M,P}$.

However, there are states that satisfy $\mathcal{C}_{M,P}$ that cannot be written as a product state. A notable example is the ground states of Levin-Wen model~\cite{Levin2005} for which Eq.~\eqref{eq:ent_scaling} holds \emph{exactly}~\cite{Levin2006,Kitaev2006} locally everywhere. Therefore, the reduced density matrices of the ground state of the Levin-Wen model must obey $\mathcal{C}_{M,P}$ exactly. Therefore, our class of density matrices include local reduced density matrices of topologically ordered systems, which cannot be mapped from a product state by a finite-depth quantum circuit~\cite{Bravyi2006}. For concreteness, we have added an explicit calculation for Kitaev's toric code~\cite{Kitaev2002} in Appendix~\ref{appendix:examples}.

Let us end with a side remark. While Theorem~\ref{thm:main1} does prove that the reduced density matrices of Levin-Wen models are consistent with some global state on an infinite lattice, it \emph{does not} prove that they are consistent with some global state on a torus. That would require a periodic boundary condition, which is incompatible with the definition of the snake; see Definition~\ref{definition:snake}. However, the state itself is still long-range entangled; see Ref.~\cite{SKK2019}.

\section{Discussion}
\label{sec:discussion}
In this paper, we have put forward a conjectured duality between many-body quantum states obeying Eq.~\eqref{eq:ent_scaling} globally and the set of density matrices obeying the same equation \emph{locally}. If true, in numerical studies of systems that are expected to obey Eq.~\eqref{eq:ent_scaling}, one can take a shortcut and just minimize the energy and the free energy of the ansatz over a set of density matrices on bounded regions that obey Eq.~\eqref{eq:ent_scaling}. The number of parameters in the latter case is polynomial in the system size, an exponential improvement over a naive approach in which the minimization is taken over exponentially many parameters. 

We provided nontrivial evidences for this conjecture, by proving this conjecture in the context of translationally invariant systems. This has led to a nontrivial upper bound to the ground state energy and the thermodynamic free energy of generic interacting quantum many-body system, which subsumes the mean-field theory bound. 

While there are systems in two spatial dimensions that break Eq.~\eqref{eq:ent_scaling}~\cite{Zou2016,Williamson2019}, the requisite entropy scaling law can be restored by applying a finite-depth quantum circuit. Viewed this way, one can (in principle) simply redefine the Hamiltonian so that Eq.~\eqref{eq:ent_scaling} holds, at which point the machinery of this paper again applies. We are unaware of any solvable models of gapped two-dimensional quantum many-body systems which cannot be accommodated using this approach. 

This is not to say that two-dimensional quantum many-body systems is solved, of course. In order to utilize our result, one must actually find a practical numerical method that can minimize the energy and the free energy in the family of states that satisfy $\mathcal{C}_L$ and $\mathcal{C}_{M,P}$. The former set of constraints are linear, and as such, more benign. On the other hand, the latter set of constraints are non-linear, making it somewhat difficult to make progress on.

Moreover, realistic systems will not satisfy Eq.~\eqref{eq:ent_scaling} exactly. The effect of this approximation error can be studied using the approach of Ref.~\cite{Kim2016}. However, for practical purposes, it will be desirable to tighten this bound.

Overcoming these challenges will be an important step we need to take to make our approach practical. We end with some thoughts on how to extend the current result.
\begin{enumerate}
    \item While we have focused on quantum spins, an analogous statement for fermions are expected to hold. The main question is whether the statements in Section~\ref{sec:merging_algebra} can be reproduced. A fermionic version of SSA would follow from the monotonicity of quantum relative entropy~\cite{Uhlmann1977} and its equality condition was investigated by Petz~\cite{Petz1988}. In this setup, the partial trace operation needs to be replaced by a conditional expectation.
    \item The maximum-entropy state obtained in this paper, restricted to a quasi-one-dimensional cluster extending in the $\hat{x}$-direction, becomes a snake constructed from a set of marginals that form a ``chain.'' This is an instance of the Markovian matrix product density operator~\cite{Kim2017}, which admits an exact contraction for any $2$-point correlation function along the chain. One may be able to prove a similar statement for quasi-one-dimensional cluster in \emph{any} direction. The readers are encouraged to verify that this is true in the $\hat{y}$-direction.
    \item One may be able to simplify $\mathcal{C}_{M,P}$ further. For classical states, the last condition of $\mathcal{C}_{M,P}$ is implied by the first two conditions, reducing the number of nonlinear-constraints to two and also reducing the size of the maximal cluster one needs to consider. To what extent we can reduce these constraints is an important open problem.
    \item If we have a set of translationally invariant marginals that satisfy $\mathcal{C}_{M,P}$ individually, their convex combinations are also consistent with a convex combination of translationally invariant states on an infinite lattice, each of which are the maximum-entropy extensions of the marginals. This is true even if the convex combination of the marginals do not satisfy $\mathcal{C}_{M,P}$ as a whole. Moreover, we can compute the the convex combination of the extensions if the set is finite and if the maximum-entropy extensions of the marginals on an infinite lattice are mutually orthogonal to each other.\footnote{More precisely, the overlap should vanish faster than the inverse of the system volume as we take the thermodynamic limit.} If these conditions are met, the entropy density converges to a convex combination of the individual maximum entropy densities in the thermodynamic limit.\footnote{However, the convex combination of the extensions may not be the maximum-entropy state consistent with the convex combination of the marginals.} This way, one can relax the condition $\mathcal{C}_{M,P}$.
\end{enumerate}

\section*{Acknowledgement}
This work was supported by the Australian Research Council via the Centre of Excellence in Engineered Quantum Systems (EQUS) project number CE170100009.

\appendix

\section{Descendants}
\label{appendix:descendants}
In this Appendix, we derive the identities in Table~\ref{table:constraints} from $\mathcal{C}_{M,P}$. To start with, the following identity will be useful.
\begin{equation}
\begin{aligned}
&4\entropy{\centeredTikZ{\squaresempty{2}{2}{0}{0};}}
    -
    2\entropy{\centeredTikZ{\squaresempty{2}{1}{0}{0};}}
    -
    2\entropy{\centeredTikZ{\squaresempty{1}{2}{0}{0};}}
    +
    \entropy{\centeredTikZ{\squaresempty{1}{1}{0}{0};}}
    -
    \entropy{\centeredTikZ{\squaresempty{3}{3}{0}{0};}}\\
    &=\left(2\entropy{\centeredTikZ{\squaresempty{2}{2}{0}{0};}} - \entropy{\centeredTikZ{\squaresempty{1}{2}{0}{0};}}
    - \entropy{\centeredTikZ{\squaresempty{3}{2}{0}{0};}}\right) \\
    &+ \left(\entropy{ \centeredTikZ{\squaresempty{3}{2}{0}{0};}}+ 
    \entropy{ \centeredTikZ{\squaresempty{2}{2}{0}{0};}}
    -
    \entropy{ \centeredTikZ{\squaresempty{2}{1}{0}{0};}}
    -
    \entropy{ \centeredTikZ{\squaresempty{3}{2}{0}{0};\squaresempty{2}{1}{0}{2}}}\right)\\
    &+
    \left(
    \entropy{ \centeredTikZ{\squaresempty{3}{2}{0}{0};\squaresempty{2}{1}{0}{2}}}
    +
    \entropy{ \centeredTikZ{\squaresempty{2}{2}{0}{0};}}
    -
    \entropy{ \centeredTikZ{\squaresempty{2}{1}{0}{0};\squaresempty{1}{1}{0}{1}}}
    -
    \entropy{ \centeredTikZ{\squaresempty{3}{3}{0}{0};}}\right) \\
    &-
    \left(
    \entropy{\centeredTikZ{\squaresempty{2}{1}{0}{0};}}
    +
    \entropy{\centeredTikZ{\squaresempty{1}{2}{0}{0};}}
    -
    \entropy{\centeredTikZ{\squaresempty{1}{1}{0}{0};}}
    -
    \entropy{\centeredTikZ{\squaresempty{2}{1}{0}{0};\squaresempty{1}{1}{0}{1}}}
    \right).
\end{aligned}
\end{equation}
Note that the last term vanishes because of Eq.~\eqref{constraint:primary1}. Moreover, the terms in the paranthesis on line $2, 3,$ and $4$ must be nonnegative because of SSA. Because the overall sum must be $0$ by Eq.~\eqref{constraint:primary3}, each of these terms must vanish. Moreover, we can alternatively consider the following decomposition:
\begin{equation}
\begin{aligned}
&4\entropy{\centeredTikZ{\squaresempty{2}{2}{0}{0};}}
    -
    2\entropy{\centeredTikZ{\squaresempty{2}{1}{0}{0};}}
    -
    2\entropy{\centeredTikZ{\squaresempty{1}{2}{0}{0};}}
    +
    \entropy{\centeredTikZ{\squaresempty{1}{1}{0}{0};}}
    -
    \entropy{\centeredTikZ{\squaresempty{3}{3}{0}{0};}}\\
    &=\left(2\entropy{\centeredTikZ{\squaresempty{2}{2}{0}{0};}} - \entropy{\centeredTikZ{\squaresempty{2}{1}{0}{0};}}
    - \entropy{\centeredTikZ{\squaresempty{2}{3}{0}{0};}}\right) \\
    &+ \left(\entropy{ \centeredTikZ{\squaresempty{2}{3}{0}{0};}}+ 
    \entropy{ \centeredTikZ{\squaresempty{2}{2}{0}{0};}}
    -
    \entropy{ \centeredTikZ{\squaresempty{1}{2}{0}{0};}}
    -
    \entropy{ \centeredTikZ{\squaresempty{3}{2}{0}{0};\squaresempty{2}{1}{0}{2}}}\right)\\
    &+
    \left(
    \entropy{ \centeredTikZ{\squaresempty{3}{2}{0}{0};\squaresempty{2}{1}{0}{2}}}
    +
    \entropy{ \centeredTikZ{\squaresempty{2}{2}{0}{0};}}
    -
    \entropy{ \centeredTikZ{\squaresempty{2}{1}{0}{0};\squaresempty{1}{1}{0}{1}}}
    -
    \entropy{ \centeredTikZ{\squaresempty{3}{3}{0}{0};}}\right) \\
    &-
    \left(
    \entropy{\centeredTikZ{\squaresempty{2}{1}{0}{0};}}
    +
    \entropy{\centeredTikZ{\squaresempty{1}{2}{0}{0};}}
    -
    \entropy{\centeredTikZ{\squaresempty{1}{1}{0}{0};}}
    -
    \entropy{\centeredTikZ{\squaresempty{2}{1}{0}{0};\squaresempty{1}{1}{0}{1}}}
    \right).
\end{aligned}
\end{equation}
Again, the last term vanishes and the remaining terms must vanish. We can summarize these identities as follows.

\begin{lemma}
If \mymarginal satisfies $\mathcal{C}_L$ and $\mathcal{C}_{M,P}$, 
\begin{equation}
\begin{aligned}
    \entropy{ \centeredTikZ{\squaresempty{3}{2}{0}{0};}} &= 2\entropy{ \centeredTikZ{\squaresempty{2}{2}{0}{0};}} - \entropy{ \centeredTikZ{\squaresempty{1}{2}{0}{0};}},\\
    \entropy{ \centeredTikZ{\squaresempty{2}{3}{0}{0};}} &= 2\entropy{ \centeredTikZ{\squaresempty{2}{2}{0}{0};}} - \entropy{ \centeredTikZ{\squaresempty{2}{1}{0}{0};}},\\
    \entropy{ \centeredTikZ{\squaresempty{2}{3}{0}{0};\squaresempty{1}{2}{2}{0}}} &=\entropy{ \centeredTikZ{\squaresempty{3}{2}{0}{0};}}+ 
    \entropy{ \centeredTikZ{\squaresempty{2}{2}{0}{0};}}
    -
    \entropy{ \centeredTikZ{\squaresempty{2}{1}{0}{0};}},\\
    \entropy{ \centeredTikZ{\squaresempty{2}{3}{0}{0};\squaresempty{1}{2}{2}{0}}} &=\entropy{ \centeredTikZ{\squaresempty{2}{3}{0}{0};}}+ 
    \entropy{ \centeredTikZ{\squaresempty{2}{2}{0}{0};}}
    -
    \entropy{ \centeredTikZ{\squaresempty{1}{2}{0}{0};}},
    \\
    \entropy{ \centeredTikZ{\squaresempty{3}{3}{0}{0};}}&=
    \entropy{ \centeredTikZ{\squaresempty{3}{2}{0}{0};\squaresempty{2}{1}{0}{2}}}
    +
    \entropy{ \centeredTikZ{\squaresempty{2}{2}{0}{0};}}
    -
    \entropy{ \centeredTikZ{\squaresempty{2}{1}{0}{0};\squaresempty{1}{1}{0}{1}}}.
\end{aligned}
\end{equation}
\label{lemma:descendant1}
\end{lemma}

Similarly, we can consider the following identity.

\begin{equation}
\begin{aligned}
&4\entropy{ \centeredTikZ{\squaresempty{2}{2}{0}{0};}}
    -
    2\entropy{ \centeredTikZ{\squaresempty{2}{1}{0}{0};}}
    -
    2\entropy{ \centeredTikZ{\squaresempty{1}{2}{0}{0};}}
    +
    \entropy{ \centeredTikZ{\squaresempty{1}{1}{0}{0};}}
    -
    \entropy{\centeredTikZ{\squaresempty{3}{3}{0}{0};} }\\
    &=\left(2\entropy{ \centeredTikZ{\squaresempty{2}{2}{0}{0};}} - \entropy{ \centeredTikZ{\squaresempty{2}{1}{0}{0};}}
    - \entropy{ \centeredTikZ{\squaresempty{2}{3}{0}{0};}}\right) \\
    &+ \left(\entropy{ \centeredTikZ{\squaresempty{2}{3}{0}{0};}}+ 
    \entropy{ \centeredTikZ{\squaresempty{2}{2}{0}{0};}}
    -
    \entropy{ \centeredTikZ{\squaresempty{1}{2}{0}{0};}}
    -
    \entropy{ \centeredTikZ{\squaresempty{3}{2}{0}{1};\squaresempty{2}{1}{1}{0}}}\right)\\
    &+
    \left(
    \entropy{ \centeredTikZ{\squaresempty{3}{2}{0}{1};\squaresempty{2}{1}{1}{0}}}
    +
    \entropy{ \centeredTikZ{\squaresempty{2}{2}{0}{0};}}
    -
    \entropy{ \centeredTikZ{\squaresempty{2}{1}{0}{0};\squaresempty{1}{1}{1}{-1}}}
    -
    \entropy{ \centeredTikZ{\squaresempty{3}{3}{0}{0};}}\right) \\
    &-
    \left(
    \entropy{\centeredTikZ{\squaresempty{2}{1}{0}{0};} }
    +
    \entropy{\centeredTikZ{\squaresempty{1}{2}{0}{0};} }
    -
    \entropy{\centeredTikZ{\squaresempty{1}{1}{0}{0};} }
    -
    \entropy{\centeredTikZ{\squaresempty{2}{1}{0}{0};\squaresempty{1}{1}{1}{-1}} }
    \right).
\end{aligned}
\end{equation}
and
\begin{equation}
\begin{aligned}
&4\entropy{ \centeredTikZ{\squaresempty{2}{2}{0}{0};}}
    -
    2\entropy{ \centeredTikZ{\squaresempty{2}{1}{0}{0};}}
    -
    2\entropy{ \centeredTikZ{\squaresempty{1}{2}{0}{0};}}
    +
    \entropy{ \centeredTikZ{\squaresempty{1}{1}{0}{0};}}
    -
    \entropy{\centeredTikZ{\squaresempty{3}{3}{0}{0};} }\\
    &=\left(2\entropy{ \centeredTikZ{\squaresempty{2}{2}{0}{0};}} - \entropy{ \centeredTikZ{\squaresempty{1}{2}{0}{0};}}
    - \entropy{ \centeredTikZ{\squaresempty{3}{2}{0}{0};}}\right) \\
    &+ \left(\entropy{ \centeredTikZ{\squaresempty{3}{2}{0}{0};}}+ 
    \entropy{ \centeredTikZ{\squaresempty{2}{2}{0}{0};}}
    -
    \entropy{ \centeredTikZ{\squaresempty{2}{1}{0}{0};}}
    -
    \entropy{ \centeredTikZ{\squaresempty{3}{2}{0}{1};\squaresempty{2}{1}{1}{0}}}\right)\\
    &+
    \left(
    \entropy{ \centeredTikZ{\squaresempty{3}{2}{0}{1};\squaresempty{2}{1}{1}{0}}}
    +
    \entropy{ \centeredTikZ{\squaresempty{2}{2}{0}{0};}}
    -
    \entropy{ \centeredTikZ{\squaresempty{2}{1}{0}{0};\squaresempty{1}{1}{1}{-1}}}
    -
    \entropy{ \centeredTikZ{\squaresempty{3}{3}{0}{0};}}\right) \\
    &-
    \left(
    \entropy{\centeredTikZ{\squaresempty{2}{1}{0}{0};} }
    +
    \entropy{\centeredTikZ{\squaresempty{1}{2}{0}{0};} }
    -
    \entropy{\centeredTikZ{\squaresempty{1}{1}{0}{0};} }
    -
    \entropy{\centeredTikZ{\squaresempty{2}{1}{0}{0};\squaresempty{1}{1}{1}{-1}} }
    \right).
\end{aligned}
\end{equation}

Like before, the last term vanishes because of Eq.~\eqref{constraint:primary2}. Moreover, the terms in the paranthesis on line $2,3,$ and $4$ are all nonnegative because of SSA. Because the overall sum is $0$ by Eq.~\eqref{constraint:primary3}, each of these terms must vanish. Neglecting the ones that have already appeared in Lemma~\ref{lemma:descendant1}, we arrive at the following lemma.
\begin{lemma}
If \mymarginal satisfies $\mathcal{C}_L$ and $\mathcal{C}_{M,P}$,
\begin{equation}
\begin{aligned}
    \entropy{ \centeredTikZ{\squaresempty{3}{2}{0}{1};\squaresempty{2}{1}{1}{0}}} &=\entropy{ \centeredTikZ{\squaresempty{2}{3}{0}{0};}}+ 
    \entropy{ \centeredTikZ{\squaresempty{2}{2}{0}{0};}}
    -
    \entropy{ \centeredTikZ{\squaresempty{1}{2}{0}{0};}},\\
    \entropy{ \centeredTikZ{\squaresempty{3}{2}{0}{1};\squaresempty{2}{1}{1}{0}}} &=\entropy{ \centeredTikZ{\squaresempty{3}{2}{0}{0};}}+ 
    \entropy{ \centeredTikZ{\squaresempty{2}{2}{0}{0};}}
    -
    \entropy{ \centeredTikZ{\squaresempty{2}{1}{0}{0};}},\\
    \entropy{ \centeredTikZ{\squaresempty{3}{3}{0}{0};}}&=
    \entropy{ \centeredTikZ{\squaresempty{3}{2}{0}{1};\squaresempty{2}{1}{1}{0}}}
    +
    \entropy{ \centeredTikZ{\squaresempty{2}{2}{0}{0};}}
    -
    \entropy{ \centeredTikZ{\squaresempty{2}{1}{0}{0};\squaresempty{1}{1}{1}{-1}}}.
\end{aligned}
\end{equation}
\label{lemma:descendant2}
\end{lemma}

We will also need the following identities:
\begin{lemma}
If \mymarginal satisfies $\mathcal{C}_L$ and $\mathcal{C}_{M,P}$,
\begin{equation}
\begin{aligned}
\entropy{\centeredTikZ{
\squaresempty{3}{1}{0}{0}
}} = 
2\entropy{\centeredTikZ{
\squaresempty{2}{1}{0}{0}
}}
-
\entropy{\centeredTikZ{
\squaresempty{1}{1}{0}{0}
}}, \\
\entropy{\centeredTikZ{
\squaresempty{1}{3}{0}{0}
}} = 
2\entropy{\centeredTikZ{
\squaresempty{1}{2}{0}{0}
}}
-
\entropy{\centeredTikZ{
\squaresempty{1}{1}{0}{0}
}}.
\end{aligned}
\end{equation}
\label{lemma:descendant3}
\end{lemma}
\begin{proof}
The first equation can be derived as follows. Let us rearrange the first result within Lemma~\ref{lemma:descendant1} as follows:
\begin{equation}
\entropy{\centeredTikZ{\squaresempty{2}{2}{0}{0};\squaresfilled{1}{2}{2}{0};}}
+
\entropy{\centeredTikZ{\squaresempty{2}{2}{1}{0};\squaresfilled{1}{2}{0}{0}}}
-
\entropy{\centeredTikZ{\squaresfilled{3}{2}{0}{0};\squaresempty{1}{2}{1}{0};}}
-
\entropy{\centeredTikZ{\squaresempty{3}{2}{0}{0}}} = 0.
\end{equation}
Interpreting this as a conditional mutual information between \centeredTikZmini{\squaresfilled{3}{2}{0}{0};\squaresempty{1}{2}{0}{0};} and \centeredTikZmini{\squaresfilled{3}{2}{0}{0};\squaresempty{1}{2}{2}{0};} conditioned on \centeredTikZmini{\squaresfilled{3}{2}{0}{0};\squaresempty{1}{2}{1}{0};} being $0$, we can apply Eq.~\eqref{eq:ssa_monotinicity_intro} to obtain the following result:
\begin{equation}
\begin{aligned}
 &\entropy{\centeredTikZ{\squaresempty{2}{2}{0}{0};\squaresfilled{1}{2}{2}{0};\squaresfilled{1}{1}{0}{0};}}
+
\entropy{\centeredTikZ{\squaresempty{2}{2}{1}{0};\squaresfilled{1}{2}{0}{0}}}
-
\entropy{\centeredTikZ{\squaresfilled{3}{2}{0}{0};\squaresempty{1}{2}{1}{0};}}
-
\entropy{\centeredTikZ{\squaresempty{3}{2}{0}{0};\squaresfilled{1}{1}{0}{0};}}\\
&\leq
     \entropy{\centeredTikZ{\squaresempty{2}{2}{0}{0};\squaresfilled{1}{2}{2}{0};}}
+
\entropy{\centeredTikZ{\squaresempty{2}{2}{1}{0};\squaresfilled{1}{2}{0}{0}}}
-
\entropy{\centeredTikZ{\squaresfilled{3}{2}{0}{0};\squaresempty{1}{2}{1}{0};}}
-
\entropy{\centeredTikZ{\squaresempty{3}{2}{0}{0}}} \\
&=0
\end{aligned}
\end{equation}
By SSA, the left hand side should be nonnegative. Therefore, it must be equal to $0$. Now, note that
\begin{equation}
    \entropy{\centeredTikZ{\squaresempty{2}{2}{0}{0};\squaresfilled{1}{2}{2}{0};\squaresfilled{1}{1}{0}{0};}}
    -
    \entropy{\centeredTikZ{\squaresfilled{3}{2}{0}{0};\squaresempty{1}{2}{1}{0};}}
     = 
     \entropy{\centeredTikZ{\squaresfilled{3}{2}{0}{0};\squaresempty{2}{1}{0}{1};}}-
     \entropy{\centeredTikZ{\squaresfilled{3}{2}{0}{0};\squaresempty{1}{1}{1}{1};}}.
\end{equation}
because of Eq.~\eqref{constraint:primary2}.

Therefore, we can conclude that
\begin{equation}
   \entropy{\centeredTikZ{\squaresfilled{3}{2}{0}{0};\squaresempty{2}{1}{0}{1};}}
   +
   \entropy{\centeredTikZ{\squaresfilled{3}{2}{0}{0};\squaresempty{2}{2}{1}{0}}}
   -
\entropy{\centeredTikZ{\squaresfilled{3}{2}{0}{0};\squaresempty{1}{1}{1}{1};}} 
- 
\entropy{\centeredTikZ{\squaresempty{3}{2}{0}{0};\squaresfilled{1}{1}{0}{0};}} =0.
\end{equation}
Viewing the left hand side of this equation as a conditional mutual information between \centeredTikZmini{\squaresfilled{3}{2}{0}{0};\squaresempty{1}{1}{0}{1};} and \centeredTikZmini{\squaresfilled{3}{2}{0}{0};\squaresempty{1}{1}{2}{1};\squaresempty{2}{1}{1}{0};} conditioned on \centeredTikZmini{\squaresfilled{3}{2}{0}{0};\squaresempty{1}{1}{1}{1};}, we can use Eq.~\eqref{eq:ssa_monotinicity_intro} to derive the following result:
\begin{equation}
\begin{aligned}
&\entropy{\centeredTikZ{\squaresfilled{3}{2}{0}{0};\squaresempty{2}{1}{0}{1};}}
   +
   \entropy{\centeredTikZ{\squaresfilled{3}{2}{0}{0};\squaresempty{2}{1}{1}{1}}}
   -
\entropy{\centeredTikZ{\squaresfilled{3}{2}{0}{0};\squaresempty{1}{1}{1}{1};}} 
- 
\entropy{\centeredTikZ{\squaresempty{3}{2}{0}{0};\squaresfilled{3}{1}{0}{0};}} \\
&\leq 
\entropy{\centeredTikZ{\squaresfilled{3}{2}{0}{0};\squaresempty{2}{1}{0}{1};}}
   +
   \entropy{\centeredTikZ{\squaresfilled{3}{2}{0}{0};\squaresempty{2}{2}{1}{0}}}
   -
\entropy{\centeredTikZ{\squaresfilled{3}{2}{0}{0};\squaresempty{1}{1}{1}{1};}} 
- 
\entropy{\centeredTikZ{\squaresempty{3}{2}{0}{0};\squaresfilled{1}{1}{0}{0};}}\\
&=0.
\end{aligned}    
\end{equation}
Using translational invariance, we conclude 
\begin{equation}
    \entropy{\centeredTikZ{
\squaresempty{3}{1}{0}{0}
}} = 
2\entropy{\centeredTikZ{
\squaresempty{2}{1}{0}{0}
}}
-
\entropy{\centeredTikZ{
\squaresempty{1}{1}{0}{0}
}}.
\end{equation}

The other equation can be derived in a similar way, by making the following changes to the argument above:
\begin{equation}
\begin{aligned}
    \centeredTikZ{\squaresempty{3}{2}{0}{0};} \to 
    \centeredTikZ{\squaresempty{2}{3}{0}{0};} \qquad \text{ and } \qquad  
    \centeredTikZ{\squaresfilled{3}{2}{0}{0};\squaresempty{2}{1}{0}{1};\squaresempty{1}{1}{1}{0};}
    \to 
    \centeredTikZ{\squaresfilled{2}{3}{0}{0};\squaresempty{2}{1}{0}{1};\squaresempty{1}{1}{0}{2};}.
\end{aligned}
\end{equation}
\end{proof}

\section{Level-$2$ $\to$ Level-$1$}
\label{appendix:2to1}
We prove that level-$2$ snakes become a level-$1$ snake upon taking a partial trace over either of the rows. Specifically,
\twotoone*

Below, we derive Proposition~\ref{prop:2to1} in three steps. First, we derive a set of entropy identities in Appendix~\ref{appendix:2to1_ent}. From these identities, we derive a set of merging algebra identities in Appendix~\ref{appendix:2to1_merge}. With these identities, we complete the proof in Appendix~\ref{appendix:2to1_complete}

\subsection{Entropy identities}
\label{appendix:2to1_ent}
We will need to prove the following identities:
\begin{align}
    \entropy{\centeredTikZ{\squaresempty{2}{2}{1}{0};\squaresempty{1}{1}{0}{1};}}
    &= \entropy{\centeredTikZ{\squaresempty{1}{2}{1}{0};\squaresempty{1}{1}{0}{1};}} + 
    \entropy{\centeredTikZ{\squaresempty{2}{2}{1}{0};}} - 
    \entropy{\centeredTikZ{\squaresempty{1}{2}{1}{0};}},\label{eq:prop2to1_ent1}\\
    \entropy{\centeredTikZ{\squaresempty{3}{1}{0}{1};\squaresempty{1}{1}{2}{0};}} &= \entropy{\centeredTikZ{\squaresempty{2}{1}{0}{1};\squaresempty{1}{1}{1}{0};}} + \entropy{\centeredTikZ{\squaresempty{2}{1}{0}{0};}} - \entropy{\centeredTikZ{\squaresempty{1}{1}{0}{0};}},\label{eq:prop2to1_ent2}
\end{align}
as well as their rotated versions (by $\pi$):
\begin{align}
    \entropy{\centeredTikZ{\squaresempty{3}{1}{0}{0};\squaresempty{2}{1}{0}{1};}}
    &= \entropy{\centeredTikZ{\squaresempty{2}{1}{0}{0};\squaresempty{1}{1}{0}{1};}} + 
    \entropy{\centeredTikZ{\squaresempty{2}{2}{1}{0};}} - 
    \entropy{\centeredTikZ{\squaresempty{1}{2}{1}{0};}},\label{eq:prop2to1_ent3}\\
    \entropy{\centeredTikZ{\squaresempty{3}{1}{0}{0};\squaresempty{1}{1}{0}{1};}} &= \entropy{\centeredTikZ{\squaresempty{2}{1}{0}{0};\squaresempty{1}{1}{0}{1};}} + \entropy{\centeredTikZ{\squaresempty{2}{1}{0}{0};}} - \entropy{\centeredTikZ{\squaresempty{1}{1}{0}{0};}}.\label{eq:prop2to1_ent4}
\end{align}

\begin{lemma}
\begin{equation}
    \entropy{\centeredTikZ{\squaresempty{2}{2}{1}{0};\squaresempty{1}{1}{0}{1};}}
    = \entropy{\centeredTikZ{\squaresempty{1}{2}{1}{0};\squaresempty{1}{1}{0}{1};}} + 
    \entropy{\centeredTikZ{\squaresempty{2}{2}{1}{0};}} - 
    \entropy{\centeredTikZ{\squaresempty{1}{2}{1}{0};}}.
\end{equation}
\label{lemma:prop2to1_entlemma1}
\end{lemma}
\begin{proof}
By the \hyperref[constraints:cell1]{Type I-a) cluster constraint}, 
\begin{equation}
\entropy{\centeredTikZ{\squaresempty{2}{2}{0}{0};\squaresfilled{1}{2}{2}{0};}}
+
\entropy{\centeredTikZ{\squaresempty{2}{2}{1}{0};\squaresfilled{1}{2}{0}{0}}}
-
\entropy{\centeredTikZ{\squaresfilled{3}{2}{0}{0};\squaresempty{1}{2}{1}{0};}}
-
\entropy{\centeredTikZ{\squaresempty{3}{2}{0}{0}}} = 0.
\label{eq:lemma803temp1}
\end{equation}
Interpreting this as a conditional mutual information between \centeredTikZmini{\squaresfilled{3}{2}{0}{0};\squaresempty{1}{2}{0}{0};} and \centeredTikZmini{\squaresfilled{3}{2}{0}{0};\squaresempty{1}{2}{2}{0};} conditioned on \centeredTikZmini{\squaresfilled{3}{2}{0}{0};\squaresempty{1}{2}{1}{0};} being $0$, we can apply Eq.~\eqref{eq:ssa_monotinicity_intro} to obtain the following result:
\begin{equation}
\begin{aligned}
 &\entropy{\centeredTikZ{\squaresempty{2}{2}{0}{0};\squaresfilled{1}{2}{2}{0};\squaresfilled{1}{1}{0}{0};}}
+
\entropy{\centeredTikZ{\squaresempty{2}{2}{1}{0};\squaresfilled{1}{2}{0}{0}}}
-
\entropy{\centeredTikZ{\squaresfilled{3}{2}{0}{0};\squaresempty{1}{2}{1}{0};}}
-
\entropy{\centeredTikZ{\squaresempty{3}{2}{0}{0};\squaresfilled{1}{1}{0}{0};}}\\
&\leq
     \entropy{\centeredTikZ{\squaresempty{2}{2}{0}{0};\squaresfilled{1}{2}{2}{0};}}
+
\entropy{\centeredTikZ{\squaresempty{2}{2}{1}{0};\squaresfilled{1}{2}{0}{0}}}
-
\entropy{\centeredTikZ{\squaresfilled{3}{2}{0}{0};\squaresempty{1}{2}{1}{0};}}
-
\entropy{\centeredTikZ{\squaresempty{3}{2}{0}{0}}}. \label{eq:lemma803temp0}
\end{aligned}
\end{equation}
By SSA, the first line of Eq.~\eqref{eq:lemma803temp0} must be nonnegative. By Eq.~\eqref{eq:lemma803temp1}, the second line must be $0$. Therefore, the first line must be $0$. This proves our claim.
\end{proof}

\begin{lemma}
\begin{equation}
    \entropy{\centeredTikZ{\squaresempty{3}{1}{0}{1};\squaresempty{1}{1}{2}{0};}} = \entropy{\centeredTikZ{\squaresempty{2}{1}{0}{1};\squaresempty{1}{1}{1}{0};}} + \entropy{\centeredTikZ{\squaresempty{2}{1}{0}{0};}} - \entropy{\centeredTikZ{\squaresempty{1}{1}{0}{0};}}.
\end{equation}
\label{lemma:prop2to1_entlemma2}
\end{lemma}
\begin{proof}
Using Lemma~\ref{lemma:prop2to1_entlemma1} and the \hyperref[constraints:snake]{Snake-c) constraint}, one can derive
\begin{equation}
    \entropy{\centeredTikZ{\squaresempty{2}{2}{1}{0};\squaresempty{1}{1}{0}{1};}} = \entropy{\centeredTikZ{\squaresempty{2}{2}{0}{0};}}
    +
    \entropy{\centeredTikZ{\squaresempty{2}{1}{0}{0};}}
    -
    \entropy{\centeredTikZ{\squaresempty{1}{1}{0}{0};}},
\end{equation}
which can be rewritten as
\begin{equation}
    \entropy{\centeredTikZ{\squaresempty{2}{2}{1}{0};\squaresfilled{1}{1}{0}{1};}}
    +
    \entropy{\centeredTikZ{\squaresfilled{2}{2}{1}{0};\squaresempty{2}{1}{0}{1};}}
    -
    \entropy{\centeredTikZ{\squaresfilled{2}{2}{1}{0};\squaresempty{2}{1}{0}{1};\squaresfilled{1}{1}{0}{1};}}
    -
    \entropy{\centeredTikZ{\squaresempty{2}{2}{1}{0};\squaresempty{1}{1}{0}{1};}}=0. \label{eq:lemma818temp0}
\end{equation}
Using Eq.~\eqref{eq:ssa_monotonicity_intro2}, 
\begin{equation}
    \begin{aligned}
    \entropy{\centeredTikZ{\squaresempty{2}{2}{1}{0};\squaresfilled{1}{1}{0}{1};}}
    +
    \entropy{\centeredTikZ{\squaresfilled{2}{2}{1}{0};\squaresempty{2}{1}{0}{1};\squaresempty{1}{2}{2}{0};}}
    -
    \entropy{\centeredTikZ{\squaresfilled{2}{2}{1}{0};\squaresempty{2}{1}{0}{1};\squaresfilled{1}{1}{0}{1};\squaresempty{1}{2}{2}{0};}}
    -
    \entropy{\centeredTikZ{\squaresempty{2}{2}{1}{0};\squaresempty{1}{1}{0}{1};}}\\
    \leq
    \entropy{\centeredTikZ{\squaresempty{2}{2}{1}{0};\squaresfilled{1}{1}{0}{1};}}
    +
    \entropy{\centeredTikZ{\squaresfilled{2}{2}{1}{0};\squaresempty{2}{1}{0}{1};}}
    -
    \entropy{\centeredTikZ{\squaresfilled{2}{2}{1}{0};\squaresempty{2}{1}{0}{1};\squaresfilled{1}{1}{0}{1};}}
    -
    \entropy{\centeredTikZ{\squaresempty{2}{2}{1}{0};\squaresempty{1}{1}{0}{1};}}.\label{eq:lemma818_temp1}
    \end{aligned}
\end{equation}
By SSA, the first line of Eq.~\eqref{eq:lemma818_temp1} must be nonnegative. By Eq.~\eqref{eq:lemma818temp0}, the second line must be $0$. Therefore, the first line must be $0$. The main claim follows from this fact, by using the \hyperref[constraints:snake]{Snake-c) constraint} and the identity in Lemma~\ref{lemma:prop2to1_entlemma1}.
\end{proof}

The proof of Eqs.~\eqref{eq:prop2to1_ent3} and~\eqref{eq:prop2to1_ent4} follows the same strategy up to a $\pi$ rotation in the chosen constraints and subsystems. As such, we state these results without proof.

\begin{lemma}
\begin{equation}
    \entropy{\centeredTikZ{\squaresempty{3}{1}{0}{0};\squaresempty{2}{1}{0}{1};}}
    = \entropy{\centeredTikZ{\squaresempty{2}{1}{0}{0};\squaresempty{1}{1}{0}{1};}} + 
    \entropy{\centeredTikZ{\squaresempty{2}{2}{1}{0};}} - 
    \entropy{\centeredTikZ{\squaresempty{1}{2}{1}{0};}}.
\end{equation}
\label{lemma:prop2to1_entlemma3}
\end{lemma}

\begin{lemma}
\begin{equation}
    \entropy{\centeredTikZ{\squaresempty{3}{1}{0}{0};\squaresempty{1}{1}{0}{1};}} = \entropy{\centeredTikZ{\squaresempty{2}{1}{0}{0};\squaresempty{1}{1}{0}{1};}} + \entropy{\centeredTikZ{\squaresempty{2}{1}{0}{0};}} - \entropy{\centeredTikZ{\squaresempty{1}{1}{0}{0};}}.
\end{equation}
\label{lemma:prop2to1_entlemma4}
\end{lemma}

Moreover, using Lemma~\ref{lemma:prop2to1_entlemma1} and the \hyperref[constraints:snake]{Snake-c) constraint}, one can conclude:
\begin{corollary}
\begin{equation}
    \entropy{\centeredTikZ{\squaresempty{2}{2}{1}{0};\squaresempty{1}{1}{0}{1};}} = \entropy{\centeredTikZ{\squaresempty{2}{2}{0}{0};}}
    +
    \entropy{\centeredTikZ{\squaresempty{2}{1}{0}{0};}}
    -
    \entropy{\centeredTikZ{\squaresempty{1}{1}{0}{0};}}.
\end{equation}
\label{corollary:prop2to1_1}
\end{corollary}
Similarly, by a $\pi$-rotation,
\begin{corollary}
\begin{equation}
    \entropy{\centeredTikZ{\squaresempty{2}{2}{0}{0};\squaresempty{1}{1}{2}{0};}} = \entropy{\centeredTikZ{\squaresempty{2}{2}{0}{0};}}
    +
    \entropy{\centeredTikZ{\squaresempty{2}{1}{0}{0};}}
    -
    \entropy{\centeredTikZ{\squaresempty{1}{1}{0}{0};}}.
\end{equation}
\label{corollary:prop2to1_2}
\end{corollary}

\subsection{Merging algebra identities}
\label{appendix:2to1_merge}
Let us use the key identities that follow from Appendix~\ref{appendix:2to1_ent}.

\begin{lemma}
\begin{equation}
\begin{tikzcd}
\centeredTikZ{\squaresdotted{3}{2}{0}{0};\squaresempty{2}{1}{0}{1};\squaresempty{1}{1}{1}{0};} \maxmerge \centeredTikZ{\squaresdotted{1}{2}{0}{0};\squaresempty{2}{2}{1}{0};} 
 \arrow[r, equal] \arrow[d, equal]&  
\centeredTikZ{\squaresdotted{3}{2}{0}{0};\squaresempty{2}{1}{0}{1};\squaresempty{1}{1}{1}{0};} \rightmerge \centeredTikZ{\squaresdotted{1}{2}{0}{0};\squaresempty{2}{2}{1}{0};}\arrow[d, equal]
 \\
\centeredTikZ{\squaresdotted{3}{2}{0}{0};\squaresempty{2}{1}{0}{1};} \maxmerge \centeredTikZ{\squaresdotted{1}{2}{0}{0};\squaresempty{2}{2}{1}{0};} \arrow[r, equal]
& 
\centeredTikZ{\squaresdotted{3}{2}{0}{0};\squaresempty{2}{1}{0}{1};} \rightmerge \centeredTikZ{\squaresdotted{1}{2}{0}{0};\squaresempty{2}{2}{1}{0};}
\end{tikzcd}
\label{eq:lemma849_temp-1}
\end{equation}
\label{lemma:prop2to1_rectangle1}
\end{lemma}
\begin{proof}
The horizontal identity at the top of Eq.~\eqref{eq:lemma849_temp-1} can be derived by applying Lemma~\ref{lemma:prop2to1_entlemma1} to Lemma~\ref{lemma:fundamental}. Specifically, $\centeredTikZmini{\squaresempty{2}{2}{1}{0};
\squaresempty{1}{1}{0}{1};}$ is a state that satisfies the first condition of Lemma~\ref{lemma:fundamental}. Therefore, the top left corner of Eq.~\eqref{eq:lemma849_temp-1} is precisely $\centeredTikZmini{\squaresempty{2}{2}{1}{0};
\squaresempty{1}{1}{0}{1};}$. The third condition of Lemma~\ref{lemma:fundamental} in this context establishes the top horizontal identity.

Corollary~\ref{corollary:prop2to1_1} implies that $\centeredTikZmini{\squaresempty{2}{2}{1}{0};
\squaresempty{1}{1}{0}{1};}$ again satisfies the first condition of Lemma~\ref{lemma:fundamental} but with a different choice of marginals. Specifically, the bottom left corner of Eq.~\eqref{eq:lemma849_temp-1} is precisely $\centeredTikZmini{\squaresempty{2}{2}{1}{0};
\squaresempty{1}{1}{0}{1};}$. The third condition of Lemma~\ref{lemma:fundamental} in this context establishes the bottom horizontal identity.

The bottom and top on the left corners of Eq.~\eqref{eq:lemma849_temp-1} are identical density matrices, namely $\centeredTikZmini{\squaresempty{2}{2}{1}{0};
\squaresempty{1}{1}{0}{1};}$. This leads to the vertical identities.
\end{proof}

\begin{lemma}
\begin{equation}
    \centeredTikZ{\squaresdotted{3}{2}{0}{0};\squaresempty{3}{1}{0}{1};\squaresempty{1}{1}{2}{0};}
    =
    \centeredTikZ{\squaresdotted{3}{2}{0}{0};\squaresempty{2}{1}{1}{1};\squaresempty{1}{1}{2}{0};}
    \maxmerge
    \centeredTikZ{\squaresdotted{3}{2}{0}{0};\squaresempty{2}{1}{0}{1};}
    =
    \centeredTikZ{\squaresdotted{3}{2}{0}{0};\squaresempty{2}{1}{1}{1};\squaresempty{1}{1}{2}{0};}
    \rightmerge
    \centeredTikZ{\squaresdotted{3}{2}{0}{0};\squaresempty{2}{1}{0}{1};}.
\end{equation}
\label{lemma:prop2to1_initial1}
\end{lemma}
\begin{proof}
This is a straightforward application of Lemma~\ref{lemma:fundamental} to Lemma~\ref{lemma:prop2to1_entlemma2}, by using the equivalence of the first and the third condition of Lemma~\ref{lemma:fundamental}.
\end{proof}

As in Appendix~\ref{appendix:2to1_ent}, a $\pi$-rotated versions of these statements can be proven in a similar way. We state them without proof.
\begin{lemma}
\begin{equation}
\begin{tikzcd}
\centeredTikZ{\squaresdotted{3}{2}{0}{0};\squaresempty{2}{1}{1}{0};\squaresempty{1}{1}{1}{1};} \maxmerge \centeredTikZ{\squaresdotted{3}{2}{0}{0};\squaresempty{2}{2}{0}{0};} 
 \arrow[r, equal] \arrow[d, equal]&  
\centeredTikZ{\squaresdotted{3}{2}{0}{0};\squaresempty{2}{1}{1}{0};\squaresempty{1}{1}{1}{1};} \rightmerge \centeredTikZ{\squaresdotted{3}{2}{0}{0};\squaresempty{2}{2}{0}{0};}
\arrow[d, equal]
 \\
\centeredTikZ{\squaresdotted{3}{2}{0}{0};\squaresempty{2}{1}{1}{0};} \maxmerge \centeredTikZ{\squaresdotted{3}{2}{0}{0};\squaresempty{2}{2}{0}{0};} \arrow[r, equal]
& 
\centeredTikZ{\squaresdotted{3}{2}{0}{0};\squaresempty{2}{1}{1}{0};} \rightmerge \centeredTikZ{\squaresdotted{3}{2}{0}{0};\squaresempty{2}{2}{0}{0};}
\end{tikzcd}
\end{equation}
\label{lemma:prop2to1_rectangle2}
\end{lemma}

\begin{lemma}
\begin{equation}
    \centeredTikZ{\squaresdotted{3}{2}{0}{0};\squaresempty{3}{1}{0}{0};\squaresempty{1}{1}{0}{1};}
    =
    \centeredTikZ{\squaresdotted{3}{2}{0}{0};\squaresempty{2}{1}{0}{0};\squaresempty{1}{1}{0}{1};}
    \maxmerge
    \centeredTikZ{\squaresdotted{3}{2}{0}{0};\squaresempty{2}{1}{1}{0};}
    =
    \centeredTikZ{\squaresdotted{3}{2}{0}{0};\squaresempty{2}{1}{0}{0};\squaresempty{1}{1}{0}{1};}
    \rightmerge
    \centeredTikZ{\squaresdotted{3}{2}{0}{0};\squaresempty{2}{1}{1}{0};}.
\end{equation}
\label{lemma:prop2to1_initial2}
\end{lemma}

\subsection{Completing the proof}
\label{appendix:2to1_complete}
Now we are in a position to complete the proof of Proposition~\ref{prop:2to1}. We restate the proposition for the reader's convenience.
\twotoone*
\begin{proof}
By the mutation lemma (Lemma~\ref{lemma:snake_mutation}),

\begin{equation}
\adulthorizontalsnake{1}{y}{N}{y\shortplus 1} = 
    \left(\left(\twobytwounit{(2,y)} \rightmerge \twobytwounit{(3,y)}\right) \cdots \right) \rightmerge \twobytwounit{(N,y)}.\label{eq:lemma10_tempwhat}
\end{equation}

Consider an object $\mathfrak{S}_i$, which is defined below:
\begin{equation}
\mathfrak{S}_{i+1}
=
\Tr_{(i,y)}\left(\mathfrak{S}_i \rightmerge \twobytwounit{(i\shortplus 1,y)}\right),\label{eq:lemma10_temp2}
\end{equation}
for $i$ such that $1<i<N$, where
\begin{equation}
    \mathfrak{S}_2:=
    \mym{
    \centeredTikZ{
    \emptysquare{0.5}{0};
    \emptysquare{-0.25}{0.5};
    \emptysquare{0.25}{0.5};
    }_{(2,y)}}.
\end{equation}

We will show that
\begin{equation}
\mathfrak{S}_i
=
\mym{
    \centeredTikZ{
    \emptysquare{0.5}{0};
    \emptysquare{-0.25}{0.5};
    \emptysquare{0.25}{0.5};
    }_{(i,y)}}
    \rightmerge
\horizontalbaby{(1,y\shortplus 1)}{(i\shortminus 1, y\shortplus 1)}.
\label{eq:lemma10_temp3}
\end{equation}

The $i=3$ case follows from Lemma~\ref{lemma:prop2to1_initial1}. For $i>3$, note that the right-merge in Eq.~\eqref{eq:lemma10_temp3} ``commutes" with the right-merge of the $2\times 2$ cluster in Eq.~\eqref{eq:lemma10_temp2}, by using the commutation lemma (Lemma~\ref{lemma:commutation}). After exchanging these right-merges, we have

\begin{equation}
\begin{aligned}
    \mathfrak{S}_{i+1} &= \Tr_{(i,y)} \left(\left(\mym{
    \centeredTikZ{
    \emptysquare{0.5}{0};
    \emptysquare{-0.25}{0.5};
    \emptysquare{0.25}{0.5};
    }_{(i,y)}}
    \rightmerge
    \twobytwounit{(i\shortplus 1,y)}
    \right)\rightmerge \horizontalbaby{(1,y\shortplus 1)}{(i\shortminus 1, y\shortplus 1)} \right)\\
    &= 
     \Tr_{(i,y)}
      \left(\mym{
    \centeredTikZ{\squaresempty{2}{2}{1}{0};\squaresempty{1}{1}{0}{1};}_{(i\shortplus 1,y)}}
    \rightmerge \horizontalbaby{(1,y\shortplus 1)}{(i\shortminus 1, y\shortplus 1)} \right)\\
    &=
    \mym{
    \centeredTikZ{\squaresempty{3}{1}{0}{1};\squaresempty{1}{1}{2}{0};}_{(i\shortplus 1,y)}}
    \rightmerge \horizontalbaby{(1,y\shortplus 1)}{(i\shortminus 1, y\shortplus 1)},
\end{aligned}
\end{equation}

using the Lemma~\ref{lemma:prop2to1_rectangle1} in the second line. Moreover, 
\begin{equation}
    \begin{aligned}
    \mathfrak{S}_{i+1} &= \left(\horizontalunit{(i, y\shortplus 1)} \rightmerge 
    \mym{
    \centeredTikZ{
    \emptysquare{0.5}{0};
    \emptysquare{-0.25}{0.5};
    \emptysquare{0.25}{0.5};
    }_{(i\shortplus 1,y)}}\right) \rightmerge \horizontalbaby{(1,y\shortplus 1)}{(i\shortminus 1, y\shortplus 1)}
    \\
    &=\left(\horizontalunit{(i, y\shortplus 1)}
    \rightmerge
    \horizontalbaby{(1,y\shortplus 1)}{(i\shortminus 1, y\shortplus 1)}\right)
    \rightmerge
    \mym{
    \centeredTikZ{
    \emptysquare{0.5}{0};
    \emptysquare{-0.25}{0.5};
    \emptysquare{0.25}{0.5};
    }_{(i\shortplus 1,y)}},
    \end{aligned}
\end{equation}
using Lemma~\ref{lemma:prop2to1_initial1} in the first line and using the commutation lemma (Lemma~\ref{lemma:commutation}) in the second line. 

After applying the splitting lemma (Lemma~\ref{lemma:snake_splitting}),
\begin{equation}
    \mathfrak{S}_{i+1} = \horizontalbaby{(1, y\shortplus 1)}{(i,y\shortplus 1)} \rightmerge \mym{
    \centeredTikZ{
    \emptysquare{0.5}{0};
    \emptysquare{-0.25}{0.5};
    \emptysquare{0.25}{0.5};
    }_{(i\shortplus 1,y)}}.
\end{equation}
Moreover, we can use Lemma~\ref{lemma:prop2to1_initial1} and the merging lemma (Lemma~\ref{lemma:merging_lemma}) in the following way. Let $\lambda$ be the reduced density matrix of the level-$1$ snake over the clusters ranging from column $1$ to $i-1$, $\sigma$ be the reduced density matrix of the level-$1$ snake over the clusters on column $i-1$ and $i$, and $\tau$ be the \centeredTikZmini{\squaresempty{2}{1}{0}{1};\squaresempty{1}{1}{1}{0};}-shaped cluster anchored at $(i+1, y)$. We can apply Lemma~\ref{lemma:merging_lemma} to precisely these choices of $\lambda,\sigma,$ and $\tau$. Thus, we conclude that 
\begin{equation}
\mathfrak{S}_{i+1} = \horizontalbaby{(1, y\shortplus 1)}{(i,y\shortplus 1)} \maxmerge \mym{
    \centeredTikZ{
    \emptysquare{0.5}{0};
    \emptysquare{-0.25}{0.5};
    \emptysquare{0.25}{0.5};
    }_{(i\shortplus 1,y)}}.
\end{equation}
By using the second condition of Lemma~\ref{lemma:fundamental}, we conclude that
\begin{equation}
\mathfrak{S}_{i+1} = \mym{
    \centeredTikZ{
    \emptysquare{0.5}{0};
    \emptysquare{-0.25}{0.5};
    \emptysquare{0.25}{0.5};
    }_{(i\shortplus 1,y)}} \rightmerge \horizontalbaby{(1, y\shortplus 1)}{(i,y\shortplus 1)},
\end{equation}
proving Eq.~\eqref{eq:lemma10_temp3}. 

From Eq.~\eqref{eq:lemma10_temp3}, the main claim readily follows.

\begin{equation}
\begin{aligned}
    \Tr_y\left(\adulthorizontalsnake{1}{y}{N}{y\shortplus 1} \right)
    &=\Tr_{y, \text{last cluster}}\left(\mym{
    \centeredTikZ{
    \emptysquare{0.5}{0};
    \emptysquare{-0.25}{0.5};
    \emptysquare{0.25}{0.5};
    }_{(N,y)}} \rightmerge \horizontalbaby{(1, y\shortplus 1)}{(N\shortminus 1,y\shortplus 1)} \right)\\
    &= \horizontalunit{(N,y\shortplus 1)} \rightmerge \horizontalbaby{(1, y\shortplus 1)}{(N\shortminus 1,y\shortplus 1)} \\
    &= \horizontalbaby{(1,y\shortplus 1)}{(N, y\shortplus 1)}.
\end{aligned}
\end{equation}

The remainder of our main claim can be shown in an exactly analogous way, by ``rotating" every subsystem and constraints by $\pi$. Specifically, the order of the right-merges in Eq.~\eqref{eq:lemma10_tempwhat} is reversed. Moreover, Lemma~\ref{lemma:prop2to1_rectangle1} and Lemma~\ref{lemma:prop2to1_initial1} are changed to Lemma~\ref{lemma:prop2to1_rectangle2} and Lemma~\ref{lemma:prop2to1_initial2} respectively.
\end{proof}

\section{Level-$1$ $\to$ Level-$2$}
\label{appendix:1to2}
Now, we prove that the upward and the downward extensions map level-$1$ snakes to level-$2$ snakes. The main result of this appendix is Proposition~\ref{prop:onetotwo}.

\onetotwo*

To prove Proposition~\ref{prop:onetotwo}, it will be convenient to introduce new composite marginals, defined below.
\begin{definition}
\begin{numcases}{}
    \text{Composite 1: } \centeredTikZ{\squaresdotted{4}{2}{0}{0};\squaresempty{2}{2}{0}{0};\squaresempty{2}{1}{2}{0};} := \centeredTikZ{\squaresdotted{4}{2}{0}{0}; \squaresempty{2}{2}{0}{0};\squaresempty{1}{1}{2}{0};} \rightmerge
    \centeredTikZ{\squaresdotted{4}{2}{0}{0}; \squaresempty{2}{1}{2}{0};},\label{eq:composite1}\\[8pt]
    \text{Composite 2: }
    \centeredTikZ{\squaresempty{4}{2}{0}{0};\squaresdotted{1}{1}{3}{1};} := \centeredTikZ{\squaresdotted{4}{2}{-1}{0}; \squaresempty{2}{2}{0}{0};\squaresempty{1}{1}{2}{0};} \rightmerge 
    \centeredTikZ{\squaresdotted{4}{2}{0}{0}; \squaresempty{2}{2}{0}{0};},\label{eq:composite2}\\[8pt]
    \text{Composite 3: }
    \centeredTikZ{\squaresdotted{4}{2}{0}{0};\squaresempty{2}{2}{2}{0};\squaresempty{2}{1}{0}{1};} := \centeredTikZ{\squaresdotted{4}{2}{0}{0}; \squaresempty{2}{2}{2}{0};\squaresempty{1}{1}{1}{1};} \rightmerge
    \centeredTikZ{\squaresdotted{4}{2}{0}{0}; \squaresempty{2}{1}{0}{1};},\label{eq:composite3}\\[8pt]
    \text{Composite 4: }
    \centeredTikZ{\squaresempty{4}{2}{0}{0};\squaresdotted{1}{1}{0}{0};} := \centeredTikZ{\squaresdotted{4}{2}{-1}{0}; \squaresempty{2}{2}{0}{0};\squaresempty{1}{1}{-1}{1};} \rightmerge 
    \centeredTikZ{\squaresdotted{4}{2}{0}{0}; \squaresempty{2}{2}{2}{0};}.\label{eq:composite4}
\end{numcases}
\label{def:composite1}
\end{definition}
\noindent
We will study the properties of these objects in the remainder of Appendix~\ref{appendix:1to2}.

As a side comment, the readers may soon notice that the analysis of these objects are significantly more involved than the reduced density matrices of the fundamental marginals. This is to be expected. Unlike the fundamental marginals, it is a priori not even clear if the composite marginals are consistent with the fundamental marginals. Due to this reason, the structure of Appendix~\ref{appendix:1to2} will be quite different from that of Appendix~\ref{appendix:2to1}. We won't be able to completely decouple the analysis on the entropy from the analysis on the merging algebra.

Our analysis is organized as follows. In Appendix~\ref{sec:onetotwo_identities_basic}, we derive the basic identities that \emph{do not} involve the composite objects. In Appendix~\ref{sec:onetotwo_identities_basic}, we will derive the identities involving the composite objects. By using these identities, we will prove Proposition~\ref{prop:onetotwo} in  Appendix~\ref{sec:onetotwo_complete}.

\subsection{Basic identities}
\label{sec:onetotwo_identities_basic}
In this appendix, we will prove simple identities over the marginals of \mymarginal.
\begin{lemma}
\begin{equation}
    \centeredTikZ{\squaresdotted{3}{2}{0}{0};\squaresempty{2}{2}{0}{0};\squaresempty{1}{1}{2}{0};}
    =
    \centeredTikZ{\squaresdotted{3}{2}{0}{0};\squaresempty{3}{1}{0}{0};} \maxmerge \centeredTikZ{\squaresdotted{3}{2}{0}{0};\squaresempty{2}{2}{0}{0};} =
    \centeredTikZ{\squaresdotted{3}{2}{0}{0};\squaresempty{3}{1}{0}{0};} \rightmerge \centeredTikZ{\squaresdotted{3}{2}{0}{0};\squaresempty{2}{2}{0}{0};}.
\end{equation}
\label{lemma:prop1to2_basic1}
\end{lemma}
\begin{proof}
From the \hyperref[constraints:cell1]{Type I-b) cluster constraint},
\begin{equation}
\entropy{\centeredTikZ{\squaresempty{3}{3}{0}{0};\squaresdotted{3}{1}{0}{2};}}
+
\entropy{\centeredTikZ{\squaresdotted{3}{3}{0}{0};\squaresempty{2}{2}{0}{1};}}
-
\entropy{\centeredTikZ{\squaresdotted{3}{3}{0}{0};\squaresempty{2}{1}{0}{1};}}
-
\entropy{\centeredTikZ{\squaresempty{3}{3}{0}{0};\squaresdotted{1}{1}{2}{2};}}
=0.
\end{equation}
Using the monotonicity of conditional mutual information (Eq.~\eqref{eq:ssa_monotinicity_intro}), 
\begin{equation}
    \begin{aligned}
    \entropy{\centeredTikZ{\squaresempty{3}{3}{0}{0};\squaresdotted{3}{1}{0}{2};\squaresfilled{3}{1}{0}{0};}}
+
\entropy{\centeredTikZ{\squaresdotted{3}{3}{0}{0};\squaresempty{2}{2}{0}{1};}}
-
\entropy{\centeredTikZ{\squaresdotted{3}{3}{0}{0};\squaresempty{2}{1}{0}{1};}}
-
\entropy{\centeredTikZ{\squaresempty{3}{3}{0}{0};\squaresdotted{1}{1}{2}{2};\squaresfilled{3}{1}{0}{0};}} \\
\leq 
\entropy{\centeredTikZ{\squaresempty{3}{3}{0}{0};\squaresdotted{3}{1}{0}{2};}}
+
\entropy{\centeredTikZ{\squaresdotted{3}{3}{0}{0};\squaresempty{2}{2}{0}{1};}}
-
\entropy{\centeredTikZ{\squaresdotted{3}{3}{0}{0};\squaresempty{2}{1}{0}{1};}}
-
\entropy{\centeredTikZ{\squaresempty{3}{3}{0}{0};\squaresdotted{1}{1}{2}{2};}}.
    \end{aligned}
    \label{eq:556}
\end{equation}
The first line of Eq.~\eqref{eq:556} must be nonnegative because of SSA. The second line is $0$. Therefore, the first line must be $0$. Plugging in this result to the first condition of Lemma~\ref{lemma:fundamental}, we arrive at our claim.
\end{proof}

Of course, the $\pi$-rotated version can be proved in a similar way, by replacing the \hyperref[constraints:cell1]{Type I-b) cluster constraint} to the \hyperref[constraints:cell2]{Type II-b) cluster constraint}. This leads to the following lemma.
\begin{lemma}
\begin{equation}
    \centeredTikZ{\squaresdotted{3}{2}{0}{0};\squaresempty{2}{2}{1}{0};\squaresempty{1}{1}{0}{1};}
    =
    \centeredTikZ{\squaresdotted{3}{2}{0}{0};\squaresempty{3}{1}{0}{1};} \maxmerge \centeredTikZ{\squaresdotted{3}{2}{0}{0};\squaresempty{2}{2}{1}{0};} =
    \centeredTikZ{\squaresdotted{3}{2}{0}{0};\squaresempty{3}{1}{0}{1};} \rightmerge \centeredTikZ{\squaresdotted{3}{2}{0}{0};\squaresempty{2}{2}{1}{0};}.
\end{equation}
\label{lemma:prop1to2_basic2}
\end{lemma}

We can also derive the following identity:
\begin{lemma}
\begin{equation}
    \centeredTikZ{\squaresempty{3}{2}{0}{0};} = \centeredTikZ{\squaresdotted{3}{2}{0}{0};\squaresempty{2}{2}{0}{0};\squaresempty{1}{1}{2}{0};} \maxmerge \centeredTikZ{\squaresdotted{3}{2}{0}{0};\squaresempty{2}{2}{1}{0};}
    =
    \centeredTikZ{\squaresdotted{3}{2}{0}{0};\squaresempty{2}{2}{0}{0};\squaresempty{1}{1}{2}{0};} \rightmerge \centeredTikZ{\squaresdotted{3}{2}{0}{0};\squaresempty{2}{2}{1}{0};}.
\end{equation}
\label{lemma:prop1to2_basic3}
\end{lemma}
\begin{proof}
From the \hyperref[constraints:cell1]{Type I-a) cluster constraint},
\begin{equation}
    \entropy{\centeredTikZ{\squaresdotted{3}{2}{0}{0};\squaresempty{2}{2}{0}{0};}}
    +
    \entropy{\centeredTikZ{\squaresdotted{3}{2}{0}{0};\squaresempty{2}{2}{1}{0};}}
    -
    \entropy{\centeredTikZ{\squaresdotted{3}{2}{0}{0};\squaresempty{1}{2}{1}{0};}}
    -
    \entropy{\centeredTikZ{\squaresempty{3}{2}{0}{0};}}
    =0.
\end{equation}
Using Eq.~\eqref{eq:ssa_monotonicity_intro2},
\begin{equation}
    \begin{aligned}
    &\entropy{\centeredTikZ{\squaresdotted{3}{2}{0}{0};\squaresempty{2}{2}{0}{0};\squaresempty{1}{1}{2}{0};}}
    +
    \entropy{\centeredTikZ{\squaresdotted{3}{2}{0}{0};\squaresempty{2}{2}{1}{0};}}
    -
    \entropy{\centeredTikZ{\squaresdotted{3}{2}{0}{0};\squaresempty{1}{2}{1}{0};\squaresempty{1}{1}{2}{0};}}
    -
    \entropy{\centeredTikZ{\squaresempty{3}{2}{0}{0};}}
    \\
    &\leq 
    \entropy{\centeredTikZ{\squaresdotted{3}{2}{0}{0};\squaresempty{2}{2}{0}{0};}}
    +
    \entropy{\centeredTikZ{\squaresdotted{3}{2}{0}{0};\squaresempty{2}{2}{1}{0};}}
    -
    \entropy{\centeredTikZ{\squaresdotted{3}{2}{0}{0};\squaresempty{1}{2}{1}{0};}}
    -
    \entropy{\centeredTikZ{\squaresempty{3}{2}{0}{0};}}.
    \end{aligned}
    \label{eq:930}
\end{equation}
The first line of Eq.~\eqref{eq:930} must be nonnegative because of SSA. The second line is $0$. Therefore, the first line must be $0$. Plugging in this result to the first condition of Lemma~\ref{lemma:fundamental}, we arrive at our claim.
\end{proof}

The $\pi$-rotated version is proved in a similar way. We state the result without proof.
\begin{lemma}
\begin{equation}
    \centeredTikZ{\squaresempty{3}{2}{0}{0};} = \centeredTikZ{\squaresdotted{3}{2}{0}{0};\squaresempty{2}{2}{1}{0};\squaresempty{1}{1}{0}{1};} \maxmerge \centeredTikZ{\squaresdotted{3}{2}{0}{0};\squaresempty{2}{2}{0}{0};}
    =
   \centeredTikZ{\squaresdotted{3}{2}{0}{0};\squaresempty{2}{2}{1}{0};\squaresempty{1}{1}{0}{1};} \rightmerge \centeredTikZ{\squaresdotted{3}{2}{0}{0};\squaresempty{2}{2}{0}{0};}.
\end{equation}
\label{lemma:prop1to2_basic4}
\end{lemma}

\subsection{Composite identities}
\label{sec:onetotwo_identities_composite}
Now, we will derive a number of basic identities involving the objects defined in Definition~\ref{def:composite1}. The key findings are summarized below.

\begin{equation}
    \begin{tikzcd}
    \text{\compositeone} \arrow[r,equal] &\centeredTikZ{\squaresdotted{4}{2}{0}{0};\squaresempty{2}{2}{0}{0};\squaresempty{1}{1}{2}{0};} \maxmerge  \centeredTikZ{\squaresdotted{4}{2}{0}{0};\squaresempty{2}{1}{2}{0};} \arrow[r, equal] \arrow[d, equal]& \centeredTikZ{\squaresdotted{4}{2}{0}{0};\squaresempty{2}{2}{0}{0};\squaresempty{1}{1}{2}{0};} \rightmerge  \centeredTikZ{\squaresdotted{4}{2}{0}{0};\squaresempty{2}{1}{2}{0};}\arrow[d,equal] \\
    &\centeredTikZ{\squaresdotted{4}{2}{0}{0};\squaresempty{2}{2}{0}{0};}
    \maxmerge
    \centeredTikZ{\squaresdotted{4}{2}{0}{0};\squaresempty{3}{1}{1}{0};\squaresempty{1}{1}{1}{1};} \arrow[r, equal] & \centeredTikZ{\squaresdotted{4}{2}{0}{0};\squaresempty{2}{2}{0}{0};}
    \rightmerge
    \centeredTikZ{\squaresdotted{4}{2}{0}{0};\squaresempty{3}{1}{1}{0};\squaresempty{1}{1}{1}{1};}
    \end{tikzcd}
    \label{eq:identities_composite1}
\end{equation}
\begin{equation}
    \begin{tikzcd}
    \text{\compositetwo} \arrow[r, equal] &
    \centeredTikZ{\squaresdotted{4}{2}{-1}{0};\squaresempty{2}{2}{0}{0};\squaresempty{1}{1}{2}{0};}
    \maxmerge 
    \centeredTikZ{\squaresdotted{4}{2}{0}{0};\squaresempty{2}{2}{0}{0};}\arrow[r, equal]\arrow[d, equal] & \centeredTikZ{\squaresdotted{4}{2}{-1}{0};\squaresempty{2}{2}{0}{0};\squaresempty{1}{1}{2}{0};}
    \rightmerge 
    \centeredTikZ{\squaresdotted{4}{2}{0}{0};\squaresempty{2}{2}{0}{0};} \arrow[d, equal]\\
    &
    \text{\compositeone}
    \maxmerge 
    \centeredTikZ{\squaresdotted{4}{2}{0}{0};\squaresempty{2}{2}{1}{0};}\arrow[r, equal] & 
    \text{\compositeone}
    \rightmerge 
    \centeredTikZ{\squaresdotted{4}{2}{0}{0};\squaresempty{2}{2}{1}{0};}
    \end{tikzcd}
    \label{eq:identities_composite2}
\end{equation}

We will also state the $\pi$-rotated versions of these statements at the end. Since the underlying analysis is essentially the same, we will focus on the proof of Eqs.~\eqref{eq:identities_composite1} and~\eqref{eq:identities_composite2}.

Let us first prove Eq.~\eqref{eq:identities_composite1}.
\begin{lemma}
\label{lemma:prop1to2_ent1}
\begin{align}
    \entropy{\centeredTikZ{\squaresdotted{4}{2}{0}{0};\squaresempty{2}{2}{0}{0};\squaresempty{2}{1}{2}{0};}}
     &= 
    \entropy{\centeredTikZ{\squaresempty{2}{2}{0}{0};\squaresempty{1}{1}{2}{0};}}
    +
    \entropy{\centeredTikZ{\squaresempty{2}{1}{0}{0};}}
    -
    \entropy{\centeredTikZ{\squaresempty{1}{1}{0}{0};}}
    \label{eq:350} \\
    &=
    \entropy{\centeredTikZ{\squaresempty{2}{2}{0}{0};}}
    +
    \entropy{\centeredTikZ{\squaresempty{3}{1}{0}{0};
    \squaresempty{1}{1}{0}{1};}}
    -
    \entropy{\centeredTikZ{\squaresempty{1}{2}{0}{0};}}.\label{eq:350_1}
\end{align}
Moreover,
\begin{align}
\centeredTikZ{\squaresdotted{4}{2}{0}{0};\squaresempty{4}{1}{0}{0};\squaresempty{2}{1}{0}{1};}&=
\centeredTikZ{\squaresdotted{4}{2}{0}{0};\squaresempty{2}{2}{0}{0};\squaresempty{1}{1}{2}{0};} \maxmerge  \centeredTikZ{\squaresdotted{4}{2}{0}{0};\squaresempty{2}{1}{2}{0};} = 
\centeredTikZ{\squaresdotted{4}{2}{0}{0};\squaresempty{2}{2}{0}{0};\squaresempty{1}{1}{2}{0};} \rightmerge \centeredTikZ{\squaresdotted{4}{2}{0}{0};\squaresempty{2}{1}{2}{0};}\label{eq:349}\\
&=
    \centeredTikZ{\squaresdotted{4}{2}{0}{0};\squaresempty{2}{2}{0}{0};}
    \maxmerge
    \centeredTikZ{\squaresdotted{4}{2}{0}{0};\squaresempty{3}{1}{1}{0};\squaresempty{1}{1}{1}{1};}
    =
    \centeredTikZ{\squaresdotted{4}{2}{0}{0};\squaresempty{2}{2}{0}{0};}
    \rightmerge
    \centeredTikZ{\squaresdotted{4}{2}{0}{0};\squaresempty{3}{1}{1}{0};\squaresempty{1}{1}{1}{1};}.\label{eq:349_1}
\end{align}
\end{lemma}
\begin{proof}
By Corollary~\ref{corollary:prop2to1_2}, 
\begin{equation}
     \entropy{\centeredTikZ{\squaresdotted{4}{2}{0}{0};\squaresempty{2}{2}{0}{0};\squaresfilled{1}{1}{2}{0};}}
     +
     \entropy{\centeredTikZ{\squaresdotted{4}{2}{0}{0};\squaresfilled{2}{2}{0}{0};\squaresempty{2}{1}{1}{0};}}
     -
     \entropy{\centeredTikZ{\squaresdotted{4}{2}{0}{0};\squaresfilled{2}{2}{0}{0};\squaresfilled{1}{1}{2}{0};\squaresempty{1}{1}{1}{0};}}
     -
     \entropy{\centeredTikZ{\squaresdotted{4}{2}{0}{0};\squaresempty{2}{2}{0}{0};\squaresempty{1}{1}{2}{0};}}
     =0.
\end{equation}
By Lemma~\ref{lemma:fundamental} (specifically, using the equivalence of the first and the third condition),
\begin{equation}
\begin{aligned}
    \centeredTikZ{\squaresdotted{4}{2}{0}{0};\squaresempty{2}{2}{0}{0};\squaresempty{1}{1}{2}{0};} &= \centeredTikZ{\squaresdotted{4}{2}{0}{0};\squaresempty{2}{2}{0}{0};\squaresfilled{1}{1}{2}{0};} \maxmerge \centeredTikZ{\squaresdotted{4}{2}{0}{0};\squaresfilled{2}{2}{0}{0};\squaresempty{2}{1}{1}{0};} = 
    \centeredTikZ{\squaresdotted{4}{2}{0}{0};\squaresempty{2}{2}{0}{0};\squaresfilled{1}{1}{2}{0};} \rightmerge \centeredTikZ{\squaresdotted{4}{2}{0}{0};\squaresfilled{2}{2}{0}{0};\squaresempty{2}{1}{1}{0};} \\
    &=
    \centeredTikZ{\squaresdotted{4}{2}{0}{0};\squaresempty{2}{2}{0}{0};} \maxmerge \centeredTikZ{\squaresdotted{4}{2}{0}{0};\squaresempty{2}{1}{1}{0};} = 
    \centeredTikZ{\squaresdotted{4}{2}{0}{0};\squaresempty{2}{2}{0}{0};} \rightmerge \centeredTikZ{\squaresdotted{4}{2}{0}{0};\squaresempty{2}{1}{1}{0};}.
\end{aligned}
    \label{eq:lemma850_temp0}
\end{equation}

Moreover, applying the \hyperref[constraints:snake]{Snake-b) constraint} to Lemma~\ref{lemma:fundamental} (again, using the equivalence of the first and the third condition),
\begin{equation}
\begin{aligned}
\centeredTikZ{\squaresdotted{4}{2}{0}{0};\squaresempty{3}{1}{1}{0};} &= 
\centeredTikZ{\squaresdotted{4}{2}{0}{0};\squaresempty{3}{1}{1}{0};\squaresfilled{1}{1}{3}{0};} 
\maxmerge
\centeredTikZ{\squaresdotted{4}{2}{0}{0};\squaresempty{2}{1}{2}{0};\squaresfilled{1}{1}{1}{0};}
=
\centeredTikZ{\squaresdotted{4}{2}{0}{0};\squaresempty{3}{1}{1}{0};\squaresfilled{1}{1}{3}{0};} 
\rightmerge
\centeredTikZ{\squaresdotted{4}{2}{0}{0};\squaresempty{2}{1}{2}{0};\squaresfilled{1}{1}{1}{0};} \\
&= 
\centeredTikZ{\squaresdotted{4}{2}{0}{0};\squaresempty{2}{1}{1}{0};} 
\maxmerge
\centeredTikZ{\squaresdotted{4}{2}{0}{0};\squaresempty{2}{1}{2}{0};}
=
\centeredTikZ{\squaresdotted{4}{2}{0}{0};\squaresempty{2}{1}{1}{0};} 
\rightmerge
\centeredTikZ{\squaresdotted{4}{2}{0}{0};\squaresempty{2}{1}{2}{0};}.
\end{aligned}
\label{eq:lemma850_temp1}
\end{equation}

We can plug in Eqs.~\eqref{eq:lemma850_temp0} and~\eqref{eq:lemma850_temp1} to the merging lemma (Lemma~\ref{lemma:merging_lemma}), resulting in the following identities:
\begin{equation}
\begin{aligned}
\centeredTikZ{\squaresdotted{4}{2}{0}{0};\squaresempty{2}{2}{0}{0};\squaresempty{1}{1}{2}{0};} \maxmerge  \centeredTikZ{\squaresdotted{4}{2}{0}{0};\squaresempty{2}{1}{2}{0};} &= 
\centeredTikZ{\squaresdotted{4}{2}{0}{0};\squaresempty{2}{2}{0}{0};\squaresempty{1}{1}{2}{0};} \rightmerge \centeredTikZ{\squaresdotted{4}{2}{0}{0};\squaresempty{2}{1}{2}{0};}\\
&=
\centeredTikZ{\squaresdotted{4}{2}{0}{0};\squaresempty{4}{1}{0}{0};\squaresempty{2}{1}{0}{1};},
\end{aligned}
\end{equation}
establishing Eq.~\eqref{eq:349}. From the equivalence of the third and the first condition of Lemma~\ref{lemma:fundamental}, Eq.~\eqref{eq:350} follows.

Next, from Lemma~\ref{lemma:prop2to1_entlemma3},
\begin{equation}
    \entropy{\centeredTikZ{\squaresdotted{4}{2}{0}{0};\squaresempty{2}{2}{0}{0};\squaresempty{1}{1}{2}{0};\squaresfilled{1}{1}{2}{0};}}+
    \entropy{\centeredTikZ{\squaresdotted{4}{2}{0}{0};\squaresempty{2}{2}{0}{0};\squaresempty{1}{1}{2}{0};\squaresfilled{1}{2}{0}{0};}}-
    \entropy{\centeredTikZ{\squaresdotted{4}{2}{0}{0};\squaresempty{2}{2}{0}{0};\squaresempty{1}{1}{2}{0};\squaresfilled{1}{2}{0}{0};\squaresfilled{1}{1}{2}{0};}}-
    \entropy{\centeredTikZ{\squaresdotted{4}{2}{0}{0};\squaresempty{2}{2}{0}{0};\squaresempty{1}{1}{2}{0};}}=0
\end{equation}
By using the equivalence of the first and the third condition in Lemma~\ref{lemma:fundamental},
\begin{equation}
    \begin{aligned}
\centeredTikZ{\squaresdotted{4}{2}{0}{0};\squaresempty{2}{2}{0}{0};\squaresempty{1}{1}{2}{0};}
&= \centeredTikZ{\squaresdotted{4}{2}{0}{0};\squaresempty{2}{2}{0}{0};\squaresempty{1}{1}{2}{0};\squaresfilled{1}{1}{2}{0};}
\maxmerge 
\centeredTikZ{\squaresdotted{4}{2}{0}{0};\squaresempty{2}{2}{0}{0};\squaresempty{1}{1}{2}{0};\squaresfilled{1}{2}{0}{0};}
=\centeredTikZ{\squaresdotted{4}{2}{0}{0};\squaresempty{2}{2}{0}{0};\squaresempty{1}{1}{2}{0};\squaresfilled{1}{1}{2}{0};}
\rightmerge 
\centeredTikZ{\squaresdotted{4}{2}{0}{0};\squaresempty{2}{2}{0}{0};\squaresempty{1}{1}{2}{0};\squaresfilled{1}{2}{0}{0};}\\
&= \centeredTikZ{\squaresdotted{4}{2}{0}{0};\squaresempty{2}{2}{0}{0};\squaresempty{1}{1}{2}{0};\squaresdotted{1}{1}{2}{0};}
\maxmerge 
\centeredTikZ{\squaresdotted{4}{2}{0}{0};\squaresempty{2}{2}{0}{0};\squaresempty{1}{1}{2}{0};\squaresdotted{1}{2}{0}{0};}
=\centeredTikZ{\squaresdotted{4}{2}{0}{0};\squaresempty{2}{2}{0}{0};\squaresempty{1}{1}{2}{0};\squaresdotted{1}{1}{2}{0};}
\rightmerge 
\centeredTikZ{\squaresdotted{4}{2}{0}{0};\squaresempty{2}{2}{0}{0};\squaresempty{1}{1}{2}{0};\squaresdotted{1}{2}{0}{0};}.\label{eq:551}
    \end{aligned}
\end{equation}

Moreover, applying Lemma~\ref{lemma:prop2to1_entlemma4} to the first condition of Lemma~\ref{lemma:fundamental}, we obtain
\begin{equation}
    \begin{aligned}
\centeredTikZ{\squaresdotted{4}{2}{0}{0};\squaresempty{2}{1}{1}{0};\squaresempty{1}{1}{1}{1};}
\maxmerge
\centeredTikZ{\squaresdotted{4}{2}{0}{0};\squaresempty{2}{1}{2}{0};}
&=
\centeredTikZ{\squaresdotted{4}{2}{0}{0};\squaresempty{2}{1}{1}{0};\squaresempty{1}{1}{1}{1};}
\rightmerge
\centeredTikZ{\squaresdotted{4}{2}{0}{0};\squaresempty{2}{1}{2}{0};}\\
&=
\centeredTikZ{\squaresdotted{4}{2}{0}{0};\squaresempty{3}{1}{1}{0};\squaresempty{1}{1}{1}{1};}.
    \end{aligned}\label{eq:551_1}
\end{equation}

We can plug in Eqs.~\eqref{eq:551} and~\eqref{eq:551_1} to the merging lemma (Lemma~\ref{lemma:merging_lemma}), resulting in the following identities:
\begin{equation}
\centeredTikZ{\squaresdotted{4}{2}{0}{0};\squaresempty{4}{1}{0}{0};\squaresempty{2}{1}{0}{1};}=
    \centeredTikZ{\squaresdotted{4}{2}{0}{0};\squaresempty{2}{2}{0}{0};}
    \maxmerge
    \centeredTikZ{\squaresdotted{4}{2}{0}{0};\squaresempty{3}{1}{1}{0};\squaresempty{1}{1}{1}{1};}
    =
    \centeredTikZ{\squaresdotted{4}{2}{0}{0};\squaresempty{2}{2}{0}{0};}
    \rightmerge
    \centeredTikZ{\squaresdotted{4}{2}{0}{0};\squaresempty{3}{1}{1}{0};\squaresempty{1}{1}{1}{1};},
\end{equation}
establishing Eq.~\eqref{eq:349_1}. From the equivalence of the third and the first condition of Lemma~\ref{lemma:fundamental}, Eq.~\eqref{eq:350_1} follows.
\end{proof}

Let us prove the identity at the top of Eq.~\eqref{eq:identities_composite2}.
\begin{lemma}
\label{lemma:prop1to2_ent2}
\begin{equation}
    \entropy{\centeredTikZ{\squaresempty{4}{2}{0}{0};\squaresdotted{1}{1}{3}{1};}}=
    \entropy{\centeredTikZ{\squaresempty{2}{2}{0}{0};}}
    +
    \entropy{\centeredTikZ{\squaresempty{2}{2}{0}{0};\squaresempty{1}{1}{2}{0};}}
    -
    \entropy{\centeredTikZ{\squaresempty{1}{2}{0}{0};}}.
    \label{eq:351}
\end{equation}
Moreover,
\begin{equation}
\centeredTikZ{\squaresempty{4}{2}{0}{0};\squaresdotted{1}{1}{3}{1};}
=
\centeredTikZ{\squaresdotted{4}{2}{-1}{0};\squaresempty{2}{2}{0}{0};\squaresempty{1}{1}{2}{0};}
    \maxmerge 
    \centeredTikZ{\squaresdotted{4}{2}{0}{0};\squaresempty{2}{2}{0}{0};}
    =
    \centeredTikZ{\squaresdotted{4}{2}{-1}{0};\squaresempty{2}{2}{0}{0};\squaresempty{1}{1}{2}{0};}
    \rightmerge 
    \centeredTikZ{\squaresdotted{4}{2}{0}{0};\squaresempty{2}{2}{0}{0};} \\
\label{eq:352}
\end{equation}
\end{lemma}
\begin{proof}
By Lemma~\ref{lemma:prop2to1_entlemma3}, 
\begin{equation}
     \entropy{\centeredTikZ{\squaresdotted{4}{2}{-1}{0};\squaresempty{2}{2}{0}{0};\squaresfilled{1}{1}{2}{0};}}
     +
     \entropy{\centeredTikZ{\squaresdotted{4}{2}{-1}{0};\squaresfilled{2}{2}{0}{0};\squaresempty{2}{1}{1}{0};\squaresempty{1}{1}{1}{1};}}
     -
     \entropy{\centeredTikZ{\squaresdotted{4}{2}{-1}{0};\squaresfilled{2}{2}{0}{0};\squaresfilled{1}{1}{2}{0};\squaresempty{1}{2}{1}{0};}}
     -
     \entropy{\centeredTikZ{\squaresdotted{4}{2}{-1}{0};\squaresempty{2}{2}{0}{0};\squaresempty{1}{1}{2}{0};}}
     =0.
\end{equation}
By Lemma~\ref{lemma:fundamental} (specifically, using the equivalence of the first and the third condition),
\begin{equation}
\begin{aligned}
    \centeredTikZ{\squaresdotted{4}{2}{-1}{0};\squaresempty{2}{2}{0}{0};\squaresempty{1}{1}{2}{0};} &= \centeredTikZ{\squaresdotted{4}{2}{-1}{0};\squaresempty{2}{2}{0}{0};\squaresfilled{1}{1}{2}{0};} \maxmerge \centeredTikZ{\squaresdotted{4}{2}{-1}{0};\squaresfilled{2}{2}{0}{0};\squaresempty{2}{1}{1}{0};\squaresempty{1}{1}{1}{1};} = 
    \centeredTikZ{\squaresdotted{4}{2}{-1}{0};\squaresempty{2}{2}{0}{0};\squaresfilled{1}{1}{2}{0};} \rightmerge \centeredTikZ{\squaresdotted{4}{2}{-1}{0};\squaresfilled{2}{2}{0}{0};\squaresempty{2}{1}{1}{0};\squaresempty{1}{1}{1}{1};} \\
    &=
    \centeredTikZ{\squaresdotted{4}{2}{-1}{0};\squaresempty{2}{2}{0}{0};} \maxmerge \centeredTikZ{\squaresdotted{4}{2}{-1}{0};\squaresempty{2}{1}{1}{0};\squaresempty{1}{1}{1}{1};} = 
    \centeredTikZ{\squaresdotted{4}{2}{-1}{0};\squaresempty{2}{2}{0}{0};} \rightmerge \centeredTikZ{\squaresdotted{4}{2}{-1}{0};\squaresempty{2}{1}{1}{0};\squaresempty{1}{1}{1}{1};}.
\end{aligned}
    \label{eq:lemma918_temp0}
\end{equation}
Moreover, by applying the \hyperref[constraints:cell1]{Type I-a) cluster constraint} to the first condition of Lemma~\ref{lemma:fundamental}, we get
\begin{equation}
    \centeredTikZ{\squaresdotted{4}{2}{0}{0};\squaresempty{2}{2}{0}{0};}
    \maxmerge
    \centeredTikZ{\squaresdotted{4}{2}{0}{0};\squaresempty{2}{2}{1}{0};}
    =
    \centeredTikZ{\squaresdotted{4}{2}{0}{0};\squaresempty{2}{2}{0}{0};}
    \rightmerge
    \centeredTikZ{\squaresdotted{4}{2}{0}{0};\squaresempty{2}{2}{1}{0};}.
\label{eq:lemma918_temp1}
\end{equation}
By plugging in Eqs.~\eqref{eq:lemma918_temp0} and~\eqref{eq:lemma918_temp1} to the merging lemma (Lemma~\ref{lemma:merging_lemma}), we obtain the following identity:
\begin{equation}
    \centeredTikZ{\squaresdotted{4}{2}{-1}{0};\squaresempty{2}{2}{0}{0};\squaresempty{1}{1}{2}{0};}
    \maxmerge 
    \centeredTikZ{\squaresdotted{4}{2}{0}{0};\squaresempty{2}{2}{0}{0};}
    =
    \centeredTikZ{\squaresdotted{4}{2}{-1}{0};\squaresempty{2}{2}{0}{0};\squaresempty{1}{1}{2}{0};}
    \rightmerge 
    \centeredTikZ{\squaresdotted{4}{2}{0}{0};\squaresempty{2}{2}{0}{0};},
\end{equation}
establishing Eq.~\eqref{eq:352}. From the equivalence of the third and the first condition of Lemma~\ref{lemma:fundamental}, Eq.~\eqref{eq:351} follows.
\end{proof}

It remains to prove the equivalence of \compositetwo $\,$ with the bottom row of Eq.~\eqref{eq:identities_composite2}. This follows from two conditions, the consistency condition and the entropy condition. (see the first condition of Lemma~\ref{lemma:fundamental}.) First, 
\begin{equation}
\begin{aligned}
    \text{\compositetwo} &\consistent \text{\compositeone} \\
    \text{\compositetwo} &\consistent \centeredTikZ{\squaresdotted{4}{2}{0}{0};\squaresempty{2}{2}{1}{0};}. \label{eq:compositetwo_consistency}
\end{aligned}
\end{equation}
The first consistency condition can be readily verified from the bottom of Eq.~\eqref{eq:identities_composite1} and the the top right corner of Eq.~\eqref{eq:identities_composite2}. By using this relation, after taking a partial trace on $\text{Supp}(\text{Composite 2})\setminus \text{Supp}(\text{Composite 1})$ on \compositetwo, we obtain \compositeone. The second consistency condition also follows easily. Take a partial trace on the bottom-right corner of \compositetwo, using the top-right corner of Eq.~\eqref{eq:identities_composite2}. By the \hyperref[constraints:cell1]{Type I-a) cluster constraint}, the resulting density matrix must be consistent with $\mymarginal$, which proves the claim. Therefore, we have proved the consistency conditions in Eq.~\eqref{eq:compositetwo_consistency}. 

Second, we prove the entropy condition.
\begin{corollary}
\label{corollary:prop1to2_prelim1}
\begin{equation}
    \entropy{\centeredTikZ{\squaresempty{4}{2}{0}{0};\squaresdotted{2}{1}{2}{1};}}
    =
    \entropy{\centeredTikZ{\squaresempty{2}{2}{0}{0};}}
    +
    \entropy{\centeredTikZ{\squaresempty{3}{1}{0}{0};
    \squaresempty{1}{1}{0}{1};}}
    -
    \entropy{\centeredTikZ{\squaresempty{1}{2}{0}{0};}}.\label{eq:407}
\end{equation}
\end{corollary}
\begin{proof}
From Lemma~\ref{lemma:prop1to2_ent1},
\begin{equation}
    \entropy{\centeredTikZ{\squaresdotted{4}{2}{0}{0};\squaresempty{2}{2}{0}{0};\squaresempty{2}{1}{2}{0};}}
     = 
    \entropy{\centeredTikZ{\squaresempty{2}{2}{0}{0};\squaresempty{1}{1}{2}{0};}}
    +
    \entropy{\centeredTikZ{\squaresempty{2}{1}{0}{0};}}
    -
    \entropy{\centeredTikZ{\squaresempty{1}{1}{0}{0};}}.
\end{equation}
From Corollary~\ref{corollary:prop2to1_2}, 
\begin{equation}
\entropy{\centeredTikZ{\squaresempty{2}{2}{0}{0};\squaresempty{1}{1}{2}{0};}} = \entropy{\centeredTikZ{\squaresempty{2}{2}{0}{0};}}
    +
    \entropy{\centeredTikZ{\squaresempty{2}{1}{0}{0};}}
    -
    \entropy{\centeredTikZ{\squaresempty{1}{1}{0}{0};}}.
\end{equation}
Combining these two,
\begin{equation}
    \entropy{\centeredTikZ{\squaresdotted{4}{2}{0}{0};\squaresempty{2}{2}{0}{0};\squaresempty{2}{1}{2}{0};}}
     = \entropy{\centeredTikZ{\squaresempty{2}{2}{0}{0};}}+ 2\entropy{\centeredTikZ{\squaresempty{2}{1}{0}{0};}}
    -
    \entropy{\centeredTikZ{\squaresempty{1}{1}{0}{0};}}.
    \label{eq:corollary221_temp0}
\end{equation}

From Lemma~\ref{lemma:prop2to1_entlemma4},
\begin{equation}
\begin{aligned}
\entropy{\centeredTikZ{\squaresempty{3}{1}{0}{0};\squaresempty{1}{1}{0}{1};}} &= \entropy{\centeredTikZ{\squaresempty{2}{1}{0}{0};\squaresempty{1}{1}{0}{1};}} + \entropy{\centeredTikZ{\squaresempty{2}{1}{0}{0};}} - \entropy{\centeredTikZ{\squaresempty{1}{1}{0}{0};}} \\
&= \entropy{\centeredTikZ{\squaresempty{1}{2}{0}{1};}} + 2\entropy{\centeredTikZ{\squaresempty{2}{1}{0}{0};}} - \entropy{\centeredTikZ{\squaresempty{1}{1}{0}{0};}},
\end{aligned}
\label{eq:corollary221_temp1}
\end{equation}
where in the second line we used the \hyperref[constraints:snake]{Snake-a) constraint}. From Eqs.~\eqref{eq:corollary221_temp0} and~\eqref{eq:corollary221_temp1}, Eq.~\eqref{eq:407} immediately follows.
\end{proof}
Thus, we have proved the entropy condition as well. By plugging in Eqs.~\eqref{eq:407} and~\eqref{eq:compositetwo_consistency} to the first condition of Lemma~\ref{lemma:fundamental}, we conclude
\begin{lemma}
\begin{equation}
\begin{aligned}
\text{\compositetwo}&=
    \text{\compositeone}
    \maxmerge 
    \centeredTikZ{\squaresdotted{4}{2}{0}{0};\squaresempty{2}{2}{1}{0};} \\
    &=\text{\compositeone}
    \rightmerge 
    \centeredTikZ{\squaresdotted{4}{2}{0}{0};\squaresempty{2}{2}{1}{0};},
\end{aligned}
\end{equation}
\end{lemma}
\noindent
completing the proof of Eq.~\eqref{eq:identities_composite2}.

As before, we can apply the analogous proofs to \compositethree $\,$ and \compositefour $\,$ by rotating all the involved subsystems and constraints by $\pi$. We will state them below without proofs.

\begin{equation}
    \begin{tikzcd}
    \text{\compositethree} \arrow[r,equal] &\centeredTikZ{\squaresdotted{4}{2}{0}{0};\squaresempty{2}{2}{2}{0};\squaresempty{1}{1}{1}{1};} \maxmerge  \centeredTikZ{\squaresdotted{4}{2}{0}{0};\squaresempty{2}{1}{0}{1};} \arrow[r, equal] \arrow[d, equal]& \centeredTikZ{\squaresdotted{4}{2}{0}{0};\squaresempty{2}{2}{2}{0};\squaresempty{1}{1}{1}{1};} \rightmerge  \centeredTikZ{\squaresdotted{4}{2}{0}{0};\squaresempty{2}{1}{0}{1};} \\
    &\centeredTikZ{\squaresdotted{4}{2}{0}{0};\squaresempty{2}{2}{2}{0};}
    \maxmerge
    \centeredTikZ{\squaresdotted{4}{2}{0}{0};\squaresempty{3}{1}{0}{1};\squaresempty{1}{1}{2}{0};} \arrow[r, equal] & \centeredTikZ{\squaresdotted{4}{2}{0}{0};\squaresempty{2}{2}{2}{0};}
    \rightmerge
    \centeredTikZ{\squaresdotted{4}{2}{0}{0};\squaresempty{3}{1}{0}{1};\squaresempty{1}{1}{2}{0};}
    \end{tikzcd}
    \label{eq:identities_composite3}
\end{equation}
\begin{equation}
    \begin{tikzcd}
    \text{\compositefour} \arrow[r, equal] &
    \centeredTikZ{\squaresdotted{4}{2}{-1}{0};\squaresempty{2}{2}{0}{0};\squaresempty{1}{1}{-1}{1};}
    \maxmerge 
    \centeredTikZ{\squaresdotted{4}{2}{0}{0};\squaresempty{2}{2}{2}{0};}\arrow[r, equal]\arrow[d, equal] & \centeredTikZ{\squaresdotted{4}{2}{-1}{0};\squaresempty{2}{2}{0}{0};\squaresempty{1}{1}{-1}{1};}
    \rightmerge 
    \centeredTikZ{\squaresdotted{4}{2}{0}{0};\squaresempty{2}{2}{2}{0};} \arrow[d, equal]\\
    &
    \text{\compositethree}
    \maxmerge 
    \centeredTikZ{\squaresdotted{4}{2}{0}{0};\squaresempty{2}{2}{1}{0};}\arrow[r, equal] &
    \text{\compositethree}
    \rightmerge 
    \centeredTikZ{\squaresdotted{4}{2}{0}{0};\squaresempty{2}{2}{1}{0};}
    \end{tikzcd}
    \label{eq:identities_composite4}
\end{equation}

\subsection{Completing the proof}
\label{sec:onetotwo_complete}
With the identities derived so far, we are in a position to prove the main result of this appendix, Proposition~\ref{prop:onetotwo}. We restate this below for the reader's convenience.
\onetotwo*
\begin{proof}
Let us define the following object:
\begin{equation}
    \mathfrak{S}_{i+1} = \left(\mathfrak{S}_{i} \rightmerge \horizontalunit{(i\shortplus 1, y)} \right) \rightmerge \twobytwounit{(i,y)},
\end{equation}
for $i\geq 3$ where
\begin{equation}
    \mathfrak{S}_3 := \twobytwounit{(2,y)} \rightmerge \horizontalunit{(3,y)}.
\end{equation}
$\mathfrak{S}_i$ is related to the main claim of this proposition in the following way. Using the commutation lemma (Lemma~\ref{lemma:commutation}), one can show that
\begin{equation}
    \mathcal{E}_{y,\uparrow} \left(\horizontalbaby{(1,y)}{(N,y)} \right)
    =
    \mathfrak{S}_{N} \rightmerge \twobytwounit{(N,y)}.
    \label{eq:859}
\end{equation}

We will show that
\begin{equation}
    \mathfrak{S}_i = \adulthorizontalsnake{1}{y}{i\shortminus 1}{y\shortplus 1} \rightmerge \mym{\centeredTikZ{\squaresempty{2}{1}{0}{0};\squaresempty{1}{1}{0}{1};}_{(i,y)}}.
\end{equation}
The $i=3$ case follows directly from Lemma~\ref{lemma:prop2to1_rectangle2}. 

For $i>3$, note that
\begin{equation}
    \adulthorizontalsnake{1}{y}{i\shortminus 1}{y\shortplus 1}
    = 
    \twobytwounit{(i\shortminus 1, y)} \rightmerge
    \adulthorizontalsnake{1}{y}{i\shortminus 2}{y\shortplus 1}
\end{equation}
by the splitting lemma (Lemma~\ref{lemma:snake_splitting}). Therefore,

\begin{equation}
\begin{aligned}
    \mathfrak{S}_{i+1} &= 
    \left(\left(\left(\twobytwounit{(i\shortminus 1, y)} \rightmerge
    \adulthorizontalsnake{1}{y}{i\shortminus 2}{y\shortplus 1} \right)
    \rightmerge 
    \mym{\centeredTikZ{\squaresempty{2}{1}{0}{0};\squaresempty{1}{1}{0}{1};}_{(i,y)}}
    \right) \rightmerge \horizontalunit{(i\shortplus 1, y)} \right) \rightmerge \twobytwounit{(i,y)}\\
    &=
    \left(\left(\left(\twobytwounit{(i\shortminus 1, y)}
    \rightmerge 
    \mym{\centeredTikZ{\squaresempty{2}{1}{0}{0};\squaresempty{1}{1}{0}{1};}_{(i,y)}}
    \right) \rightmerge \horizontalunit{(i\shortplus 1, y)} \right) \rightmerge \twobytwounit{(i,y)}\right) \rightmerge
    \adulthorizontalsnake{1}{y}{i\shortminus 2}{y\shortplus 1},
\end{aligned}
\end{equation}

which follows from the commutation lemma (Lemma~\ref{lemma:commutation}). Note that
\begin{equation}
\begin{aligned}
&\left(\left(\left(\twobytwounit{(i\shortminus 1, y)}
    \rightmerge 
    \mym{\centeredTikZ{\squaresempty{2}{1}{0}{0};\squaresempty{1}{1}{0}{1};}_{(i,y)}}
    \right) \rightmerge \horizontalunit{(i\shortplus 1, y)} \right) \rightmerge \twobytwounit{(i,y)}\right)  \\
    &=
\left(\left(
    \mym{\centeredTikZ{\squaresempty{2}{2}{-2}{0};\squaresempty{1}{1}{0}{0};}_{(i,y)}} \rightmerge \horizontalunit{(i\shortplus 1, y)} \right) \rightmerge \twobytwounit{(i,y)}\right) \\
    &=
\left(
    \mym{\centeredTikZ{\squaresempty{2}{2}{-2}{0};\squaresempty{2}{1}{0}{0};}_{(i\shortplus 1,y)}} \rightmerge \twobytwounit{(i,y)}\right)\\
    &=
    \mym{\centeredTikZ{\squaresempty{3}{2}{-2}{0};\squaresempty{2}{1}{0}{0};}_{(i\shortplus 1,y)}}\\
    &=
    \mym{\centeredTikZ{\squaresempty{3}{2}{0}{0}}_{(i,y)}} \rightmerge
    \mym{\centeredTikZ{\squaresempty{2}{1}{0}{0};\squaresempty{1}{1}{0}{1};}_{(i\shortplus 1,y)}},
\end{aligned}
\end{equation}
using the top of Lemma~\ref{lemma:prop2to1_rectangle2}, the top of Eq.~\eqref{eq:identities_composite1}, the bottom of Eq.~\eqref{eq:identities_composite2}, and the top of Eq.~\eqref{eq:identities_composite2} in sequence.

Therefore,
\begin{equation}
\begin{aligned}
    \mathfrak{S}_{i+1} &= \left(\mym{\centeredTikZ{\squaresempty{3}{2}{0}{0}}_{(i,y)}} \rightmerge
    \mym{\centeredTikZ{\squaresempty{2}{1}{0}{0};\squaresempty{1}{1}{0}{1};}_{(i\shortplus 1,y)}}\right) \rightmerge \adulthorizontalsnake{1}{y}{i\shortminus 2}{y\shortplus 1} \\
    &=
    \left(\mym{\centeredTikZ{\squaresempty{3}{2}{0}{0}}_{(i,y)}} \rightmerge \adulthorizontalsnake{1}{y}{i\shortminus 2}{y\shortplus 1}\right) \rightmerge \mym{\centeredTikZ{\squaresempty{2}{1}{0}{0};\squaresempty{1}{1}{0}{1};}_{(i\shortplus 1,y)}},
\end{aligned}
\end{equation}
using the commutation lemma (Lemma~\ref{lemma:commutation}). Using the \hyperref[constraints:cell1]{Type I-a) cluster constraint} and the splitting lemma (Lemma~\ref{lemma:snake_splitting}), we conclude
\begin{equation}
    \mathfrak{S}_{i+1} = \adulthorizontalsnake{1}{y}{i}{y\shortplus 1} \rightmerge \mym{\centeredTikZ{\squaresempty{2}{1}{0}{0};\squaresempty{1}{1}{0}{1};}_{(i\shortplus 1, y)}},
\end{equation}
proving the recursion relation.

Using Eq.~\eqref{eq:859}, we have

\begin{equation}
\begin{aligned}
\mathcal{E}_{y,\uparrow} \left(\horizontalbaby{(1,y)}{(N,y)} \right)
    &=
    \mathfrak{S}_{N} \rightmerge \twobytwounit{(N,y)}\\
    &= \left(\adulthorizontalsnake{1}{y}{N\shortminus 1}{y\shortplus 1} \rightmerge \mym{\centeredTikZ{\squaresempty{2}{1}{0}{0};\squaresempty{1}{1}{0}{1};}_{(N, y)}} \right) \rightmerge \twobytwounit{(N,y)}\\
    &= \left(\left(\twobytwounit{(N\shortminus 1, y)} \rightmerge \adulthorizontalsnake{1}{y}{N\shortminus 2}{y\shortplus 1} \right)\rightmerge \mym{\centeredTikZ{\squaresempty{2}{1}{0}{0};\squaresempty{1}{1}{0}{1};}_{(N, y)}} \right) \rightmerge \twobytwounit{(N,y)} \\
    &= \left(\left(\twobytwounit{(N\shortminus 1, y)} \rightmerge \mym{\centeredTikZ{\squaresempty{2}{1}{0}{0};\squaresempty{1}{1}{0}{1};}_{(N, y)}} \right) \rightmerge \twobytwounit{(N,y)}\right) 
    \rightmerge \adulthorizontalsnake{1}{y}{N\shortminus 2}{y\shortplus 1},
\end{aligned}
\end{equation}

using the splitting lemma (Lemma~\ref{lemma:snake_splitting}) in the third line and the commutation lemma (Lemma~\ref{lemma:commutation}) in the last line. Moreover,
\begin{equation}
    \begin{aligned}
    &\left(\twobytwounit{(N\shortminus 1, y)} \rightmerge \mym{\centeredTikZ{\squaresempty{2}{1}{0}{0};\squaresempty{1}{1}{0}{1};}_{(N, y)}} \right) \rightmerge \twobytwounit{(N,y)} \\
    &= 
  \mym{\centeredTikZ{\squaresempty{3}{1}{0}{0};\squaresempty{2}{1}{0}{1};}_{(N, y)}}  \rightmerge \twobytwounit{(N,y)}\\
    &= \mym{\centeredTikZ{\squaresempty{3}{2}{0}{0};}_{(N,y)}},
    \end{aligned}
\end{equation}
using the top of Lemma~\ref{lemma:prop2to1_rectangle2} in the first line and using Lemma~\ref{lemma:prop1to2_basic3} in the second line. Applying the splitting lemma (Lemma~\ref{lemma:snake_splitting}), the first identity of our main claim is proved.

The second identity follows from an analogous analysis, by rotating every constraints/subsystems by $\pi$.
\end{proof}

\section{Twist}
\label{appendix:twist}
In this appendix, we prove the following \emph{twist identity}.
\twist*
\noindent
The proof of the twist identity involves the following new composite object:
\begin{definition}
\begin{equation}
    \centeredTikZ{\squaresdotted{4}{3}{0}{0};\squaresempty{2}{3}{1}{0};\squaresempty{1}{2}{0}{1};\squaresempty{1}{2}{3}{0};}:=
    \centeredTikZ{\squaresdotted{4}{3}{0}{0};\squaresempty{2}{3}{1}{0};\squaresempty{1}{2}{0}{1};} \rightmerge
    \centeredTikZ{\squaresdotted{4}{3}{0}{0};\squaresempty{2}{2}{2}{0};}.
\end{equation}
\label{def:twist_transfer}
\end{definition}

The key identity is the following:
\begin{equation}
\begin{aligned}
    \left(\centeredTikZ{\squaresdotted{5}{3}{0}{0};\squaresempty{2}{3}{1}{0};\squaresempty{1}{2}{0}{1};\squaresempty{1}{2}{3}{0};}
    \rightmerge 
    \centeredTikZ{\squaresdotted{5}{3}{0}{0};\squaresempty{2}{2}{3}{0};}\right)
    \rightmerge
    \centeredTikZ{\squaresdotted{5}{3}{0}{0};\squaresempty{2}{2}{2}{1};} \\ = 
     \left(\centeredTikZ{\squaresdotted{5}{3}{-1}{0};\squaresempty{2}{3}{1}{0};\squaresempty{1}{2}{0}{1};\squaresempty{1}{2}{3}{0};}
    \rightmerge 
    \centeredTikZ{\squaresdotted{5}{3}{0}{0};\squaresempty{2}{2}{0}{1};}\right)
    \rightmerge
    \centeredTikZ{\squaresdotted{5}{3}{0}{0};\squaresempty{2}{2}{1}{0};}.\label{eq:twist_key}
\end{aligned}
\end{equation}
Roughly speaking, we can use Eq.~\eqref{eq:twist_key} in the following way. First, expand the left hand side of Eq.~\eqref{eq:twist_main_result} by writing down the action of $\mathcal{E}_{y, \uparrow}$ explicitly and also expressing the snake as a sequence of right-merges, involving marginals from the left end to the right end sequentially. One can apply Eq.~\eqref{eq:twist_key} from the left end to the right end on this object. After applying some extra identities at the end, one can show that the end result is equal to the right hand side of Eq.~\eqref{eq:twist_main_result}. This fact can be verified by expanding the right hand side in a way similar to the left hand side, by explicitly writing down the action of $\mathcal{E}_{y,\downarrow}$ and expressing the snake as a sequence of right-merges, involving marginals from the \emph{right to the left}. 

As we shall see later, it is not too difficult to show that Eq.~\eqref{eq:twist_key} would follow if these ``commutation relations'' hold:

\begin{equation}
\begin{aligned}
\left(
\centeredTikZ{\squaresdotted{4}{3}{0}{0};\squaresempty{2}{3}{0}{0};\squaresempty{1}{2}{2}{0};}
\rightmerge 
\centeredTikZ{\squaresdotted{4}{3}{0}{0};\squaresempty{2}{2}{2}{0};}\right)
\rightmerge
\centeredTikZ{\squaresdotted{4}{3}{0}{0};\squaresempty{2}{2}{1}{1};}
= 
\left(
\centeredTikZ{\squaresdotted{4}{3}{0}{0};\squaresempty{2}{3}{0}{0};\squaresempty{1}{2}{2}{0};}
\rightmerge 
\centeredTikZ{\squaresdotted{4}{3}{0}{0};\squaresempty{2}{2}{1}{1};}
\right)
\rightmerge
\centeredTikZ{\squaresdotted{4}{3}{0}{0};\squaresempty{2}{2}{2}{0};},\\
\left(
\centeredTikZ{\squaresdotted{4}{3}{0}{0};\squaresempty{2}{3}{2}{0};\squaresempty{1}{2}{1}{1};}
\rightmerge 
\centeredTikZ{\squaresdotted{4}{3}{0}{0};\squaresempty{2}{2}{0}{1};}\right)
\rightmerge
\centeredTikZ{\squaresdotted{4}{3}{0}{0};\squaresempty{2}{2}{1}{0};}
=
\left(
\centeredTikZ{\squaresdotted{4}{3}{0}{0};\squaresempty{2}{3}{2}{0};\squaresempty{1}{2}{1}{1};}
\rightmerge 
\centeredTikZ{\squaresdotted{4}{3}{0}{0};\squaresempty{2}{2}{1}{0};}
\right)
\rightmerge
\centeredTikZ{\squaresdotted{4}{3}{0}{0};\squaresempty{2}{2}{0}{1};},
\end{aligned}
\label{eq:255}
\end{equation}

While these identities are indeed correct, they do not immediately follow from the commutation lemma (Lemma~\ref{lemma:commutation}). This is because the commutation lemma requires two right-merged marginals to have disjoint supports. 

In the rest of this appendix, we will prove Eqs.~\eqref{eq:255} in Appendix~\ref{appendix:twist_commutation}. We will then complete the proof of Proposition~\ref{prop:twist} in Appendix~\ref{appendix:twist_complete}.

\subsection{Commutation relations}
\label{appendix:twist_commutation}
It will be convenient to work with the following objects.
\begin{definition}
\begin{equation}
\begin{aligned}
\centeredTikZ{\squaresdotted{4}{3}{0}{0}; \squaresempty{4}{2}{0}{0}; \squaresempty{3}{1}{0}{2};}&:=
    \left(
\centeredTikZ{\squaresdotted{4}{3}{0}{0};\squaresempty{2}{3}{0}{0};\squaresempty{1}{2}{2}{0};}
\rightmerge 
\centeredTikZ{\squaresdotted{4}{3}{0}{0};\squaresempty{2}{2}{2}{0};}\right)
\rightmerge
\centeredTikZ{\squaresdotted{4}{3}{0}{0};\squaresempty{2}{2}{1}{1};}, \\
\centeredTikZ{\squaresempty{4}{3}{0}{0};\squaresdotted{1}{1}{0}{0};} &:= \left(
\centeredTikZ{\squaresdotted{4}{3}{0}{0};\squaresempty{2}{3}{2}{0};\squaresempty{1}{2}{1}{1};} \rightmerge \centeredTikZ{\squaresdotted{4}{3}{0}{0};\squaresempty{2}{2}{0}{1};}
\right) \rightmerge
\centeredTikZ{\squaresdotted{4}{3}{0}{0};\squaresempty{2}{2}{1}{0};}.
\end{aligned}
\end{equation}
\end{definition}

\begin{lemma}
\label{lemma:twist_commutation1}
\begin{equation}
    \centeredTikZ{\squaresdotted{4}{3}{0}{0}; \squaresempty{4}{2}{0}{0}; \squaresempty{3}{1}{0}{2};\squaresfilled{1}{3}{0}{0};\squaresfilled{1}{1}{1}{0};}=
    \centeredTikZ{\squaresdotted{4}{3}{0}{0}; \squaresempty{4}{2}{0}{0}; \squaresempty{3}{1}{0}{2};\squaresdotted{1}{3}{0}{0};\squaresdotted{1}{1}{1}{0};}\label{eq:lemma7_first}
\end{equation}
\end{lemma}
\begin{proof}
To prove Eq.~\eqref{eq:lemma7_first}, note that
\begin{equation}
\begin{aligned}
\centeredTikZ{\squaresdotted{4}{3}{0}{0}; \squaresempty{4}{2}{0}{0}; \squaresempty{3}{1}{0}{2};\squaresfilled{1}{3}{0}{0};}&= \left(\centeredTikZ{\squaresdotted{4}{3}{0}{0}; \squaresempty{3}{2}{0}{0};\squaresempty{2}{1}{0}{2};\squaresfilled{1}{3}{0}{0};} \rightmerge \centeredTikZ{\squaresdotted{4}{3}{0}{0};\squaresempty{2}{2}{2}{0};}\right) \rightmerge \centeredTikZ{\squaresdotted{4}{3}{0}{0};\squaresempty{2}{2}{1}{1};} \\
&= \centeredTikZ{\squaresdotted{4}{3}{0}{0}; \squaresempty{3}{2}{1}{0};\squaresempty{1}{1}{1}{2};} \rightmerge \centeredTikZ{\squaresdotted{4}{3}{0}{0};\squaresempty{2}{2}{1}{1};}.
\end{aligned}
\end{equation}
To see why, note that
\begin{equation}
    \entropy{\centeredTikZ{\squaresempty{3}{2}{0}{0};\squaresempty{1}{1}{0}{2};}} = 
    \entropy{\centeredTikZ{\squaresempty{2}{2}{0}{0};\squaresempty{1}{1}{0}{2};}}
    +
    \entropy{\centeredTikZ{\squaresempty{2}{2}{0}{0};}}
    -
    \entropy{\centeredTikZ{\squaresempty{1}{2}{0}{0};}},
\end{equation}
which follows from \hyperref[constraints:cell1]{Type I-b) cluster constraint} and the monotonicity of the conditional mutual information (Eq.~\eqref{eq:ssa_monotinicity_intro}). Then by Petz's theorem (Theorem~\ref{thm:Petz}), 
\begin{equation}
    \centeredTikZ{\squaresdotted{4}{3}{0}{0}; \squaresempty{3}{2}{1}{0};\squaresempty{1}{1}{1}{2};} 
    = \centeredTikZ{\squaresdotted{4}{3}{0}{0}; \squaresempty{3}{2}{0}{0};\squaresempty{2}{1}{0}{2};\squaresdotted{1}{3}{0}{0};} \rightmerge \centeredTikZ{\squaresdotted{4}{3}{0}{0};\squaresempty{2}{2}{2}{0};}.
\end{equation}
Moreover, note that 
\begin{equation}
    \entropy{\centeredTikZ{\squaresempty{3}{2}{0}{0};\squaresempty{2}{1}{0}{2};}}
     = \entropy{\centeredTikZ{\squaresempty{3}{2}{0}{0};\squaresempty{1}{1}{0}{2};}} + \entropy{\centeredTikZ{\squaresempty{2}{2}{0}{0};}} - \entropy{\centeredTikZ{\squaresempty{2}{1}{0}{0};\squaresempty{1}{1}{0}{1};}},
\end{equation}
which follows from the \hyperref[constraints:cell1]{Type I-b) cluster constraint} and Eq.~\eqref{eq:ssa_monotonicity_intro2}. By Petz's theorem (Theorem~\ref{thm:Petz}), 
\begin{equation}
\centeredTikZ{\squaresdotted{4}{3}{0}{0}; \squaresempty{3}{2}{1}{0};\squaresempty{1}{1}{1}{2};} \rightmerge \centeredTikZ{\squaresdotted{4}{3}{0}{0};\squaresempty{2}{2}{1}{1};}=\centeredTikZ{\squaresdotted{4}{3}{0}{0};\squaresempty{3}{3}{1}{0};\squaresdotted{1}{1}{3}{2};},
\end{equation}
which immediately implies Eq.~\eqref{eq:lemma7_first}.
\end{proof}

\begin{lemma}
\label{lemma:twist_commutation2}
\begin{equation}
    \centeredTikZ{\squaresdotted{4}{3}{0}{0}; \squaresempty{4}{2}{0}{0}; \squaresempty{3}{1}{0}{2};\squaresfilled{1}{2}{3}{0};\squaresfilled{1}{1}{2}{2};}
    =
    \centeredTikZ{\squaresdotted{4}{3}{0}{0}; \squaresempty{4}{2}{0}{0}; \squaresempty{3}{1}{0}{2};\squaresdotted{1}{2}{3}{0};\squaresdotted{1}{1}{2}{2};}.\label{eq:lemma7_second}
\end{equation}
\end{lemma}
\begin{proof}
To prove Eq.~\eqref{eq:lemma7_second}, recall the \hyperref[constraints:cell1]{Type I-c) cluster constraint}, which through Lemma~\ref{lemma:fundamental} implies
\begin{equation}
    \centeredTikZ{\squaresdotted{4}{3}{0}{0};\squaresempty{3}{2}{0}{0};\squaresempty{2}{1}{0}{2};}
= \centeredTikZ{\squaresdotted{4}{3}{0}{0};\squaresempty{2}{3}{0}{0};\squaresdotted{1}{2}{2}{0};}\maxmerge
\centeredTikZ{\squaresdotted{4}{3}{0}{0};\squaresdotted{2}{3}{0}{0};\squaresempty{2}{2}{1}{0};}=
 \centeredTikZ{\squaresdotted{4}{3}{0}{0};\squaresempty{2}{3}{0}{0};\squaresdotted{1}{2}{2}{0};}\rightmerge
\centeredTikZ{\squaresdotted{4}{3}{0}{0};\squaresdotted{2}{3}{0}{0};\squaresempty{2}{2}{1}{0};}. \label{eq:lemma7_temp3}
\end{equation}
Moreover, by the \hyperref[constraints:cell1]{Type I-a) cluster constraint}, through Lemma~\ref{lemma:fundamental},
\begin{equation}
    \centeredTikZ{\squaresdotted{4}{3}{0}{0}; \squaresempty{3}{2}{1}{0};}
    = \centeredTikZ{\squaresdotted{4}{3}{0}{0}; \squaresempty{2}{2}{1}{0};} \maxmerge
    \centeredTikZ{\squaresdotted{4}{3}{0}{0}; \squaresempty{2}{2}{2}{0};}
    = \centeredTikZ{\squaresdotted{4}{3}{0}{0}; \squaresempty{2}{2}{1}{0};} \rightmerge
    \centeredTikZ{\squaresdotted{4}{3}{0}{0}; \squaresempty{2}{2}{2}{0};}.\label{eq:lemma7_temp4}
\end{equation}
Therefore, Eq.~\eqref{eq:lemma7_temp3} and Eq.~\eqref{eq:lemma7_temp4} can be merged together by using the merging lemma (Lemma~\ref{lemma:merging_lemma}). Moreover, by the same lemma, the merged state is precisely the following state:
\begin{equation}
    \centeredTikZ{\squaresdotted{4}{3}{0}{0}; \squaresempty{3}{2}{0}{0};\squaresempty{2}{1}{0}{2};} \maxmerge \centeredTikZ{\squaresdotted{4}{3}{0}{0};\squaresempty{2}{2}{2}{0};} 
    =\centeredTikZ{\squaresdotted{4}{3}{0}{0}; \squaresempty{3}{2}{0}{0};\squaresempty{2}{1}{0}{2};} \rightmerge \centeredTikZ{\squaresdotted{4}{3}{0}{0};\squaresempty{2}{2}{2}{0};}
    \label{eq:lemma7_temptemp}
\end{equation}
Therefore, 
\begin{equation}
\begin{aligned}
\centeredTikZ{\squaresdotted{4}{3}{0}{0}; \squaresempty{4}{2}{0}{0}; \squaresempty{3}{1}{0}{2};\squaresfilled{1}{2}{3}{0};} &= \centeredTikZ{\squaresdotted{4}{3}{0}{0}; \squaresempty{3}{2}{0}{0};\squaresempty{2}{1}{0}{2};}  \rightmerge \centeredTikZ{\squaresdotted{4}{3}{0}{0};\squaresempty{2}{2}{1}{1};} \\
&=
\centeredTikZ{\squaresdotted{4}{3}{0}{0};\squaresempty{3}{3}{0}{0};}, \label{eq:lemma7_temp6}
\end{aligned}
\end{equation}
where the first line follows from the above-mentioned merging process\footnote{To be more specific, the merged state is consistent with the two density matrices in Eq.~\eqref{eq:lemma7_temptemp}. As such, by taking a partial trace, we obtain (one of the) original marginals.} and the second line follows from \hyperref[constraints:cell1]{Type I-d) cluster constraint} and Lemma~\ref{lemma:fundamental}.
\end{proof}

\begin{lemma}
\label{lemma:twist_commutation_entropy}
\begin{equation}
    \entropy{\centeredTikZ{\squaresempty{4}{3}{0}{0};\squaresdotted{1}{1}{3}{2};}}
    =\entropy{\centeredTikZ{\squaresempty{3}{3}{0}{0};}} + \entropy{\centeredTikZ{\squaresempty{2}{2}{0}{0};}}
    -\entropy{\centeredTikZ{\squaresempty{1}{2}{0}{0};}}
\end{equation}
\end{lemma}
\begin{proof}
Let us recall the definition.
\begin{equation}
    \centeredTikZ{\squaresdotted{4}{3}{0}{0}; \squaresempty{4}{2}{0}{0}; \squaresempty{3}{1}{0}{2};}:=
    \left(
\centeredTikZ{\squaresdotted{4}{3}{0}{0};\squaresempty{2}{3}{0}{0};\squaresempty{1}{2}{2}{0};}
\rightmerge 
\centeredTikZ{\squaresdotted{4}{3}{0}{0};\squaresempty{2}{2}{2}{0};}\right)
\rightmerge
\centeredTikZ{\squaresdotted{4}{3}{0}{0};\squaresempty{2}{2}{1}{1};}.\label{eq:216}
\end{equation}

Viewing the two right-merges in Eq.~\eqref{eq:216} together as a quantum channel, we can use Lemma~\ref{lemma:twist_commutation1} and Lemma~\ref{lemma:twist_commutation2} can be applied to the seventh condition of Lemma~\ref{lemma:fundamental}. To be more concrete, these two lemmas imply that\footnote{We do not actually use this equation. The purpose of this equation is just to specify the involved subsystems.}
\begin{equation}
    \centeredTikZ{\squaresdotted{4}{3}{0}{0}; \squaresempty{4}{2}{0}{0}; \squaresempty{3}{1}{0}{2};}
    =
    \centeredTikZ{\squaresdotted{4}{3}{0}{0}; \squaresempty{4}{2}{0}{0}; \squaresempty{3}{1}{0}{2};\squaresdotted{1}{2}{3}{0};\squaresdotted{1}{1}{2}{2};}
    \maxmerge 
    \centeredTikZ{\squaresdotted{4}{3}{0}{0}; \squaresempty{4}{2}{0}{0}; \squaresempty{3}{1}{0}{2};\squaresdotted{1}{3}{0}{0};\squaresdotted{1}{1}{1}{0};}
    =
    \centeredTikZ{\squaresdotted{4}{3}{0}{0}; \squaresempty{4}{2}{0}{0}; \squaresempty{3}{1}{0}{2};\squaresdotted{1}{2}{3}{0};\squaresdotted{1}{1}{2}{2};}
    \rightmerge 
    \centeredTikZ{\squaresdotted{4}{3}{0}{0}; \squaresempty{4}{2}{0}{0}; \squaresempty{3}{1}{0}{2};\squaresdotted{1}{3}{0}{0};\squaresdotted{1}{1}{1}{0};}.
\end{equation}
Using the first condition of Lemma~\ref{lemma:fundamental},

\begin{equation}
\begin{aligned}
    \entropy{\centeredTikZ{\squaresempty{4}{3}{0}{0};\squaresdotted{1}{1}{3}{2};}} &= 
    \entropy{\centeredTikZ{\squaresempty{3}{2}{0}{0};\squaresempty{2}{1}{0}{2};}} + 
    \entropy{\centeredTikZ{\squaresempty{2}{2}{1}{0};\squaresempty{2}{2}{0}{1};}} - \entropy{\centeredTikZ{\squaresempty{1}{2}{2}{0};\squaresempty{1}{2}{1}{1};}} \\
    &= \entropy{\centeredTikZ{\squaresempty{3}{3}{0}{0};}} +
    \entropy{\centeredTikZ{\squaresempty{2}{1}{0}{0};\squaresempty{1}{1}{0}{1}}} - 
    \entropy{\centeredTikZ{\squaresempty{2}{2}{0}{0};}}
    + 
    \entropy{\centeredTikZ{\squaresempty{2}{2}{1}{0};\squaresempty{2}{2}{0}{1};}} - \entropy{\centeredTikZ{\squaresempty{1}{2}{2}{0};\squaresempty{1}{2}{1}{1};}}
\end{aligned}
\label{eq:lemma7_temp5}
\end{equation}

The last two terms in the second line of Eq.~\eqref{eq:lemma7_temp5} can be broken down further. First, applying the monotonicity of conditional mutual information (Eq.~\eqref{eq:ssa_monotinicity_intro}) to the \hyperref[constraints:cell2]{Type II-d) cluster constraint},\footnote{Specifically, we are removing a $1\times 1$ cluster on the top right corner.} 
\begin{equation}
\begin{aligned}
\entropy{\centeredTikZ{\squaresempty{2}{2}{1}{0};\squaresempty{2}{2}{0}{1};}}
&= \entropy{\centeredTikZ{\squaresempty{3}{2}{0}{0};\squaresempty{2}{2}{0}{1};}} - 
\entropy{\centeredTikZ{\squaresempty{2}{2}{0}{0};}}
 + \entropy{\centeredTikZ{\squaresempty{2}{1}{0}{1};\squaresempty{1}{1}{1}{0};}} \\
&= 2\entropy{\centeredTikZ{\squaresempty{2}{2}{0}{0};}}- \entropy{\centeredTikZ{\squaresempty{1}{1}{0}{0};}}.
\end{aligned}
\end{equation}
Second, applying the monotonicity of conditional mutual information (Eq.~\eqref{eq:ssa_monotinicity_intro}) to the \hyperref[constraints:cell2]{Type II-a) cluster constraint},\footnote{Specifically, we are removing one $1\times 1$ cluster from the top row and another $1\times 1$ cluster from the bottom row of a $2\times 3$ cluster.}
\begin{equation}
\begin{aligned}
\entropy{\centeredTikZ{\squaresempty{1}{2}{2}{0};\squaresempty{1}{2}{1}{1};}} 
&= \entropy{\centeredTikZ{\squaresempty{2}{1}{0}{0};\squaresempty{1}{1}{0}{1};}} +
\entropy{\centeredTikZ{\squaresempty{2}{1}{0}{1};\squaresempty{1}{1}{1}{0};}}
-\entropy{\centeredTikZ{\squaresempty{2}{1}{0}{0};}}.
\end{aligned}
\end{equation}

Plugging in these results, we get
\begin{equation}
    \entropy{\centeredTikZ{\squaresempty{4}{3}{0}{0};\squaresdotted{1}{1}{3}{2};}}
    =\entropy{\centeredTikZ{\squaresempty{3}{3}{0}{0};}} + \entropy{\centeredTikZ{\squaresempty{2}{2}{0}{0};}}
    -\entropy{\centeredTikZ{\squaresempty{1}{2}{0}{0};}}
\end{equation}
\end{proof}

\begin{proposition}
\begin{equation}
\centeredTikZ{\squaresempty{4}{3}{0}{0};\squaresdotted{1}{1}{3}{2};}
=
\centeredTikZ{\squaresdotted{4}{3}{0}{0};\squaresempty{3}{3}{0}{0};}
\rightmerge
\centeredTikZ{\squaresdotted{4}{3}{0}{0};\squaresempty{2}{2}{2}{0};}.
\end{equation}
\label{prop:336}
\end{proposition}
\begin{proof}
From Lemma~\ref{lemma:twist_commutation1}, we see that
\begin{equation}
\centeredTikZ{\squaresdotted{4}{3}{0}{0};
\squaresfilled{2}{3}{0}{0};
\squaresempty{2}{2}{2}{0};
\squaresfilled{1}{1}{2}{2};
}
=
\centeredTikZ{\squaresdotted{4}{3}{0}{0};
\squaresempty{2}{2}{2}{0};
}.
\end{equation}
Moreover, in Eq.~\eqref{eq:lemma7_temp6}, we have proved that
\begin{equation}
    \centeredTikZ{\squaresdotted{4}{3}{0}{0}; \squaresempty{4}{2}{0}{0}; \squaresempty{3}{1}{0}{2};\squaresfilled{1}{2}{3}{0};} =
\centeredTikZ{\squaresdotted{4}{3}{0}{0};\squaresempty{3}{3}{0}{0};}.
\end{equation}
These two identities, together with Lemma~\ref{lemma:twist_commutation_entropy}, can be applied to the first condition of Lemma~\ref{lemma:fundamental}. The main claim follows from the equivalence of the first and the third condition of Lemma~\ref{lemma:fundamental}. 
\end{proof}

After applying the \hyperref[constraints:cell1]{Type I-d) cluster constraint}, we obtain the following ``commutation relation.''
\begin{proposition}
\begin{equation}
\begin{aligned}
\left(
\centeredTikZ{\squaresdotted{4}{3}{0}{0};\squaresempty{2}{3}{0}{0};\squaresempty{1}{2}{2}{0};}
\rightmerge 
\centeredTikZ{\squaresdotted{4}{3}{0}{0};\squaresempty{2}{2}{2}{0};}\right)
\rightmerge
\centeredTikZ{\squaresdotted{4}{3}{0}{0};\squaresempty{2}{2}{1}{1};}\\
=
\left(
\centeredTikZ{\squaresdotted{4}{3}{0}{0};\squaresempty{2}{3}{0}{0};\squaresempty{1}{2}{2}{0};}
\rightmerge 
\centeredTikZ{\squaresdotted{4}{3}{0}{0};\squaresempty{2}{2}{1}{1};}
\right)
\rightmerge
\centeredTikZ{\squaresdotted{4}{3}{0}{0};\squaresempty{2}{2}{2}{0};}.
\end{aligned}
\end{equation}
\label{prop:twist_key1}
\end{proposition}

As before, a $\pi$-rotated version of the entire proof can be reproduced in an analogous way.
\begin{proposition}
\begin{equation}
\centeredTikZ{\squaresempty{4}{3}{0}{0};\squaresdotted{1}{1}{0}{0};}
=
\centeredTikZ{\squaresdotted{4}{3}{0}{0};\squaresempty{3}{3}{1}{0};}
\rightmerge
\centeredTikZ{\squaresdotted{4}{3}{0}{0};\squaresempty{2}{2}{0}{1};}.
\end{equation}
\label{prop:336_1}
\end{proposition}

\begin{proposition}
\begin{equation}
\begin{aligned}
\left(
\centeredTikZ{\squaresdotted{4}{3}{0}{0};\squaresempty{2}{3}{2}{0};\squaresempty{1}{2}{1}{1};}
\rightmerge 
\centeredTikZ{\squaresdotted{4}{3}{0}{0};\squaresempty{2}{2}{0}{1};}\right)
\rightmerge
\centeredTikZ{\squaresdotted{4}{3}{0}{0};\squaresempty{2}{2}{1}{0};} \\
=
\left(
\centeredTikZ{\squaresdotted{4}{3}{0}{0};\squaresempty{2}{3}{2}{0};\squaresempty{1}{2}{1}{1};}
\rightmerge 
\centeredTikZ{\squaresdotted{4}{3}{0}{0};\squaresempty{2}{2}{1}{0};}
\right)
\rightmerge
\centeredTikZ{\squaresdotted{4}{3}{0}{0};\squaresempty{2}{2}{0}{1};}
\end{aligned}
\end{equation}
\label{prop:twist_key2}
\end{proposition}

\subsection{Completing the proof}
\label{appendix:twist_complete}
\begin{lemma}
\begin{equation}
\begin{aligned}
    \left(\centeredTikZ{\squaresdotted{5}{3}{0}{0};\squaresempty{2}{3}{1}{0};\squaresempty{1}{2}{0}{1};\squaresempty{1}{2}{3}{0};}
    \rightmerge 
    \centeredTikZ{\squaresdotted{5}{3}{0}{0};\squaresempty{2}{2}{3}{0};}\right)
    \rightmerge
    \centeredTikZ{\squaresdotted{5}{3}{0}{0};\squaresempty{2}{2}{2}{1};} \\= 
     \left(\centeredTikZ{\squaresdotted{5}{3}{-1}{0};\squaresempty{2}{3}{1}{0};\squaresempty{1}{2}{0}{1};\squaresempty{1}{2}{3}{0};}
    \rightmerge 
    \centeredTikZ{\squaresdotted{5}{3}{0}{0};\squaresempty{2}{2}{0}{1};}\right)
    \rightmerge
    \centeredTikZ{\squaresdotted{5}{3}{0}{0};\squaresempty{2}{2}{1}{0};}.\label{eq:338}
\end{aligned}
\end{equation}
\label{lemma:twist_transfer}
\end{lemma}
\begin{proof}
Note that
\begin{equation}
\begin{aligned}
\centeredTikZ{\squaresdotted{5}{3}{0}{0};\squaresempty{2}{3}{1}{0};\squaresempty{1}{2}{0}{1};\squaresempty{1}{2}{3}{0};}
&=
\centeredTikZ{\squaresdotted{5}{3}{0}{0};\squaresempty{2}{3}{1}{0};\squaresempty{1}{2}{0}{1};}
\rightmerge
\centeredTikZ{\squaresdotted{5}{3}{0}{0};\squaresempty{2}{2}{2}{0};}\\
&= 
\left(
\centeredTikZ{\squaresdotted{5}{3}{0}{0};\squaresempty{2}{3}{1}{0};}
\rightmerge
\centeredTikZ{\squaresdotted{5}{3}{0}{0};\squaresempty{2}{2}{0}{1};}
\right)
\rightmerge
\centeredTikZ{\squaresdotted{5}{3}{0}{0};\squaresempty{2}{2}{2}{0};}\\
&= 
\left(
\centeredTikZ{\squaresdotted{5}{3}{0}{0};\squaresempty{2}{3}{1}{0};}
\rightmerge
\centeredTikZ{\squaresdotted{5}{3}{0}{0};\squaresempty{2}{2}{2}{0};}
\right)
\rightmerge
\centeredTikZ{\squaresdotted{5}{3}{0}{0};\squaresempty{2}{2}{0}{1};}\\
&=
\centeredTikZ{\squaresdotted{5}{3}{0}{0};\squaresempty{2}{3}{1}{0};\squaresempty{1}{2}{3}{0};}
\rightmerge
\centeredTikZ{\squaresdotted{5}{3}{0}{0};\squaresempty{2}{2}{0}{1};}.
\end{aligned}
\label{eq:334}
\end{equation}
The first line is the definition. In the second line, we applied the \hyperref[constraints:cell2]{Type II-c) cluster constraint} to Lemma~\ref{lemma:fundamental}. In the third line, we used the commutation lemma (Lemma~\ref{lemma:commutation}). In the last line, we applied the \hyperref[constraints:cell1]{Type I-c) cluster constraint} to Lemma~\ref{lemma:fundamental}.

Therefore,
\begin{equation}
\begin{aligned}
     &\left(\centeredTikZ{\squaresdotted{5}{3}{0}{0};\squaresempty{2}{3}{1}{0};\squaresempty{1}{2}{0}{1};\squaresempty{1}{2}{3}{0};}
    \rightmerge 
    \centeredTikZ{\squaresdotted{5}{3}{0}{0};\squaresempty{2}{2}{3}{0};}\right)
    \rightmerge
    \centeredTikZ{\squaresdotted{5}{3}{0}{0};\squaresempty{2}{2}{2}{1};}\\
    &=
    \left(
    \left( \centeredTikZ{\squaresdotted{5}{3}{0}{0};\squaresempty{2}{3}{1}{0};\squaresempty{1}{2}{3}{0};}
\rightmerge
\centeredTikZ{\squaresdotted{5}{3}{0}{0};\squaresempty{2}{2}{0}{1};}\right) \rightmerge 
    \centeredTikZ{\squaresdotted{5}{3}{0}{0};\squaresempty{2}{2}{3}{0};}\right)
    \rightmerge
    \centeredTikZ{\squaresdotted{5}{3}{0}{0};\squaresempty{2}{2}{2}{1};}\\
&=
 \left(
    \left( \centeredTikZ{\squaresdotted{5}{3}{0}{0};\squaresempty{2}{3}{1}{0};\squaresempty{1}{2}{3}{0};}
\rightmerge
    \centeredTikZ{\squaresdotted{5}{3}{0}{0};\squaresempty{2}{2}{3}{0};}\right)
    \rightmerge
    \centeredTikZ{\squaresdotted{5}{3}{0}{0};\squaresempty{2}{2}{2}{1};}\right) \rightmerge
    \centeredTikZ{\squaresdotted{5}{3}{0}{0};\squaresempty{2}{2}{0}{1};} \\
&=
\left(
\centeredTikZ{\squaresdotted{5}{3}{0}{0};\squaresempty{3}{3}{1}{0};}
\rightmerge
\centeredTikZ{\squaresdotted{5}{3}{0}{0};\squaresempty{2}{2}{3}{0};}
\right)\rightmerge
    \centeredTikZ{\squaresdotted{5}{3}{0}{0};\squaresempty{2}{2}{0}{1};}.
\end{aligned}
\label{eq:340}
\end{equation}
In the second line, we have plugged in the result of Eq.~\eqref{eq:334}. In the third line, we used the commutation lemma (Lemma~\ref{lemma:commutation}). In the last line, we have used Proposition~\ref{prop:twist_key1}.

One can apply a similar logic to the right hand side of Eq.~\eqref{eq:338} by rotating all the subsystems/identities by $\pi$. (For instance, Proposition~\ref{prop:twist_key1} would change to Proposition~\ref{prop:twist_key2}.) The end result is the same as the last line of Eq.~\eqref{eq:340}, completing the proof.
\end{proof}

Now, we are in a position to prove Proposition~\ref{prop:twist}. We restate the result for the reader's convenience.

\twist*
\begin{proof}
Note that
\begin{equation}
    \mathcal{E}_{y,\uparrow}\left(\adulthorizontalsnake{1}{y\shortminus 1}{N}{y} \right)
    = \mathfrak{S}_{N, \rightarrow} \rightmerge \twobytwounit{(N,y)},
    \label{eq:1106}
\end{equation}
where
\begin{equation}
    \mathfrak{S}_{i+1, \rightarrow}
    =
    \left(
    \mathfrak{S}_{i, \rightarrow}
    \rightmerge
    \twobytwounit{(i\shortplus 1,y\shortminus 1)}
    \right)
    \rightmerge
    \twobytwounit{(i, y)}
\end{equation}
for $i$ such that $3<i<N$ and
\begin{equation}
    \mathfrak{S}_{3, \rightarrow}
    :=
    \mym{
    \centeredTikZ{
    \squaresempty{2}{3}{0}{0};\squaresempty{1}{2}{2}{0};
    }_{(3,y\shortminus 1)}
    }
\end{equation}

As a warmup, note that

\begin{equation}
\begin{aligned}
    \mathfrak{S}_{4,\rightarrow}&= 
    \left(
    \mym{
    \centeredTikZ{
    \squaresempty{2}{3}{0}{0};\squaresempty{1}{2}{2}{0};
    }_{(3,y\shortminus 1)}
    }\rightmerge
    \twobytwounit{(4,y\shortminus 1)}
    \right)
    \rightmerge
    \twobytwounit{(3, y)} \\
    &=
    \mym{
    \centeredTikZ{
    \squaresempty{3}{3}{0}{0};}_{(3,y\shortminus 1)}
    }
    \rightmerge
    \twobytwounit{(4, y\shortminus 1)} \\
    &= 
    \left(
    \left(
    \mym{\centeredTikZ{\squaresempty{2}{3}{0}{0};}_{(3,y\shortminus 1)}} \rightmerge \twobytwounit{(2, y)}
    \right) \rightmerge \twobytwounit{(2,y\shortminus 1)}
    \right) \rightmerge \twobytwounit{(4,y\shortminus 1)} \\
    &= \left(\left(
    \mym{\centeredTikZ{\squaresempty{2}{3}{0}{0};}_{(3,y\shortminus 1)}} \rightmerge \twobytwounit{(2, y)}
    \right) \rightmerge \twobytwounit{(4,y\shortminus 1)} \right)
    \rightmerge \twobytwounit{(2,y\shortminus 1)} \\
    &= \left(\mym{\centeredTikZ{\squaresempty{2}{3}{0}{0};\squaresempty{1}{2}{-1}{1};}_{(3,y\shortminus 1)}} \rightmerge \twobytwounit{(4,y\shortminus 1)} \right) \rightmerge \twobytwounit{(2,y\shortminus 1)}.
\end{aligned}
\end{equation}

In the first line, we applied Proposition~\ref{prop:twist_key1} and then applied the \hyperref[constraints:cell1]{Type I-d) cluster constraint} constraint to the first condition of Lemma~\ref{lemma:fundamental}. In the third line, we applied the \hyperref[constraints:cell2]{Type II-c) and d) cluster constraints} to the first condition of Lemma~\ref{lemma:fundamental}. We then used the commutation lemma (Lemma~\ref{lemma:commutation}), followed by the \hyperref[constraints:cell2]{Type II-c) cluster constraint} applied to the first condition of Lemma~\ref{lemma:fundamental}.

From Definition~\ref{def:twist_transfer}, we conclude that
\begin{equation}
    \mathfrak{S}_{4,\rightarrow} = 
    \mym{\centeredTikZ{\squaresempty{2}{3}{1}{0};\squaresempty{1}{2}{0}{1};\squaresempty{1}{2}{3}{0};}_{(4,y\shortminus 1)}} \rightmerge \twobytwounit{(2,y\shortminus 1)}.\label{eq:944}
\end{equation}
Importantly, by the commutation lemma (Lemma~\ref{lemma:commutation}), all the right-merges appearing in the definition of $\mathfrak{S}_{i,\rightarrow}$ for $i\geq 4$ ``commute'' with the right merge by the $2\times 2$ cluster anchored at $(2,y\shortminus 1)$.

From Lemma~\ref{lemma:twist_transfer}, recall

\begin{equation}
    \left(\centeredTikZ{\squaresdotted{5}{3}{0}{0};\squaresempty{2}{3}{1}{0};\squaresempty{1}{2}{0}{1};\squaresempty{1}{2}{3}{0};}
    \rightmerge 
    \centeredTikZ{\squaresdotted{5}{3}{0}{0};\squaresempty{2}{2}{3}{0};}\right)
    \rightmerge
    \centeredTikZ{\squaresdotted{5}{3}{0}{0};\squaresempty{2}{2}{2}{1};} = 
     \left(\centeredTikZ{\squaresdotted{5}{3}{-1}{0};\squaresempty{2}{3}{1}{0};\squaresempty{1}{2}{0}{1};\squaresempty{1}{2}{3}{0};}
    \rightmerge 
    \centeredTikZ{\squaresdotted{5}{3}{0}{0};\squaresempty{2}{2}{0}{1};}\right)
    \rightmerge
    \centeredTikZ{\squaresdotted{5}{3}{0}{0};\squaresempty{2}{2}{1}{0};}.\label{eq:949}
\end{equation}

Applying Eq.~\eqref{eq:949} and the commutation lemma (Lemma~\ref{lemma:commutation}) repeatedly, we obtain
\begin{equation}
    \mathfrak{S}_{N,\rightarrow} =
    \left(\left(\widetilde{\mathfrak{S}} \rightmerge \twobytwounit{(N\shortminus 2, y\shortminus 1)}\right) \cdots \right) \rightmerge \twobytwounit{(2, y\shortminus 1)},
\end{equation}
where

\begin{equation}
    \widetilde{\mathfrak{S}} = \left(\left(\mym{\centeredTikZ{\squaresempty{2}{3}{1}{0};\squaresempty{1}{2}{0}{1};\squaresempty{1}{2}{3}{0};}_{(N,y\shortminus 1)}} \rightmerge \twobytwounit{(N\shortminus 3, y)}\right) \cdots \right)\rightmerge \twobytwounit{(2, y)}.
\end{equation}

Using the commutation lemma (Lemma~\ref{lemma:commutation}), 

\begin{equation}
    \mathfrak{S}_{N,\rightarrow} \rightmerge \twobytwounit{(N,y)}= \left(\left(\left(\widetilde{\mathfrak{S}}\rightmerge \twobytwounit{(N,y)}\right) \rightmerge \twobytwounit{(N\shortminus 2, y\shortminus 1)}\right) \cdots \right) \rightmerge \twobytwounit{(2, y\shortminus 1)}.\label{eq:1101_2}
\end{equation}

Moreover, again using the commutation lemma,

\begin{equation}
    \widetilde{\mathfrak{S}}\rightmerge \twobytwounit{(N,y)} = 
    \left(\left(\left(\mym{\centeredTikZ{\squaresempty{2}{3}{1}{0};\squaresempty{1}{2}{0}{1};\squaresempty{1}{2}{3}{0};}_{(N,y\shortminus 1)}}\rightmerge \twobytwounit{(N,y)}\right) \rightmerge \twobytwounit{(N\shortminus 3, y)}\right) \cdots \right)\rightmerge \twobytwounit{(2, y)}.\label{eq:1101_1}
\end{equation}

Moreover, note that
\begin{equation}
\begin{aligned}
\centeredTikZ{\squaresdotted{4}{3}{0}{0};\squaresempty{2}{3}{1}{0};\squaresempty{1}{2}{0}{1};\squaresempty{1}{2}{3}{0};}
&=
\centeredTikZ{\squaresdotted{4}{3}{0}{0};\squaresempty{2}{3}{1}{0};\squaresempty{1}{2}{0}{1};}
\rightmerge
\centeredTikZ{\squaresdotted{4}{3}{0}{0};\squaresempty{2}{2}{2}{0};}\\
&= 
\left(
\centeredTikZ{\squaresdotted{4}{3}{0}{0};\squaresempty{2}{3}{1}{0};}
\rightmerge
\centeredTikZ{\squaresdotted{4}{3}{0}{0};\squaresempty{2}{2}{0}{1};}
\right)
\rightmerge
\centeredTikZ{\squaresdotted{4}{3}{0}{0};\squaresempty{2}{2}{2}{0};}\\
&= 
\left(
\centeredTikZ{\squaresdotted{4}{3}{0}{0};\squaresempty{2}{3}{1}{0};}
\rightmerge
\centeredTikZ{\squaresdotted{4}{3}{0}{0};\squaresempty{2}{2}{2}{0};}
\right)
\rightmerge
\centeredTikZ{\squaresdotted{4}{3}{0}{0};\squaresempty{2}{2}{0}{1};}\\
&=
\centeredTikZ{\squaresdotted{4}{3}{0}{0};\squaresempty{2}{3}{1}{0};\squaresempty{1}{2}{3}{0};}
\rightmerge
\centeredTikZ{\squaresdotted{4}{3}{0}{0};\squaresempty{2}{2}{0}{1};}.
\end{aligned}
\label{eq:1025}
\end{equation}
From Eq.~\eqref{eq:1025}, we can deduce that
\begin{equation}
\begin{aligned}
\centeredTikZ{\squaresdotted{4}{3}{0}{0};\squaresempty{2}{3}{1}{0};\squaresempty{1}{2}{0}{1};\squaresempty{1}{2}{3}{0};} 
\rightmerge \centeredTikZ{\squaresdotted{4}{3}{0}{0};\squaresempty{2}{2}{2}{1};}
&=  
\left(
\centeredTikZ{\squaresdotted{4}{3}{0}{0};\squaresempty{2}{3}{1}{0};\squaresempty{1}{2}{3}{0};} 
\rightmerge \centeredTikZ{\squaresdotted{4}{3}{0}{0};\squaresempty{2}{2}{2}{1};}
\right)
\rightmerge
\centeredTikZ{\squaresdotted{4}{3}{0}{0};\squaresempty{2}{2}{0}{1};}
\\
&=\centeredTikZ{\squaresdotted{4}{3}{0}{0};\squaresempty{3}{3}{1}{0};}\rightmerge \centeredTikZ{\squaresdotted{4}{3}{0}{0};\squaresempty{2}{2}{0}{1};}\\
&=\centeredTikZ{\squaresempty{4}{3}{0}{0};\squaresdotted{1}{1}{0}{0};},
\end{aligned}
\end{equation}
using the commutation lemma (Lemma~\ref{lemma:commutation}) in the first line, applying the \hyperref[constraints:cell1]{Type I-d) cluster constraint} to the first condition in Lemma~\ref{lemma:fundamental} in the second line, and using Proposition~\ref{prop:336_1} in the last line.

Moreover, using the definition of $\centeredTikZmini{\squaresempty{3}{3}{1}{0};\squaresempty{1}{2}{0}{1};}$ and then applying the \hyperref[constraints:cell2]{Type II-b) cluster constraint} and the \hyperref[constraints:cell1]{Type I-a) cluster constraint} to the first condition of Lemma~\ref{lemma:fundamental}, we obtain

\begin{equation}
    \centeredTikZ{\squaresempty{4}{3}{0}{0};\squaresdotted{1}{1}{0}{0};}=
    \left(\left(\left(
    \centeredTikZ{\squaresdotted{4}{3}{0}{0};\squaresempty{2}{2}{2}{1};}
    \rightmerge
    \centeredTikZ{\squaresdotted{4}{3}{0}{0};\squaresempty{2}{2}{1}{1};}
    \right)
    \rightmerge
    \centeredTikZ{\squaresdotted{4}{3}{0}{0};\squaresempty{2}{2}{0}{1};}\right)
    \rightmerge
    \centeredTikZ{\squaresdotted{4}{3}{0}{0};\squaresempty{2}{2}{2}{0};}
    \right) \rightmerge
    \centeredTikZ{\squaresdotted{4}{3}{0}{0};\squaresempty{2}{2}{1}{0};}.\label{eq:1101}
\end{equation}

Plugging in Eq.~\eqref{eq:1101} to Eq.~\eqref{eq:1101_1}, we obtain a level-2 snake\footnote{See the mutation lemma (Lemma~\ref{lemma:snake_mutation}).} which is right-merged by two $2\times 2$ clusters anchored at a $y$-coordinate of $y-1$. We can plug in this expression into Eq.~\eqref{eq:1101_2}, which is equivalent to the left hand side of Eq.~\eqref{eq:1106}. The obtained expression, by the mutation lemma (Lemma~\ref{lemma:snake_mutation}), establishes our main claim.
\end{proof}

\section{Ground state example}
\label{appendix:examples}
In this section, we discuss a physical ground state that satisfies Eqs.~\eqref{eq:ent_assumption1} and~\eqref{eq:ent_assumption2}. Note that these equations are satisfied trivially for product states, \emph{e.g.,} states in the following tensor product form:
\begin{equation}
    \rho = \bigotimes_{i} \rho_i,
    \label{eq:product_state}
\end{equation}
where the tensor product is taken over a set of sites in a lattice. Therefore, we will focus on states which are in some sense ``far'' from the states of such form.

The state we describe below are \emph{topologically ordered} in a sense that there is an obstruction to preparing the ground state from a product state by applying a ``short'' quantum circuit. If one were to prepare a topologically ordered ground state from the product state, one would need to apply a quantum circuit whose depth grows \emph{extensively} with the system size. Assuming that the allowed gates are geometrically local, \emph{i.e.,} supported on a ball of bounded radius, there is a $\Omega(L)$ lower bound for a system of size $L\times L$~\cite{Bravyi2006}. This lower bounds suggest that topologically ordered ground state cannot be well-approximated by a product state. In contrast, the ansatz proposed in this paper will provide an exact description of such states in the thermodynamic limit using a density matrix on a bounded region.

\subsection{Toric Code}
Toric code is a paradigmatic model of topological order~\cite{Kitaev2002}. this is a model defined on a square lattice, with the physical degrees of freedom living on the edges of the lattice. The Hamiltonian is
\begin{equation}
    H = -\sum_{s} A_s  -\sum_p B_p,
\end{equation}
where $s$ is a site and $p$ is a plaquette. Here $A_s$ is a tensor product of Pauli-$X$ operators acting on the set of edges coincident on $s$ and $B_p$ is a tensor product of Pauli-$Z$ operators acting on the set of edges that surround the plauqette $p$. Every term commutes, and as such, any ground state $|\psi\rangle$ must satisfy $A_s|\psi\rangle = B_p|\psi\rangle = |\psi\rangle$ for all $s$ and $p$.

For our discussion, it will be convenient to consider the ``rotated'' version of the toric code in which the qubits live on the sites of the square lattice, and the stabilizers are alternating sequence of tensor product of Pauli-$X$s and $Z$s along each plaquette; see Fig.~\ref{fig:rtc}. Embedding this lattice into a triangular lattice is simple. One can simply ``shear'' the lattice; see Fig.~\ref{fig:sheared_rtc}. Now, the sites of this sheared lattice can be embedded straightforwardly into a triangular lattice.
\begin{figure}[h]
    \centering
    \includegraphics[width=0.25\columnwidth]{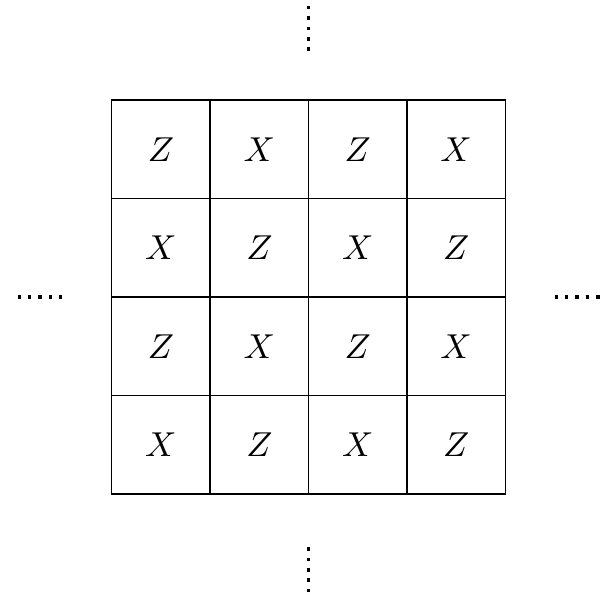}
    \caption{Rotated version of the toric code. Here $X$s ($Z$s) represent the stabilizers, each of which are tensor products of Pauli-$X$s (Pauli-$Z$s) over the qubits on the sites of the plaquette.}
    \label{fig:rtc}
\end{figure}

\begin{figure}[h]
    \centering
    \includegraphics[width=0.25\columnwidth]{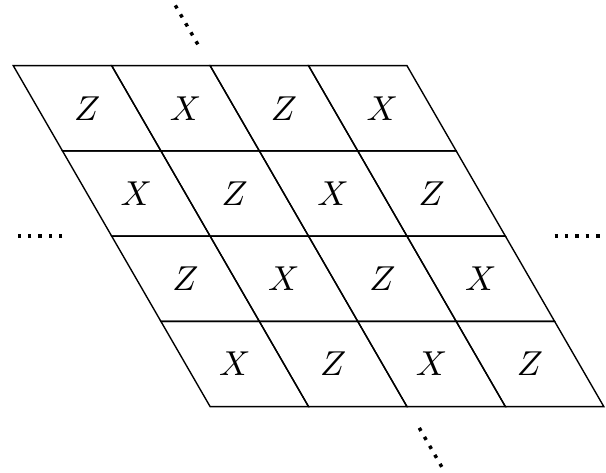}
    \caption{``Shearing'' the lattice of the rotated toric code. Here $X$s ($Z$s) represent the stabilizers, each of which are tensor products of Pauli-$X$s (Pauli-$Z$s) over the qubits on the sites of the plaquette.}
    \label{fig:sheared_rtc}
\end{figure}

Toric code is an example of \emph{stabilizer code}, for which a simple machinery has been already developed to study its ground state properties~\cite{Fattal2004,Hamma2005}. Given a subsystem, say $A$, we can compute the von Neumann entropy of the ground state reduced density matrix of $A$ as follows. First, note that $\langle A_s, B_p\rangle$ generates an Abelian group, which we denote as $S_{\text{TC}}$. Consider a subgroup $S_{\text{TC}, A} \subseteq S_{\text{TC}}$ which is generated by the elements of $S_{\text{TC}}$ that are supported on $A$ exclusively. The ground state entanglement entropy is~\cite{Fattal2004}
\begin{equation}
    S(A) = |A| \log 2 - \log |S_{\text{TC}, A}|. \label{eq:ee_stabilizer}
\end{equation}

\subsection{Verifying local constraints}
We can use Eq.~\eqref{eq:ee_stabilizer} to verify Eqs.~\eqref{eq:ent_assumption1} and~\eqref{eq:ent_assumption2}. We shall use the embedding described in Fig.~\ref{fig:sheared_rtc}.

Within this convention, we can first verify Eqs.~\eqref{eq:ent_assumption1} as follows. We restate this equation for the readers' convience below:
\begin{equation}
\begin{aligned}
\entropy{\centeredTikZ{\squaresfilled{2}{2}{0}{0};\squaresempty{2}{1}{0}{0};\squaresempty{1}{1}{0}{1};}} &= 
\entropy{\centeredTikZ{\squaresfilled{2}{2}{0}{0};\squaresempty{2}{1}{0}{0};\squaresfilled{1}{1}{0}{1};}}
+
\entropy{\centeredTikZ{\squaresfilled{2}{2}{0}{0};\squaresempty{1}{2}{0}{0};\squaresfilled{1}{1}{1}{0};}}
-
\entropy{\centeredTikZ{\squaresfilled{2}{2}{0}{0};\squaresfilled{2}{1}{0}{0};\squaresfilled{1}{1}{0}{1};\squaresempty{1}{1}{0}{0};}},\\[8pt]
\entropy{\centeredTikZ{\squaresfilled{2}{2}{0}{0};\squaresempty{2}{1}{0}{1};\squaresempty{1}{1}{1}{0};}} &= 
\entropy{\centeredTikZ{\squaresfilled{2}{2}{0}{0};\squaresempty{2}{1}{0}{1};\squaresfilled{1}{1}{1}{0};}}
+
\entropy{\centeredTikZ{\squaresfilled{2}{2}{0}{0};\squaresempty{1}{2}{1}{0};\squaresfilled{1}{1}{0}{1};}}
-
\entropy{\centeredTikZ{\squaresfilled{2}{2}{0}{0};\squaresfilled{2}{1}{0}{1};\squaresfilled{1}{1}{1}{0};\squaresempty{1}{1}{1}{1};}},
\end{aligned}
\tag{\ref{eq:ent_assumption1}}
\end{equation}
In this case, no subsystem appearing in the entropy calculation has a nontrivial stabilizer subgroup. This can be seen easily from the fact that $S_{\text{TC}}$ does not contain any element with a weight of less than $4$. Therefore,
\begin{equation}
\begin{aligned}
\entropy{\centeredTikZ{\squaresfilled{2}{2}{0}{0};\squaresempty{2}{1}{0}{0};\squaresempty{1}{1}{0}{1};}} &=  3,\\[8pt]
\entropy{\centeredTikZ{\squaresfilled{2}{2}{0}{0};\squaresempty{2}{1}{0}{1};\squaresempty{1}{1}{1}{0};}} &= 3.
\end{aligned}
\label{eq:tc_ent_1}
\end{equation}
Moreover,
\begin{equation}
\begin{aligned}
&\entropy{\centeredTikZ{\squaresfilled{2}{2}{0}{0};\squaresempty{2}{1}{0}{0};\squaresfilled{1}{1}{0}{1};}}=
\entropy{\centeredTikZ{\squaresfilled{2}{2}{0}{0};\squaresempty{1}{2}{0}{0};\squaresfilled{1}{1}{1}{0};}}
= 
\entropy{\centeredTikZ{\squaresfilled{2}{2}{0}{0};\squaresempty{2}{1}{0}{1};\squaresfilled{1}{1}{1}{0};}}
=
\entropy{\centeredTikZ{\squaresfilled{2}{2}{0}{0};\squaresempty{1}{2}{1}{0};\squaresfilled{1}{1}{0}{1};}}=2, \\[8pt]
&\entropy{\centeredTikZ{\squaresfilled{2}{2}{0}{0};\squaresfilled{2}{1}{0}{0};\squaresfilled{1}{1}{0}{1};\squaresempty{1}{1}{0}{0};}}=
\entropy{\centeredTikZ{\squaresfilled{2}{2}{0}{0};\squaresfilled{2}{1}{0}{1};\squaresfilled{1}{1}{1}{0};\squaresempty{1}{1}{1}{1};}}=1,
\end{aligned}
\label{eq:tc_ent_2}
\end{equation}
Therefore, Eqs.~\eqref{eq:ent_assumption1} are satisfied.

Verifying Eq.~\eqref{eq:ent_assumption2} is also straightforward. Let us restate the equation below.
\begin{equation}
\begin{aligned}
    \entropy{\centeredTikZ{\squaresempty{3}{3}{0}{0};}}
    =
    &\entropy{\centeredTikZ{\squaresfilled{3}{3}{0}{0};\squaresempty{2}{2}{0}{0};}}
    +
    \entropy{\centeredTikZ{\squaresfilled{3}{3}{0}{0};\squaresempty{2}{2}{1}{0};}} \\
    &+
    \entropy{\centeredTikZ{\squaresfilled{3}{3}{0}{0};\squaresempty{2}{2}{0}{1};}}
    +
    \entropy{\centeredTikZ{\squaresfilled{3}{3}{0}{0};\squaresempty{2}{2}{1}{1};}} \\
    &-
    \entropy{\centeredTikZ{\squaresfilled{3}{3}{0}{0};\squaresempty{2}{1}{0}{1};}}
    -
    \entropy{\centeredTikZ{\squaresfilled{3}{3}{0}{0};\squaresempty{2}{1}{1}{1};}}\\
    &-
    \entropy{\centeredTikZ{\squaresfilled{3}{3}{0}{0};\squaresempty{1}{2}{1}{0};}}
    -
    \entropy{\centeredTikZ{\squaresfilled{3}{3}{0}{0};\squaresempty{1}{2}{1}{1};}}
     \\
     &+ 
    \entropy{\centeredTikZ{\squaresfilled{3}{3}{0}{0};\squaresempty{1}{1}{1}{1};}},
\end{aligned}
\tag{\ref{eq:ent_assumption2}}
\end{equation}
Here the subsystems involved in the entropy calculation can have nontrivial stabilizer subgroups. For instance, the stabilizer group of the subsystem appearing in the left hand side of Eq.~\eqref{eq:ent_assumption2} is generated by four ``plaquette operators'' in our new convention, two of which are tensor product of Pauli $X$s and the rest being tensor product of Pauli $Z$s. Therefore, we have
\begin{equation}
\begin{aligned}
    \entropy{\centeredTikZ{\squaresempty{3}{3}{0}{0};}} &= 9-4 \\[8pt]
    &=5.
\end{aligned}
\end{equation}
Similarly, the first four terms appearing on the right hand side of Eq.~\eqref{eq:ent_assumption2} have a correction of $-1$, owing to the fact that each subsystem can support exactly one plaquette operator. Therefore, 
\begin{equation}
    \begin{aligned}
    &\entropy{\centeredTikZ{\squaresfilled{3}{3}{0}{0};\squaresempty{2}{2}{0}{0};}}
    =
    \entropy{\centeredTikZ{\squaresfilled{3}{3}{0}{0};\squaresempty{2}{2}{1}{0};}} =
    \entropy{\centeredTikZ{\squaresfilled{3}{3}{0}{0};\squaresempty{2}{2}{0}{1};}}=
    \entropy{\centeredTikZ{\squaresfilled{3}{3}{0}{0};\squaresempty{2}{2}{1}{1};}}\\&= 4-1\\
    &=3.
    \end{aligned}
\end{equation}
Using translational invariance and Eqs.~\eqref{eq:tc_ent_1} and~\eqref{eq:tc_ent_2}, one can verify Eq.~\eqref{eq:ent_assumption2}. Thus, we have verified both Eqs.~\eqref{eq:ent_assumption1} and~\eqref{eq:ent_assumption2}.

\subsection{Finite temperature}
\label{sec:toric_finite}
The entropy of the subsystem at finite temperature was calculated by Castelnovo and Chamon~\cite{Castelnovo2007}. (For the interested readers, the relevant equation is Eq.~(39) in Ref.~\cite{Castelnovo2007}.) Their formula, in the limit the size of the subsystem goes to infinity, reads:
\begin{equation*}
     S(A) =  \alpha |A|+ (|\partial A| - m_A) \log 2, \label{eq:cc_formula}
\end{equation*}
where $\alpha$ is a constant that depends on the temperature, $m_A$ is the number connected components, $|A|$ is the number of $X$-type stabilizers strictly supported in region $A$, and $|\partial A|$ is the number of $X$-type supported nontrivially on both $A$ and its complement.

Since Eq.~\eqref{eq:cc_formula} holds only in the $|A|\to \infty$ limit, the subsystems chosen in Eqs.~\eqref{eq:ent_assumption1} and~\eqref{eq:ent_assumption2} must consist of a large number of subclusters. While this discussion already appeared at the end of  Section~\ref{sec:summary}, let us reiterate the main point once more for the reader's convenience. Imagine replacing each cluster by a $2\times 2$ cluster, leading to
\begin{equation*}
\begin{aligned}
\centeredTikZ{\squaresempty{2}{2}{0}{0};}   &\longrightarrow
\centeredTikZ{\squaresempty{4}{4}{0}{0};}, \\
\centeredTikZ{\squaresempty{3}{3}{0}{0};}   &\longrightarrow
\centeredTikZ{\squaresempty{6}{6}{0}{0};}.
\end{aligned}
\end{equation*}
Now Eq.~\eqref{eq:ent_assumption1} and~\eqref{eq:ent_assumption2} reads:
\begin{equation*}
\begin{aligned}
\entropy{\centeredTikZ{\squaresfilled{4}{4}{0}{0};\squaresempty{4}{2}{0}{0};\squaresempty{2}{2}{0}{2};}}&=
\entropy{\centeredTikZ{\squaresfilled{4}{4}{0}{0};\squaresempty{4}{2}{0}{0};}}
+
\entropy{\centeredTikZ{\squaresfilled{4}{4}{0}{0};\squaresempty{2}{4}{0}{0};}}\\
&-
\entropy{\centeredTikZ{\squaresfilled{4}{4}{0}{0};\squaresempty{2}{2}{0}{0};}},\\
\entropy{\centeredTikZ{\squaresfilled{4}{4}{0}{0};\squaresempty{4}{2}{0}{2};\squaresempty{2}{2}{2}{0};}}&=
\entropy{\centeredTikZ{\squaresfilled{4}{4}{0}{0};\squaresempty{4}{2}{0}{2};}}
+
\entropy{\centeredTikZ{\squaresfilled{4}{4}{0}{0};\squaresempty{2}{4}{2}{0};}}\\
&-
\entropy{\centeredTikZ{\squaresfilled{4}{4}{0}{0};\squaresempty{2}{2}{2}{2};}},
\end{aligned}
\end{equation*}
\vspace{1.5cm}

and

\begin{equation*}
\begin{aligned}
\entropy{\centeredTikZ{\squaresempty{6}{6}{0}{0};}} &= 
\entropy{\centeredTikZ{\squaresfilled{6}{6}{0}{0};\squaresempty{4}{4}{0}{0};}}
+
\entropy{\centeredTikZ{\squaresfilled{6}{6}{0}{0};\squaresempty{4}{4}{2}{0};}}
-
\entropy{\centeredTikZ{\squaresfilled{6}{6}{0}{0};\squaresempty{4}{4}{0}{2};}}
-
\entropy{\centeredTikZ{\squaresfilled{6}{6}{0}{0};\squaresempty{4}{4}{2}{2};}} \\
&-
\entropy{\centeredTikZ{\squaresfilled{6}{6}{0}{0};\squaresempty{4}{2}{0}{2};}}
-
\entropy{\centeredTikZ{\squaresfilled{6}{6}{0}{0};\squaresempty{4}{2}{2}{2};}}
-
\entropy{\centeredTikZ{\squaresfilled{6}{6}{0}{0};\squaresempty{2}{4}{2}{0};}}
-
\entropy{\centeredTikZ{\squaresfilled{6}{6}{0}{0};\squaresempty{2}{4}{2}{2};}}\\
&+
\entropy{\centeredTikZ{\squaresfilled{6}{6}{0}{0};\squaresempty{2}{2}{2}{2};}}.
\end{aligned}
\end{equation*}

Similarly, one can replace the clusters in Eqs.~\eqref{eq:ent_assumption1} and~\eqref{eq:ent_assumption2} to  a $\ell \times \ell$ cluster for any $\ell>1$. These conditions can be satisfied up to an arbitrarily small error by increasing $\ell$.






\section{Efficient tensor network}
\label{appendix:tensor_network}
In this Appendix, we show that the maximum-entropy state in Eq.~\eqref{eq:global_state} is an efficient two-dimensional tensor network. The rest of this Appendix is structured as follows. First, we express the maximum entropy state as a quantum state created by a sequence of linear maps acting on bounded regions. Second, we will explain why the resulting state can be viewed as a tensor network. Lastly, we shall briefly explain in what sense our construction is efficient.

\subsection{Construction}

Let us begin by explaining how the maximum-entropy global state consistent with the marginals can be constructed. Using the convention in Fig.~\ref{fig:lattice}, we shall generate a sequence of reduced density matrices that add one coarse-grained site at a time, from the bottom to the top, and from left to right. Specifically, we shall argue that the maximum-entropy global state in Theorem~\ref{thm:entropy_maxmerge} is equal to 
\begin{equation}
\mathcal{E}_{M-1,\uparrow} \circ \ldots \circ \mathcal{E}_{1,\uparrow}\left[ \mathbb{S}_{1}\right], \label{eq:global_construction}
\end{equation}
where $\mathbb{S}_1$ is the level-$1$ snake at $y=1$; see Definition~\ref{def:snake2}. By the mutation lemma (Lemma~\ref{lemma:snake_mutation}), the level-$1$ snake can be created by a sequence of quantum channels from left to right. Also, $\mathcal{E}_{i,\uparrow}$ is a sequence of quantum channels acting from left to right. Therefore, Eq.~\eqref{eq:global_construction} is indeed a quantum state created by a sequence of localized quantum channels, in the order we have prescribed above. 

The proof of Eq.~\eqref{eq:global_construction} is based on induction. This is the key identity:
\begin{equation}
\mathcal{E}_{k-1,\uparrow} \circ \ldots \circ \mathcal{E}_{1,\uparrow}\left[ \mathbb{S}_{1}\right] = \mathbb{S}\left( \left(\adulthorizontalsnake{1}{y}{N}{y\shortplus 1} \right)_{y=1}^{k-1} \right)
\end{equation}
for $1<k<M+1$. The $k=2$ case follows from Proposition~\ref{prop:onetotwo}. The $k=3$ case follows a logic similar to the proof of Theorem~\ref{thm:forgotten}. To be more specific, in the proof of Theorem~\ref{thm:forgotten}, one concludes that the sixth condition of Lemma~\ref{lemma:fundamental} is satisfied, with the choice of $\lambda, \sigma,$ and $\Phi$ specified in Eq.~\eqref{eq:lambsigphi}. We have already established in Lemma~\ref{lemma:fundamental} that $\Phi(\lambda)$ is the maximum entropy state, which is equal to the merged state of the two level-$2$ snakes, one at $y=1$ and $2$ and the other at $y=2$ and $3$.

More generally,

\begin{equation}
    \begin{aligned}
    \mathcal{E}_{k-1,\uparrow} \circ \ldots \circ \mathcal{E}_{1,\uparrow}\left[ \mathbb{S}_{1}\right]&= \mathcal{E}_{k-1,\uparrow} \left[\mathbb{S}\left( \left(\adulthorizontalsnake{1}{y}{N}{y\shortplus 1} \right)_{y=1}^{k-2} \right) \right]\\
    &= \mathcal{E}_{k-1,\uparrow} \left[\mathbb{S}\left( \left(\adulthorizontalsnake{1}{y}{N}{y\shortplus 1} \right)_{y=k-2}^{1} \right) \right] \\
    &= \mathcal{E}_{k-1,\uparrow} \left[\adulthorizontalsnake{1}{k\shortminus 2}{N}{k\shortminus 1} \rightmerge \mathbb{S}\left( \left(\adulthorizontalsnake{1}{y}{N}{y\shortplus 1} \right)_{y=k-3}^{1}\right) \right] \\
    &= \mathcal{E}_{k-1,\uparrow} \left[\adulthorizontalsnake{1}{k\shortminus 2}{N}{k\shortminus 1}\right] \rightmerge \mathbb{S}\left( \left(\adulthorizontalsnake{1}{y}{N}{y\shortplus 1} \right)_{y=k-3}^{1}\right) \\
    &= \mathbb{S}\left( \left(\adulthorizontalsnake{1}{y}{N}{y\shortplus 1} \right)_{y=k-2}^{k-1}\right) \rightmerge \mathbb{S}\left( \left(\adulthorizontalsnake{1}{y}{N}{y\shortplus 1} \right)_{y=k-3}^{1}\right) \\
    &= \mathbb{S}\left( \left(\adulthorizontalsnake{1}{y}{N}{y\shortplus 1} \right)_{y=1}^{k-1}\right).
    \end{aligned}
\end{equation}

The first line is based on the induction hypothesis, specifically, the $k=2$ case. The second line uses the mutation lemma (Lemma~\ref{lemma:snake_mutation}). (Note that the order is reversed.) The third line uses the splitting lemma (Lemma~\ref{lemma:snake_splitting}). The fourth line is based on an observation that the map $\mathcal{E}_{k-1,\uparrow}$ acts trivially on the snake that appears after $\rightmerge$. The fifth line follows from the induction hypothesis, \emph{i.e.,} the $k=3$ case. The sixth line uses the splitting lemma (Lemma~\ref{lemma:snake_splitting}).

\subsection{Tensor network representation}
In the recast form of Eq.~\eqref{eq:global_construction}, the maximum-entropy global state can be viewed as a two-dimensional tensor network; see Ref.~\cite{Orus2014} for a review. To see why, note that Eq.~\eqref{eq:global_construction} essentially says that the global state can be created by a sequence of linear maps on bounded regions. One can view this linear map as a tensor whose legs are connected to the set of physical qubits that it acts on, with an additional set of legs that represent the newly introduced physical degrees of freedom. 

More concretely, we can formulate our argument as an induction. Suppose that 
\begin{equation}
    \mathcal{E}_{k,\uparrow} \circ \ldots \circ \mathcal{E}_{1,\uparrow}[\mathbb{S}_1]
\end{equation}
for $k < M$ is a tensor network on a $k\times L$ lattice with constant bond dimension, where $L$ is the number of degrees of freedom in $\mathbb{S}_2$. If we apply $\mathcal{E}_{k+1,\uparrow}$, it is easy to see that the resulting state is a tensor network on a $(k+1)\times L$ lattice; see Fig.~\ref{fig:tensor_network}. Moreover, for each site of the $k\times L$-sized tensor network, at most $\mathcal{O}(1)$ number of local recovery channels act. Therefore, the bond dimension of the tensor network grows at most by a constant amount. Therefore, at all steps, the bond dimension of the tensor network is bounded from above some constant. This argument establishes that the state in Eq.~\eqref{eq:global_state} has a constant bond-dimension two-dimensional tensor network representation.

\begin{figure*}[t]
    \centering
    \subfloat[][First step]{\includegraphics{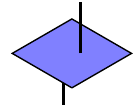}}
        \subfloat[][Second step]{\includegraphics{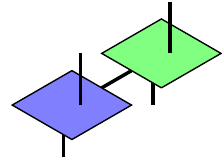}}
        \subfloat[][Third step]{\includegraphics{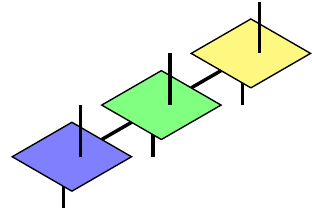}}
        \subfloat[][Completion of a single row]{\includegraphics{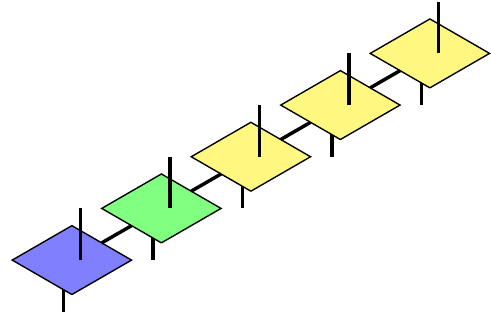}}
        \subfloat[][First step of the second row ]{\includegraphics{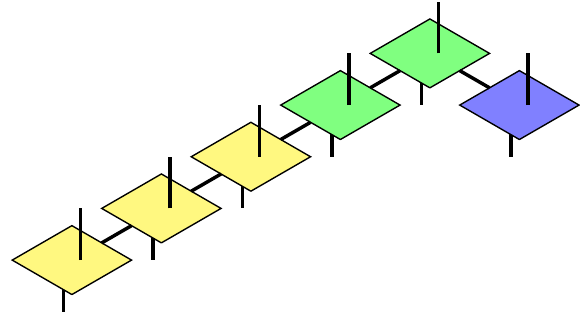}}\\
        \subfloat[][Second step of the second row]{\includegraphics{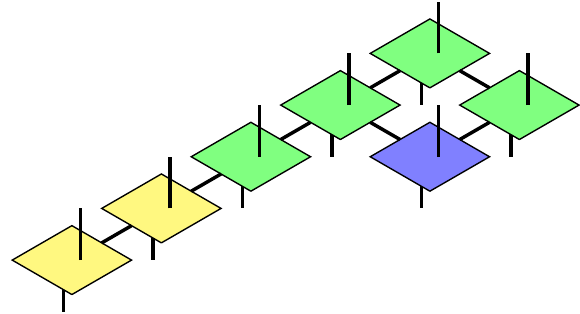}}
        \subfloat[][Repetition]{\includegraphics{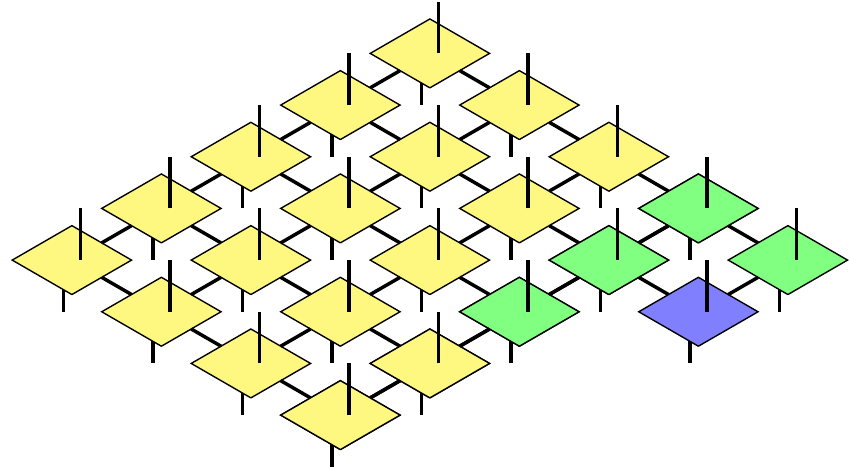}}
    \caption{Sequential creation of the tensor network that describes the maximum entropy state (Eq.~\eqref{eq:global_construction}). The outward-facing legs are the physical degrees of freedom, one representing the ``bra'' and the other representing the ``ket.'' The blue tensors are the newly introduced tensors and the the green tensors are the ones that are nontrivally affected by the local channel that creates the blue tensor. The yellow tensors remain unaffected.}
    \label{fig:tensor_network}
\end{figure*}

Having established that the maximum entropy state (Eq.~\eqref{eq:global_construction}) is a two-dimensional tensor network, let us move on to the next topic. In what sense is this tensor network efficient? In the literature, the word \emph{efficient} is used in two different ways. Sometimes, it means that the amount of data that defines the state grows polynomially with the system size. It should be easy to see that our tensor network is certainly efficient in that sense. 

A more stringent notion of efficiency concerns calculation of expectation values. The main question is, given a description of the tensor network, whether there is a polynomial-time algorithm to compute the expectation values of local observables. Let us focus on geometrically local observables to be more specific, \emph{i.e.,} observables that can be supported on a ball of constant radius. Note that any such observable (provided that its support is sufficiently small) can be computed directly from the reduced density matrices that define the tensor network, since we have already proved that the global state is consistent with the marginals. Therefore, our tensor network is efficient also in the sense that there is a polynomial-time (in fact, constant-time) algorithm to compute such expectation values. We note that two-dimensional tensor networks generally \emph{do not} posses such a property~\cite{Schuch2007} although they can be if the underlying tensors are ``isometric''~\cite{Zaletel2020}.

In fact, our result shows that there is another interesting sense in which our tensor network is efficient. We can compute not just the local expectation values, but also the global entropy! Recall that entropy of a state is a \emph{nonlinear} function of the state. As such, it is a highly nontrivial fact that one can compute the entropy polynomial time. 

To summarize, there is a sense in which the state we are studying, which does have a tensor network representation, is much more efficient than the garden-variety tensor networks. Not only is it possible to have a tensor network with constant bond dimension, but it is also possible to have an efficient scheme to estimate local expectation values, and even the global entropy. 

However, these properties are more manifest from the perspective of the local reduced density matrices. As such, we believe a better point of view is to think in terms of the marginals.

\bibliographystyle{myhamsplain2}
\bibliography{bib}
\end{document}